\newif\iftechreport
\techreporttrue

\iftechreport
\documentclass[sigconf,screen,nonacm]{acmart}
\else
\documentclass[sigconf,screen]{acmart}
\fi

\usepackage{paper}

\unless\iftechreport
\copyrightyear{2026}
\acmYear{2026}
\setcopyright{cc}
\setcctype{by}
\acmConference[CCS '26]{Proceedings of the 2026 ACM SIGSAC Conference on Computer and Communications Security}{November 15--19, 2026}{The Hague, Netherlands}
\acmBooktitle{Proceedings of the 2026 ACM SIGSAC Conference on Computer and Communications Security (CCS '26), November 15--19, 2026, The Hague, Netherlands}
\acmDOI{10.1145/3830454.3846718}
\acmISBN{979-8-4007-2871-6/2026/11}
\fi

\begin{document}


\iftechreport
\title{SCHERI: Provably Secure Speculation Under the Constant-Time Policy for CHERI (Extended Version)}
\else
\title{SCHERI: Provably Secure Speculation Under the Constant-Time Policy for CHERI}
\fi



\author{Shixin Song}
\email{shixins@mit.edu}
\orcid{0009-0007-5638-5164}
\affiliation{%
  \institution{Massachusetts Institute of Technology}
  \city{Cambridge}
  \country{United States}
}
\author{Davide Davoli}
\email{davide.davoli@mpi-sp.org}
\orcid{0009-0009-2981-2962}
\affiliation{%
  \institution{MPI-SP}
  \city{Bochum}
  \country{Germany}
}
\author{Elias Storme}
\email{elias.storme@kuleuven.be}
\orcid{0000-0002-1202-0688}
\affiliation{%
  \institution{DistriNet, KU Leuven}
  \city{Leuven}
  \country{Belgium}
}
\author{Marton Bognar}
\email{marton.bognar@kuleuven.be}
\orcid{0000-0002-8641-7549}
\affiliation{%
  \institution{DistriNet, KU Leuven}
  \city{Leuven}
  \country{Belgium}
}
\author{Dominique Devriese}
\email{dominique.devriese@kuleuven.be}
\orcid{0000-0002-3862-6856}
\affiliation{%
  \institution{DistriNet, KU Leuven}
  \city{Leuven}
  \country{Belgium}
}
\author{Frank Piessens}
\email{frank.piessens@kuleuven.be}
\orcid{0000-0001-5438-153X}
\affiliation{%
  \institution{DistriNet, KU Leuven}
  \city{Leuven}
  \country{Belgium}
}
\author{Tamara Rezk}
\email{tamara.rezk@inria.fr}
\orcid{0000-0003-3744-0248}
\affiliation{%
  \institution{Inria}
  \city{Sophia Antipolis}
  \country{France}
}

\renewcommand{\shortauthors}{Song et al.}

\begin{abstract} 
    Capability-based architectures such as \cheri provide strong support for the architectural isolation of software components.
To additionally protect against microarchitectural leakage, software can be written in a constant-time fashion.
Modern processors, however, rely heavily on speculative execution, which can invalidate the constant-time guarantees and leak isolated secrets transiently.

In this work, we show that providing secure speculation for \cheri is non-trivial, and that existing proposals fail to preserve the confidentiality guarantees.
We develop a formal framework for reasoning jointly about capability safety, speculative execution, and information-flow security, and use it to demonstrate potential leaks.
We then present \sys{}, a new processor design within this framework, and formally prove that it provides end-to-end secure speculation guarantees for the constant-time policy.

Our results provide formal foundations and practical guidance for building future capability-based processors, which are resilient to Spectre attacks for constant-time programs.

\end{abstract}




\maketitle

\section{Introduction}
\label{sec:intro}

The ongoing prevalence of memory-safety vulnerabilities has motivated  capability-based hardware architectures that deliver \emph{enforceable guarantees by construction}.
The \cheri (Capability Hardware Enhanced RISC Instructions) architecture emerged in 2010~\cite{cheriFirst}  as a  candidate for this vision, offering fine-grained memory protection and principled reasoning about pointer provenance. At the architectural level, \cheri enforces strong spatial memory-safety guarantees and enables robust compartmentalization~\cite{cheriTree}.
The \cheri architecture has gained momentum over the last decade with the development of CHERI-RISC-V~\cite{riscvcheri} and the ARM Morello project~\cite{morello} that implements \cheri in ARM.

However, \cheri memory safety assumes a non-speculative execution model. Modern high-performance processors rely heavily on speculation, and, since the discovery of transient-execution attacks~\cite{spectre}, it has become clear that speculation can violate architectural guarantees by exposing microarchitectural side effects. While misspeculated instructions are rolled back at the architectural level, their effects on microarchitectural state (e.g., caches) may persist, enabling attackers to leak secrets through side channels~\cite{guarnieri2021hardware}.

\paragraph{CSC} Prior work has demonstrated that \cheri is not immune to this problem~\cite{fuchsMasterThesis}. In particular, Fuchs~et~al.~\cite{fuchs2024csc}  introduced the \emph{Capability Speculation Contract} (CSC), which constrains speculative execution to preserve \cheri's capability invariants.
They showed that violations of these invariants in speculative implementations (e.g., CHERI-Toooba) lead to concrete attacks such as Meltdown-CF, which break CHERI's memory-safety guarantees in practice. CSC can be understood as an instance of the broader class of hardware-software contracts~\cite{guarnieri2021hardware} for secure speculation, which aim to provide principled co-design between hardware mechanisms and software reasoning.
CSC  restores
speculative sandboxing by ensuring that speculative memory accesses respect capability permissions and bounds. While effective
for preserving speculative sandboxing, CSC does not provide confidentiality guarantees when legitimately accessible secrets can still
leak through timing, cache, or other microarchitectural channels.
We show that achieving provably secure speculation for the constant-time policy in CHERI
requires moving beyond sandboxing, toward confidentiality guarantees.

\paragraph{BLACKOUT} An important step toward this goal was taken by \blackout \cite{blackout}, which extends \cheri with hardware-supported taint tracking via \emph{blinded capabilities} to enable data-oblivious computation.  \blackout raises faults when secret data influences control flow or memory accesses, thereby aiming to prevent both conventional and speculative side-channel leakage.
At a high level, \blackout moves beyond sandboxing toward non-interference-style guarantees by tracking the flow of secret data through computation. However, as we will show, \blackout does not fully achieve secure speculation for the constant-time policy.

These limitations highlight a deeper issue: existing approaches lack a unified formal framework capable of expressing and comparing memory safety, sandboxing, and secure speculation guarantees.

\paragraph{Our contributions.}
In this paper, we address this gap by developing a formal framework for reasoning about \cheri-like architectures under speculative execution and strong leakage models.
Within this framework, we:
\begin{itemize}
  [leftmargin=*]
    \item define a speculative semantics for \cheri-like processors with a strong attacker model and a notion of speculative constant time;

    \item show that the CSC does not guarantee confidentiality under leakage models stronger than speculative sandboxing;

    \item analyze \blackout under these stronger leakage models,
    identifying sources of speculative information leakage;

    \item propose \sys{}, eliminating leaks while preserving practicality;

    \item formally prove that \sys{} satisfies end-to-end secure speculation guarantees in our model.
\end{itemize}

Our results provide both formal foundations and practical guidance for building \cheri processors resilient to speculative side-channel attacks.
\iftechreport
This version of the paper is extended from the conference version~\cite{song2026scheri}, containing the full hardware semantics (\Cref{sec:apprules}) and proofs (\Cref{sec:proofs}).
\else
An extended version of this paper with full hardware semantics and proofs is available~\cite{scheri-proof}.
\fi

\section{Background}
\label{sec:background}

\paragraph{\cheri}\label{sec:cheri}
\cheri{} is a capability architecture that extends conventional ISAs with hardware-enforced memory protection. Instead of representing pointers as plain integers, \cheri{} uses \emph{capabilities}: unforgeable references that carry both an address and metadata such as bounds, permissions, and validity information. Memory accesses through capabilities are checked in hardware, preventing out-of-bounds dereferences and unauthorized use of pointers. Capabilities can only be derived from existing capabilities, ensuring monotonicity: derived capabilities cannot gain permissions or exceed the bounds of their parent capability except at controlled security boundary crossings. These mechanisms provide strong spatial memory safety and support fine-grained compartmentalization. Practical realizations include CHERI-RISC-V~\cite{riscvcheri}, and ARM Morello~\cite{morello}.
While \cheri{} substantially strengthens architectural memory safety, these guarantees are defined at the ISA level and therefore assume architecturally correct execution. They do not automatically extend to (speculative) microarchitectural behavior.

\paragraph{Speculative execution and attacks}
Modern processors improve performance through out-of-order and speculative execution. Instructions may execute before prior branches, permissions checks, or data dependencies are fully resolved. If speculation is later found incorrect, architectural state is rolled back, but transient effects on microarchitectural state, such as cache contents, predictor state, or execution timing, may remain observable.

Transient-execution attacks exploit this behavior to leak secrets via persistent microarchitectural effects.
Spectre-style attacks~\cite{spectre} mistrain predictors to transiently execute attacker-chosen instruction sequences, Meltdown-style attacks~\cite{meltdown} exploit incorrect forwarding from faulting or unauthorized operations. These attacks demonstrate that architectural safety alone is insufficient when secrets can influence transient execution.

For \cheri{}, this means that even capability-safe software may leak information if speculative execution bypasses intended checks or transiently exposes secret-dependent behavior.

\paragraph{Security properties: sandboxing and constant-time}

Mitigations for speculative execution attacks and their desired security properties can be described using a {\em hardware-software contract}~\cite{guarnieri2021hardware}.
It is useful to distinguish two types of security properties:
\begin{itemize}
  [leftmargin=*]
  \item  {\em Sandboxing} or {\em compartmentalization}, ensuring that a potentially malicious program cannot speculatively read outside of its architecturally designated sandbox, and
  \item {\em Constant-time } or {\em data-obliviousness}, ensuring that a benign program does not leak secrets through microarchitectural channels.
\end{itemize}
These are very different properties: sandboxing considers code that does not architecturally compute on secrets but tries to get access speculatively, whereas constant-time considers code that does architecturally compute on secrets, and should not leak them through side channels, even speculatively. Hence, they typically require different mitigations under speculation.
CSC aims to prevent leakage in the sandboxing model, while BLACKOUT targets the speculative constant-time property.

\section{Threat Model}
\label{sec:threatmodel}

We consider an attacker with the goal of inferring secret information from a \cheri{} system through microarchitectural side-channel observations.
The attacker cannot directly read privileged state, break cryptographic primitives, or violate the architectural memory protection guarantees of \cheri{}, but they can execute untrusted code on the same processor (concurrently or across time-sharing contexts) and trigger victim execution with chosen public inputs.
As is standard in the Spectre literature, the attacker may mistrain branch predictors, influence microarchitectural state to cause contention, and observe timing differences caused by victim execution.

We consider a broad class of microarchitectural side channels, including the state of microarchitectural structures (e.g., caches, branch predictors, the reorder buffer), speculation success and rollback behavior, and contention for execution ports and other timing-visible shared resources.
This follows prior secure speculation work that treats attacker-observable microarchitectural state abstractly rather than committing to one specific channel~\cite{prospect}.
Our threat model is intentionally stronger than the one captured by the Capability Speculation Contract (CSC): in addition to unauthorized speculative memory accesses, we consider leakage of legitimately accessed secrets through timing-visible microarchitectural effects.
We do \emph{not} consider physical side channels such as power analysis, electromagnetic emanations, or fault injection.

Our goal is therefore to enforce relative speculative constant time for victim code (formalized in \Cref{thm:sct}): secrets that are not leaked architecturally do not influence the attacker's observations.
In particular, we seek guarantees stronger than CSC-style speculative sandboxing: speculative execution should preserve confidentiality, not only prevent capability forgery.

\section{Prior Approaches and Shortcomings}
\label{sec:motivating-examples}

The security property we aim to achieve is that programs that are constant-time in the absence of speculation remain constant-time under speculation.
In this section, we show that existing Spectre mitigations on the \cheri architecture do not achieve this goal. We discuss two recently published approaches: \emph{Capability Speculation Contracts} (CSC)~\cite{fuchs2024csc} and \blackout~\cite{blackout}.

\subsection{Capability Speculation Contract}
\label{sec:csc}
\cheri provides architectural support for software sandboxing~\cite{cheriFirst},
and recently Fuchs et al.~\cite{fuchs2024csc} proposed an approach to
preserving these architectural sandboxing properties under speculation.
Their \emph{Capability Speculation Contract} (CSC) expresses that the
hardware should not perform speculative violations of security
boundaries enforced by capabilities. More precisely, it states:

\begin{contract}[Capability Speculation Contract~\cite{fuchs2024csc}]
  All instruction and data-memory accesses issued in speculation must be authorized by capabilities either (1) in the committed register file; or (2) in memory transitively reachable through (1).
\end{contract}

Hence, CSC has a different goal than our design: it does not aim to prevent side-channel leaks of accessible secrets through speculation. It is straightforward to construct examples of programs that are secure in the absence of speculation and that comply with the CSC contract, but that leak secrets through side channels when speculation is allowed.
Consider the following code:
\begin{minted}{asm}
f_load_sec:
  cld a1, 0(csp) # secret data is loaded into a1
  ccall g # after executing g, speculation starts
...
exploit_gadget: # call to g speculatively returns here
  cinoffset cs0, csp, a1 # cs0's offset advances by secret a1
  cld a1, 0(cs0) # secret is leaked
\end{minted}
Here, \texttt{csp} is the stack capability register,
\texttt{cs0} is another capability register,
and \texttt{a1} is a general purpose register.
If secret data is stored on the stack referenced by \texttt{csp},
the \texttt{cld} instruction at line 2 loads it in \texttt{a1}.
If the called \texttt{g} function speculatively returns to \texttt{exploit\_gadget},
the secret is used to set the offset of \texttt{cs0}.
Then, the load at line 7 uses the address of \texttt{cs0}, leaking the secret.
Note that the secret could also be leaked through other side channels.

\subsection{\blackout}
\label{sec:blackout}
\blackout~\cite{blackout} extends \cheri with {\em blinded capabilities} for enforcing data-oblivious execution, also during speculation.
\blackout incorporates the ideas of oblivious instruction set architectures~\cite{OISA,blime} into the \cheri architecture by repurposing an unused bit in capability metadata to mark a capability as {\em blinded}.
When data is loaded through a blinded capability into a register, the register is also marked as blinded, and the hardware implements taint-tracking to propagate blindedness.
Blinded values are treated as secrets, and the hardware prevents them from affecting timing behavior by faulting on e.g., load/store instructions with blinded addresses or branches with blinded conditions.
During speculation, the faults are not raised architecturally, but they still prevent the leakage of blinded data.
Building on these primitives, \blackout allows programmers to write data-oblivious code,
relying on the hardware to prevent side-channel leakage.

The security argument of \blackout relies on the maintenance of five key invariants, but in order to efficiently support the call stack and register spills, \blackout allows some of these invariants to be violated in a controlled way.
One invariant (I3~\cite{blackout}) says that the bounds of valid blinded and non-blinded capabilities must not simultaneously overlap. However, this complicates storing both blinded and non-blinded data on the stack. Hence, \blackout allows this invariant to be violated for the stack capability: a local blinded variable will architecturally be accessed through a blinded capability that overlaps with the non-blinded stack capability.
This is not a problem in the absence of speculation if the software stack can enforce that blinded data is only accessed through the blinded capability.
However, this opens the door to speculative execution attacks that can load and leak blinded data through the non-blinded stack capability.
Another invariant (I1) says that blinded data cannot be stored into memory using non-blinded capabilities. This is a problem for compiler-induced register spills: blinded registers cannot be spilled to the stack using the non-blinded stack capability. To address this, \blackout repurposes the capability validity tag to indicate blindedness of secret register spills.
Specifically, when the program spills a blinded register to the stack (using a \texttt{csc} instruction that stores the whole 16 bytes in the register), the \blackout processor emits a special structure named \emph{Blinded Register Record (BRR)}, which consists of the 8-byte spilled blinded value and an 8-byte marker (see \Cref{fig:background-cap}).
\blackout also sets the capability validity tag corresponding to the spill slot to $1$.
When restoring the blinded register spill from the stack using the \texttt{clc} instruction, the validity tag and the marker indicate that the restored value is a BRR rather than a capability or other public data on the stack, so the processor can properly mark the loaded data as blinded.

\subsubsection{Security Analysis}

At its core, \blackout maintains its security guarantees as long as the hardware correctly tracks when secret data is accessed.
Data can be marked as secret in two ways: either by setting the blinded bit in its associated capability, or, in the case of spilled registers on the stack, with a Blinded Register Record (BRR), which stores a magic constant right next to the spilled register value in memory.
As long as the data is accessed through the blinded capability (in the case of blinded memory) or through a capability load through the stack capability (in the case of BRRs), the processor correctly sets and propagates the taint bit.
In the following, we systematically explore how blinded data could be accessed in other ways that deviate from the above, leading to the loaded secret data incorrectly not being tainted, enabling leakage.

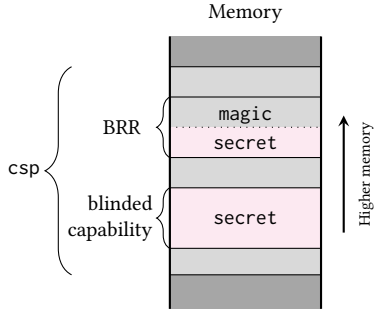
\begin{figure}[t]
  \centering
\newcommand{\cellwidth}{2}
\newcommand{\memheight}{3.6}
\newcommand{\margin}{.3}
\newcommand{\BRRheight}{.8}
\newcommand{\blindedheight}{.8}
\DeclareRobustCommand{\legendsquare}[1]{%
  \textcolor{#1}{\rule{1ex}{1ex}}%
}
\begin{adjustbox}{max width=0.85\linewidth}
\begin{tikzpicture}
\tikzstyle{every node}=[font=\small]
  \node at (.5*\cellwidth,\memheight+\margin) {Memory};
  \fill[gray!70] (0,0) rectangle ++(\cellwidth, \memheight);
  \coordinate (topstack) at (0, 8/9*\memheight);
  \coordinate (basestack) at (0, 1/8*\memheight);
  \draw[fill=gray!30] (topstack) rectangle ($ (basestack) + (\cellwidth,0) $);
  \coordinate (brr) at ($ (topstack) - (0,1.2) $);
  \fill[magenta!10] (brr) rectangle ++(\cellwidth,.5*\BRRheight);
  \draw (brr) rectangle ++(\cellwidth,\BRRheight);
  \draw[dotted] ($ (brr) + (0,.5*\BRRheight) $) -- ++(\cellwidth,0);
  \node at ($ (brr) + (.5*\cellwidth,.70*\BRRheight) $) {\texttt{magic}};
  \node at ($ (brr) + (.5*\cellwidth,.25*\BRRheight) $) {\texttt{secret}};
  \draw [decorate,decoration={brace,amplitude=5pt}]
  (brr) -- ++(0,\BRRheight) node[align=right,midway,xshift=-2em]{BRR};
  \coordinate (blinded) at ($ (topstack) - (0,2.4) $);
  \draw[fill=magenta!10] (blinded) rectangle ++(\cellwidth,\blindedheight);
  \node at ($ (blinded) + (.5*\cellwidth,.5*\blindedheight) $) {\texttt{secret}};
  \draw[thick,->,>=stealth] (\cellwidth+\margin,.28*\memheight)  -- (\cellwidth+\margin,.71*\memheight) node[pos=.5, xshift=.3cm, rotate=90] {{\scriptsize Higher memory}};
  \draw [decorate,decoration={brace,mirror,amplitude=10pt}]
  ($ (topstack) + (-1.3,0) $) -- ($ (basestack) + (-1.3,0) $) node[midway,xshift=-2em]{\texttt{csp}};
  \draw [decorate,decoration={brace,amplitude=5pt}]
  (blinded) -- ++(0,\blindedheight) node[align=right,midway,xshift=-2.5em]{blinded\\capability};
  \draw[thick] (0,0) -- ++(0,\memheight);
  \draw[thick] (\cellwidth,0) -- ++(0,\memheight);
\end{tikzpicture}
\end{adjustbox}
    \caption{In \blackout, memory referenced by blinded capabilities and Blinded Register Records (BRR) on the stack are considered secret (areas colored in \textcolor{magenta!10}{\rule{2ex}{2ex}}). However, the non-blinded stack capability (\texttt{csp}) can also access this secret data.}%
\label{fig:background-cap}
\end{figure}

\begin{enumerate}
[leftmargin=*]
  \item \label{issue:brr-all}
  Improperly accessing a BRR.
  \begin{enumerate}
  [leftmargin=*]
  \item \label{issue:brr}
    The blinded bit is only set when BRRs are loaded via full capability load instructions (\texttt{clc}). Accessing a spilled value, potentially speculatively, through a narrower load instruction (e.g., \texttt{cld}) via the stack pointer (or any non-blinded capability) will not set the blinded bit.
  \item \label{issue:magic}
    The blinded bit is only set when the BRR is tagged as a capability and its metadata corresponds to the magic constant. Even if the BRR is accessed via the stack pointer using a capability load, the blinded bit will not be set if the magic value has been corrupted or the capability tag cleared.
    This corruption could also happen transiently, e.g., as a result of incorrect store-to-load forwarding.
  \end{enumerate}
  \item \label{issue:overlap}
  Improperly accessing secret data through a non-blinded capability. In \blackout, this is always possible for function-local blinded data, as the stack capability is not blinded (cf. \Cref{fig:background-cap}).
  \begin{enumerate}
  [leftmargin=*]
    \item \label{issue:stack} Access through a load via the stack capability, bypassing compiler enforcement of exclusive access through the associated blinded capability. This can easily happen during speculation; e.g., if the stack capability has a different offset than the compiler expects.
    \item \label{issue:stalecap} Stale data on the stack might be accessed by a newly derived capability in a later function call. Even if accesses to uninitialized memory are not allowed by the compiler, this could happen transiently if the store performing the initialization is speculatively skipped before a load.
    \item \label{issue:influence} Blinded data could also be accessed through a non-blinded capability derived from the stack pointer during speculation. The CSC (see \Cref{sec:csc}) does not curb this issue, as it allows all legal capability modifications to be performed speculatively as well.
  \end{enumerate}
\end{enumerate}

\subsubsection{Concrete Attacks}
To provide evidence for the validity of the security concerns, we demonstrate proof-of-concept attacks (released publicly in \Cref{sec:openscience}) exploiting Issues~\ref{issue:brr} and \ref{issue:stack} on the \blackout prototype.
Issue~\ref{issue:stalecap} and Issue~\ref{issue:influence} can likely be exploited with additional engineering effort.
The current prototype lacks the hardware speculation features required to exploit Issue~\ref{issue:magic}, but we believe it could be exploited on future implementations of \blackout.
Below, we briefly describe our exploits for Issues~\ref{issue:brr} and \ref{issue:stack}, which rely on return address speculation~\cite{koruyeh2018spectre}: the attacker causes a return instruction from a victim function to transiently jump to attacker code by poisoning the return stack buffer.

\paragraph{Exploit for Issue~\ref{issue:brr}}\label{par:exploit-brr}
Consider the following code:

\begin{minted}{asm}
f_write_brr:
  # csp is the stack pointer capability
  # ca0 contains secret data
  csc ca0, -16(csp) # spill the secret to the stack as a BRR
  ccall g # after executing g, speculation starts
...
exploit_gadget: # call to g speculatively returns here
  cld a1, -16(csp) # secret data is loaded to a1 through csp
  cinoffset cs0, csp, a1 # cs0's offset advances by secret a1
  cld a1, 0(cs0) # secret is leaked
\end{minted}

Here, the victim function spills a secret register (using the capability store instruction \texttt{csc}) onto the stack in a BRR.
During speculation, this BRR is accessed via \texttt{cld}, a regular load instruction, which does not check whether the loaded data is a BRR, and as a result, does not set the blinded bit.
Consequently, the loaded secret data in \texttt{a1} can be leaked by using it as an offset for a load.

\paragraph{Exploit for Issue~\ref{issue:stack}}\label{par:exploit-stack}

This attack is structurally very similar to the previous one, and is exemplified below:

\begin{minted}{asm}
f_write_sec:
  # csp is the stack pointer capability
  # a0 contains secret data
  # cs0 is a blinded capability referencing [csp-8, csp)
  csd a0, 0(cs0) # store secret data to the stack through cs0
  ccall g # after executing g, speculation starts
...
exploit_gadget:  # call to g speculatively returns here
  cld a1, -8(csp) # secret data is loaded to a1 through csp
  cinoffset cs0, csp, a1 # cs0's offset advances by secret a1
  cld a1, 0(cs0) # secret is leaked
\end{minted}

The main difference with the previous attack is that in this case, secret data is stored on the stack using a blinded capability, and the compiler enforces that this data is never accessed using the (non-blinded) stack capability during normal execution.
However, this property can be broken during speculation, and the attacker can induce a transient load instruction that accesses the secret data on the stack using the stack pointer.

The leakage in CSC and \blackout shows that reasoning jointly about capability safety, speculative execution, and information leakage is subtle and deserves rigorous formal treatment. Next, we present an alternative design addressing these issues, formally verified to provide strong security guarantees under speculation.

\section{\sys}
\label{sec:semantics}

In this section, we introduce a formal execution model to study
the security of CHERI-like architectures under speculative
execution. To the best of our knowledge, our design is the first to formalize a taint-tracking mechanism where capabilities
keep track of the security level of the data they reference.
We treat
the stack capability as always tainted and mark spills of
\emph{public} values using a distinguished encoding, which we call
\emph{Untainted Register Record} (URR). URRs set the capability tag
and use a reserved marker value in the metadata field.
When the capability tag is set, capability-sized loads consult the metadata field to
distinguish data from capabilities: if the metadata matches the
reserved marker, the value is interpreted as untainted data;
otherwise, it is treated as a capability.

A key challenge of this design is that using a tainted stack capability would cause restored capabilities to be conservatively classified as secret, unnecessarily constraining speculation.
This would cause a
performance degradation, especially in constant-time software, where
memory access patterns and control flow must be public. To address
this, we encode taint information explicitly in capability
metadata. Concretely, when a capability $c$ with address
$\addr$, metadata $\stdCapMeta{\addr}$, and taint status $t$ is stored to
memory, the pair of words
$(\addr, \stdPackTaint{\stdCapMeta{\addr}, t})$ is stored instead,
where $\stdPackTaint{\cdot}$ is a function that embeds the taint bit
into reserved metadata bits. Upon loading, the hardware reconstructs
both the original metadata and the taint of the destination register.

Our design solves the issues of \blackout identified in the previous section.
Since the initial stack capability is tainted, it will remain tainted during execution
and there will be no public aliasing capability addressing the same region if public data is always stored in URRs.
This prevents speculative accesses to secrets on the stack from producing untainted values.

Throughout the remainder of this
section, we illustrate \sys's semantics and the rationale behind its
security guarantees through a running example: we formalize the
attack based on Issue~\ref{issue:stack} (\Cref{sec:blackout}) in
\sys, and use it to contrast \sys's design with alternative
approaches.
\sys is flexible enough to model realistic capability-specific speculation
mechanisms, presented in~\Cref{sec:comparison}.

\subsection{Preliminaries}
\label{sec:prelims}

\paragraph{Capabilities} In \sys, memory is accessed via
capabilities. In turn, capabilities are modeled as pairs
$
  \Cap \ni c \Coloneqq (\addr, \stdCapMeta{\addr}),
$
composed of an address $\addr \in \Addr \defsym \Nat$, and metadata
\(
  \stdCapMeta{\addr} = (p, b, e, t),
\)
where the permission $\perm\in \Perm$ constrains
memory accesses using the pointer. Permissions\footnote{For simplicity, we
  only model CHERI permissions that are required for our formal
  analysis of speculative constant-time. In particular, we do not
  model sealing and unsealing.} can either be
$\noperm$ (no access), $\readonly$ (read-only), $\readwrite$
(read-write) or $\exec$ (read+execute), with the order:
\begin{align*}
  \noperm &\pord \readonly & \readonly &\pord \readwrite & \readonly &\pord \exec
\end{align*}
The addresses $\rbase$ and $\rend$ express the range of memory
$\addrBound{\rbase, \rend}$ that is allowed to be accessed by the
capability; the bit $\taint\in \Taint$ is a \emph{taint} indicating
whether data in such range is considered secret (if
$\taint = \tainted$) or not (if $\taint = \untainted$). Following the
usual convention, we assume that $\Taint$ is a lattice where
$\untainted \le \tainted$. We use the shorthand
$\Meta \defsym \Perm \times \Addr^2 \times \Taint$ for the set of all
capability metadata, and we write $\Cap \defsym \Addr \times \Meta$ for
the set of capability data.

\paragraph{Register files and memories} In \sys, \emph{total register
  files} are modeled as functions
$\reg: \Reg \to \Val^* \times \Taint\times \Tag$ that associate each
register to:
\begin{itemize}
[leftmargin=*]
\item a sequence of \emph{values};
\item a \emph{taint}, which is equal to $\tainted$ if the register content is tainted and to $\untainted$ otherwise;
\item a \emph{capability tag}, which is equal to $\tcap$ when the
  register stores a capability, and to $\tval$, otherwise.
\end{itemize}
For simplicity, we leave the set $\Val$ of values abstract, and we assume
$\mathbb Z\cup \Bool \cup \Addr \cup \Meta \cup \Instrs\cup \Taint \cup \Tag \subseteq \Val$.
Following most CHERI
implementations, registers can hold multiple words, and their
content is considered a \emph{capability} only if the
\emph{capability tag} is $\tcap$, i.e., registers holding values in
$\Cap$, but with tag set to $\tval$ will not be considered
capabilities.

In our semantics, due to out-of-order execution, the content of some
registers may be unavailable. To model this, we introduce
\emph{partial register files} with type
$\reg: \Reg \to (\Val^* \times \Taint\times \Tag)_\bot$.  In such
register files, a register can now also be mapped to $\bot$ when its
value is not available.

\emph{Data memories} are modeled as maps from
addresses $\addr \in \Addr$ to values. We assume that each address
designates one word, i.e., one memory cell holds a single value
$v \in\Val$.
Analogously, $\pc$ increments by one word (corresponding to one instruction) per step.
As other CHERI-like systems do,
\sys keeps track of whether a certain memory region stores a
capability in a dedicated \emph{tag memory}, which we model with a mapping
$\memTag: \Addr \to \Tag$.
Recall that each capability $c\in \Cap = \Addr \times \Meta$ occupies $\szCap$ words,
so $\memTag[k]$ stores the tag that indicates whether the words in the region $[\szCap k, \szCap k + \szCap)$ are considered a capability.
Equivalently, whether the word at address $\addr$ is part of a capability is recorded in $\memTag[\lfloor \addr/\szCap \rfloor]$, which we abbreviate as $\memTag[\addr/\szCap]$.

When dealing with memories, we will often use the
notation $\mem[\addr \mapsto \val]$, to indicate the memory $\mem$
with address $\addr$ updated to $\val$, and
$\mem[[\addr, \addr+\nat) \mapsto \vec \val]$ as a shorthand for
$\mem[\addr \mapsto \val_1][\addr +1 \mapsto \val_2] \ldots
[\addr+\nat -1 \mapsto \val_{\nat}]$. Dually, the notation
$m[\addr, \addr+n)$ denotes the tuple
$(\mem(\addr), \ldots, \mem(\addr+n-1))$.

\paragraph{Running Example}
The first step to model
Issue~\ref{issue:stack} in \sys is to model the stack capability that is
exploited by the attack, which we assume is stored in a dedicated
register $s$:
\begin{equation}
  \label{eq:iss1excap}
  \reg(s) \defsym \big((42, (\readwrite, 0, 256, \tainted)), \untainted, \tcap\big).
\end{equation}
Here $\reg$ maps $s$ to a 2-word value, which is a capability (due to
tag $\tcap$) with address $42$, permissions  $\readwrite$, and bounds
$\addrBound{0,256}$. Its outer taint is $\untainted$, because the
capability, seen as a pointer, is public. Its inner taint is instead
$\tainted$ in our design, to avoid transient leaks of tainted data
(as discussed at the start of this section)---as opposed to \blackout,
where it would be $\untainted$.

\begin{figure}
\centering
\begin{equation*}
\begin{array}{rcll}
  \Exprs \ni e & \Coloneqq & v \mid x \mid \op (e^*) & \emph{Expression}\\
  \Instrs\ni \instr & \Coloneqq & x \gets e \mid \load{x}{e}{\sz} & \emph{Instruction}\\
  & & \mid \store{x}{e}{\sz} \mid \jmp{e}\\
  & & \mid \beqz{x}{\addr}
\end{array}
\end{equation*}
\caption{\sysisa syntax.}
\label{fig:sys-syntax}
\end{figure}

\subsection{Syntax}
\label{sec:syntax}

We model \sys using a simple ISA presented in
\Cref{fig:sys-syntax}, called \sysisa, which extends
\uasm~\cite{guarnieri2020spectector}. An \emph{expression} is either a value $\val$, a register
name $\vx$, or an application of an operator
$\op \in \Op$ to a sequence of expressions.  Each operator has an
associated arity, and we only consider expressions where the arity
constraints of operators are respected.
For the moment we leave the set of operators abstract; it will be further specified in \Cref{sec:exprsem}.
\emph{Instructions} include assignments $x \gets e$, which update $x$
with the value obtained by evaluating $e$, and load instructions
$\load{x}{e}{\sz}$, which evaluate $e$ to a capability and use it to
load $\sz$ consecutive memory words into $x$, represented as a
tuple.  For demonstration purposes, we only allow
$\sz \in \{1, 2\}$---i.e., memory operations over one or two
words---but the model can be extended with arbitrary sizes. For
simplicity, we assume that neither assignments nor load
instructions can modify the value of a dedicated register $\pc$ which
holds the program counter capability. The instruction
$\store{x}{e}{\sz}$ evaluates $e$ to a capability and uses it to store
the value of $\vx$ in memory.
Finally, $\jmp{e}$ is an indirect jump that
evaluates $e$ to an executable capability used as the jump target,
while $\beqz{x}{\addr}$ is a conditional branch that adds the offset
$\addr$ to the program counter if the value of $\vx$ is $0$,
and increments it by $1$ otherwise.

\paragraph{Running Example}
\attackref{par:exploit-stack}{issue:stack} can be modeled by the
following sequence of \sysisa instructions:
  \begin{equation}
    \label{eq:iss1ex}
    \mathtt P \defsym \store {x} {\mathit{t}} 1; \jmp g; \quad \mathtt Q \defsym  \load {y} {\mathit{s}} 1; \load {y} {\mathit{s}\stdCapOff{y}} 1.
\end{equation}
The program $\mathtt P$ models function \texttt{f\_write\_sec}:
it writes a secret contained in register $x$ to memory
through the tainted capability in register $t$, and jumps to the
instruction pointed to by register $g$ via the instruction $\jmp g$.  The
program $\mathtt Q$ models the leak gadget \texttt{exploit\_gadget}. The first load instruction
uses the stack capability in register $s$ to load one word of data
into register $\mathit{y}$. The second load instruction performs the
secret-dependent memory access by loading one word of data from the
capability that is obtained by summing the value of $y$ to the stack
pointer.

\subsection{Expression Semantics}
\label{sec:exprsem}

To
ensure correct taint and resolution propagation, we assume that each
operator is equipped with an interpretation $\overline{\op}$ that is:
\begin{itemize}
[leftmargin=*]
\item \emph{monotone with respect to taints}, i.e.,
  the taint of the output is greater than or equal to that of the
  inputs;
\item \emph{eager with respect to $\bot$}, meaning that if any of the
  inputs is $\bot$, the output must also be $\bot$.
\end{itemize}

To enforce monotonicity with respect to capability manipulation, we
assume that $\Op = \OpArith \uplus \OpCap$ consists of
\begin{itemize}
[leftmargin=*]
\item a set of \emph{arithmetic} operators $\OpArith$,
  modeling ordinary arithmetical operations and always returning values
  tagged with $\tval$;
\item a set of \emph{capability} operators $\OpCap$,
  modeling capability operations, such as shrinking and address
  modifications. These operators are the only ones that can
    return capabilities, i.e., values in $\Cap$ with tag $\tcap$.
\end{itemize}
We require that $n$-ary \emph{arithmetic}
operators $\op \in \OpArith$ have interpretation
$\overline \op: (\Val^* \times \Taint \times \Tag)^n_\bot \to (\Val^* \times
\Taint \times \{\tval\})_\bot$, and since these operators do not output capabilities, we do not impose further constraints.
To ensure that \emph{capability} operators $\op \in \OpCap$ are
well-behaved, we assume that they have the interpretation
$
\overline \op: (\Val^* \times \Taint \times \Tag)_\bot, (\Val^* \times \Taint \times \Tag)_\bot \to (\Val^* \times
  \Taint \times \Tag)_\bot
$.

We also require:

\begin{enumerate}
[leftmargin=*]
\item \label{ax:capop1}\emph{Capability in first position:}
the output has tag $\tcap$ iff only the first input has tag $\tcap$.
\item \label{ax:capop2}\emph{Monotonicity with respect to capabilities}: if the output
  is a capability, then its permissions and bounds are no greater than
  those of the input capability.
\item \label{ax:capop3}\emph{Propagation of capability taint}: if the input capability
  has inner taint $t$, then the output capability must also
  have inner taint $t$.
\end{enumerate}

\paragraph{Expression semantics} Given a \emph{partial register file}
$\reg$, the denotational semantics of expressions is a function
$\sem \cdot_\reg: \Exprs \to (\Val^*\times \Taint\times \Tag)_\bot$
that computes the expression value together with its taint and
capability tag, or $\bot$ if the value of any of the registers used in
the expression is not available. Observe that when an expression
evaluates to a capability, it carries two distinct taints: the
\emph{inner taint}, which describes the security level of the data
referenced by the capability, and the \emph{outer taint}, which
expresses the security level of the capability itself as a value.
The denotational semantics of expressions is defined by structural
recursion:
\begin{align*}
  \sem \val_\reg &\defsym (v, \untainted, \tval) &
                                                   \sem x_\reg &\defsym \reg(x) &
                                                                                                    \sem {\op (e^*)}_\reg &\defsym \overline \op (\sem {e^*} _\reg).
\end{align*}
For convenience, we assume that the set of operators includes an
operator $\stdCapOff$ for adding an integer offset to the pointer of a
capability, with the following interpretation:

\small
\begin{align*}
\sem {e_1 \stdCapOff e_2}_\reg&\defsym
\begin{cases}
\bot & \text{if }\sem{e_1}_\reg = \bot \lor \sem{e_2}_\reg = \bot \\
((\addr + v, \stdCapMeta{\addr}), t_1 \sqcup t_2, \tcap)
  &\text{if }\sem{e_{1}}_\reg = ((\addr, \stdCapMeta{\addr}), t_1, \tcap)\\%
  & \land\sem{e_{2}}_\reg = (v, t_2, \tval)\\%
(0, \untainted, \tval) & \text{otherwise.}
\end{cases}
\end{align*}
\normalsize

\paragraph{Running Example}
We continue on \attackref{par:exploit-stack}{issue:stack}, focusing on the capability
$\mathit{s}\stdCapOff{y}$ used for the secret-dependent load. Here, $s$
has outer taint $\untainted$ (see \eqref{eq:iss1excap}) while $y$,
loaded through the stack capability, is tainted (\Cref{sec:memop}):
$\sem y_\reg = (v, \tainted, \tval)$. Hence
$\sem{\mathit{s}\stdCapOff{y}}_\reg=((42 + v, (\readwrite, 0, 256, \tainted)),
\tainted, \tcap)$, whose outer taint is the supremum of the taints of
$y$ and $s$, i.e., $\tainted \sqcup \untainted = \tainted$. \sys therefore blocks the
secret-dependent load under speculation. In \blackout, by contrast,
the inner taint of $s$ would be $\untainted$, giving
$\sem y_\reg = (v, \untainted, \tval)$ and
$\sem{\mathit{s}\stdCapOff{y}}_\reg=((42 + v, (\readwrite, 0, 256, \untainted)),
\untainted, \tcap)$, so the load would not be prevented.

\subsection{Architectural Semantics}
\label{sec:isasem}

\begin{figure}
  \centering
  \small
  \begin{rules}[Architectural Semantics]
      \Infer[isa][isa-assign]{
        \stdReadMemInst{\mem, \sem{\pc}_\reg} = (x \gets e)
        &
        \sem{e}_{\reg} = (z, t, c)
        &
        x \neq \pc
      }{
        \cf{\mem, \reg}
        \step{\onone}
        \cf{\mem, \reg[\pc\mapsto \sem{\stdNextPc{\pc}}_\reg][x \mapsto \sem{e}_{\reg}]}
      }
      \\
      \Infer[isa][isa-jmp]{
        \stdReadMemInst{\mem, \sem{\pc}_\reg} = \jmp{e} &
        (\stdCap{\addr}, t, \tcap) = \sem{e}_{\reg}
      }{
        \cf{\mem, \reg}
        \step{\stdBr{\addr, \stdCapMeta{\addr}}}
        \cf{\mem, \reg[\pc \mapsto (\stdCap{\addr}, \textcolor{red}{\untainted}, \tcap)]}
      }
      \\
      \Infer[isa][isa-beqz]{
        \begin{array}[b]{c}
        \stdReadMemInst{\mem, \sem{\pc}_\reg} = \beqz{x}{\addr'}\quad
        (v, \_, \_) = \sem{x}_\reg\\
        (\stdCap{\addr}, t, \tcap) = (v = 0) \;?\; \sem{\pc \stdCapOff \addr'}_{\reg} : \sem{\stdNextPc{\pc}}_\reg\\
        \end{array}
      }{
        \cf{\mem, \reg}
        \step{\stdBr{v = 0}}
        \cf{\mem, \reg[\pc \mapsto (\stdCap{\addr}, t, \tcap)]}
      }
      \\
      \hspace{-4mm}
      \Infer[isa][isa-load]{
        \begin{array}[b]{c}
        \stdReadMemInst{\mem, \sem{\pc}_\reg} = \load{x}{e}{\sz} \\
        \stdReadMem{\mem, \sem{e}_{\reg}, \sz} =  (v, t, c) \quad
        \sem{e}_{\reg} = (\stdCap{\addr}, \_, \_)
        \end{array}
      }{
        \cf{\mem, \reg}
        \step{\stdLoad{\addr, \stdCapMeta{\addr}}}
        \cf{\mem, \reg[\pc \mapsto \sem{\stdNextPc{\pc}}_\reg][x \mapsto (v, t, c)]}
      }
      \\
      \hspace{-4mm}
      \Infer[isa][isa-store]{
        \begin{array}[b]{c}
          \stdWriteMem{\mem, \sem{e}_{\reg}, \sz, \sem{x}_{\reg}} = \mem'
          \quad
          v = (\taint = \untainted) \;?\; \sem{x}_{\reg} : \bot\\
          \stdReadMemInst{\mem, \sem{\pc}_\reg} = \store{x}{e}{\sz}\quad
          \sem{e}_{\reg} = ((\addr, (\perm, \rbase, \rend, \taint)), \_, \_)
        \end{array}
      }{
        \cf{\mem, \reg}
        \step{\stdStore{v, (\addr, (\perm, \rbase, \rend, \taint))}}
        \cf{\mem', \reg[\pc \mapsto \sem{\stdNextPc{\pc}}_\reg]}
      }
    \end{rules}
  \begin{rules}[Core rules of memory operations]
    \Infer[memory][check-cap]{
      b \leq l &
      l + \sz \le e &
      \perm' \pord \perm &
      \sz = \szCap \Rightarrow \addr \% \szCap = 0 &
      (\perm, \rbase, \rend, \taint)\neq \stdMagic
    }{
      \stdCheckCap{(\addr, (\perm, \rbase, \rend, \taint)), \sz, \perm'}
    }
    \\
    \Infer[memory][read-mem-inst]{
      \memData(\addr) \in \Instrs &
      \stdCheckCap{(\addr, (\perm, \rbase, \rend, \untainted)), \szInst, \exec}
    }{
      \stdReadMemInst{(\memData, \memTag), ((\addr, (\perm, \rbase, \rend, \untainted)), \untainted, \tcap)} =  \memData(\addr)
    }\\
    \Infer[memory][read-mem]{
      \begin{array}[b]{c}
        \stdCheckCap{(\addr, (\perm, \rbase, \rend, \taint_c)), \sz, \readonly}\\
        \stdReadRes{\memData[\addr, \addr + \sz), \taint_c, \memTag[\addr / \szCap], \sz} = (v, t, c)
      \end{array}
    }{
      \stdReadMem{(\memData, \memTag), ((\addr, (\perm, \rbase, \rend, \taint_c)), \_, \tcap), \sz} = (v, t, c)
    }
    \\
    \Infer[memory][write-mem]{
      \begin{array}[b]{c}
        \stdCheckCap{(\addr, (\perm, \rbase, \rend, \taint_c)), \sz, \readwrite}\\
        \stdWriteRes{(v, t, c), \taint_c, \sz} = (v', c') \\
        \mem' = (\memData[[\addr, \addr + \sz) \mapsto v'], \memTag[\addr/\szCap \mapsto c'])
      \end{array}
    }{
      \stdWriteMem{(\memData, \memTag), ((\addr, (\perm, \rbase, \rend, \taint_c)), \_, \tcap), \sz, (v, t, c)} = \mem'
    }
  \end{rules}
  \caption{\sysisa architectural semantics and memory operation semantics.}
  \label{fig:sys-isa-semantics}
\end{figure}

In this section, we equip \sys with an architectural operational
semantics, formalizing its behavior and, in
particular, the interpretation of memory operations.

Configurations are tuples $((\memData, \memTag), \reg)$, composed by a
data memory, a tag memory and a register file, and ranged over by the
meta-variable $S$.  Following a common approach in modeling
side-channel leaks~\cite{HighAssurance,CTfoundations,guarnieri2020spectector,guarnieri2021hardware}, the architectural semantics of \sys is defined as
a small-step operational semantics, where transitions
$S_0 \step o S_1$ are labeled with observations $o$ revealing the
information that the program leaks to the
adversaries.
Observations are defined by the following BNF:
\[
  \Obs \ni \obs \Coloneqq \onone \mid \stdLoad c \mid \stdStore {\val, c}\mid \stdBr c \mid \stdBr b
\]
where $c\in \Cap$, $b \in \Bool$, and $v \in (\Val^*\times \Taint\times \Tag)_\bot$. The
observation $\onone$ labels transitions that do not leak. The
observations $\stdLoad c$ and $\stdStore {v, c}$ are produced by load
and store operations from and to capability $c$, respectively. In the
$\stdStore {v, c}$ observation, $v$ is equal to the stored value when
the store targets the public region, and to $\bot$ otherwise. This
ensures that the resulting trace reveals declassified values or values
that are already public.  The observation $\stdBr c$ is produced when
an indirect jump to $c$ executes, and $\stdBr b $ is produced by
direct branches whose guard evaluates to $b$.  The branch target does
not need to be included in the observation because attackers can
derive it from their knowledge of the victim program and the value of
$b$.

The $n$-step transition relation is defined as follows:
\[
  \Infer{}{S\step \epsilon\!{}^0\, S} \quad \Infer{S_0 \step o S_1 & S_1 \step O\!{}^n\, S_2}{S_0 \step{o:O}\!{}^{n+1}\, S_2}
\]
Two ISA states $S_0$ and $S_1$ produce the same leakage
(written $S_0 \equiv_{\mathit{ISA}} S_1$) when for every $n \in \Nat$
and observation $O \in \Obs^*$, we have:
\[
  \exists S_0'. S_0 \step{O}\!{}^{n}\, S_0' \ \Leftrightarrow\  \exists S_1'. S_1 \step{O}\!{}^{n}\, S_1'.
\]

The semantics is defined in \Cref{fig:sys-isa-semantics}.  Each rule
evaluates the program counter in the register file $\reg$, and uses
the resulting capability to fetch the current instruction.  In turn,
fetching is performed via the $\stdReadMemInst{\mem, \sem{\pc}_\reg}$
function, which is formally defined by
Rule~\ref{memory:read-mem-inst}.
In this rule, the premise on the left
ensures that the memory location pointed to by the program counter
is an instruction, and the premise on the right validates the
capability that is used to fetch the next instruction via the
predicate $\stdCheckCap{\cdot, \cdot, \cdot}$, defined by
Rule~\ref{memory:check-cap}. Essentially, this rule ensures that the
capability offset is in bounds, and it has sufficient
permissions. When the size of the load is $\szCap$, the rule also ensures
that the accessed address is aligned. This requirement is necessary
for correct handling of capability tags, as $\memTag(\addr/\szCap)$ keeps
track of the tag associated with the data stored in
$\memData[\addr, \addr+\szCap)$ (see Rules~\ref{memory:write-mem}
and \ref{memory:read-mem}).
Consequently, the alignment condition
enforces consistency of that data and tag. Finally, the last
premise ensures that the metadata field of the capability is not a
dedicated value $\stdMagic$, which marks spills of untainted values; we
discuss this mechanism in more detail when we describe the semantics
of memory operations in \Cref{sec:memop}.

We now turn to the semantics of instructions. Assignments $x\gets e$ are
evaluated via Rule~\ref{isa:isa-assign}, which evaluates expression
$e$ and updates $x$ accordingly.  The $\pc$ register is updated to
point to the next instruction, while its metadata, taint and capability tag
are not changed. Rule~\ref{isa:isa-jmp} executes indirect jump instructions
$\jmp{e}$ by updating the program counter with the jump target
$\sem e_\reg$, which is leaked via the observation. Conditional branches
$\beqz{x}{\addr'}$ are executed with Rule~\ref{isa:isa-beqz}, which
evaluates the branch condition $x$ to a value $v$, and sums the offset
$\addr'$ to the program counter if $v=0$, and increments the program counter otherwise.

Load instructions $\load{x}{e}{\sz}$ are evaluated via
Rule~\ref{isa:isa-load}.  The load capability is obtained by
evaluating $\sem{e}_{\reg}$.  The destination register $x$ is then
updated with the ISA-level value produced by
$\stdReadMem{\mem, \sem{e}_{\reg}, \sz}$, together with its associated
taint and capability tag. The value of
$\stdReadMem{\mem, \sem{e}_{\reg}, \sz}$ is provided by
Rule~\ref{memory:read-mem}, which is also responsible for ensuring
that the load capability has sufficient permissions, and for providing
the correct ISA-level representation of unspills and loaded
capabilities, given by the helper function
$\stdReadRes{\cdot, \cdot, \cdot, \cdot}$.
Leaving the helper function unspecified allows
our semantics to represent different memory designs. The one of \sys
is described in detail in \Cref{sec:memop}, the one of \blackout
in \Cref{sec:modelingbo}.
Store instructions $\store{x}{e}{\sz}$ are handled by
Rule~\ref{isa:isa-store}. The store capability is obtained by
evaluating $\sem{e}_{\reg}$, and the value is taken from register
$x$. Memory is updated via
$\stdWriteMem{\mem, \sem{e}_{\reg}, \sz, \sem{x}_{\reg}}$, whose value
is provided by Rule~\ref{memory:write-mem}. This rule is responsible
for ensuring that the store capability has sufficient permissions to
write and for computing the ISA-level representation $(v', c')$ of
the value and capability tag to store, given by the helper function
$\stdWriteRes{\cdot, \cdot, \cdot}$. Again, leaving this helper
function unspecified allows our semantics to represent different
designs. The updated memory is obtained by writing the value $v'$ and
tag $c'$ at the addresses spanned by the store capability.  Note that
when the capability's inner taint is untainted, Rule~\ref{isa:isa-store} leaks the
stored value to include declassified values in the leakage trace.
\paragraph{Running Example} We illustrate \sys's architectural
semantics on~\eqref{eq:iss1ex}, assuming $g$ points to a benign,
skip-like target at $\addr'$ performing $w \gets w$; this is
\emph{not} the \attackref{par:exploit-stack}{issue:stack} behavior, where $\mathtt Q$
is reached via misspeculation (covered in \Cref{sec:sem}). The
architectural execution of
$\mathtt P \defsym \store {x} {\mathit{t}} 1; \jmp g$ is
\begin{align*}
  \cf{\mem, \reg}
  &\step{\stdStore{\bot,\, \sem{t}_\reg}}
    \cf{\mem', \reg[\pc \mapsto \sem{\stdNextPc\pc}_\reg]}\\
  &\step{\stdBr{\sem{g}_\reg}}
    \cf{\mem', \reg[\pc \mapsto (\stdCap{\addr'}, \untainted, \tcap)]}\\
  &\step\onone
    \cf{\mem', \reg[\pc \mapsto (\addr'+1,\stdCapMeta{\addr'}, \untainted, \tcap)]}.
\end{align*}
The store executes via Rule~\ref{isa:isa-store}: since $t$ has inner
taint $\tainted$, the leaked observation carries $\bot$ rather than
the declassified value, while the secret
$\reg(x) = (\val, \tainted, \tval)$ is written to $\mem$, giving
$\mem'$. The jump (Rule~\ref{isa:isa-jmp}) leaks the target
$\sem{g}_\reg = (\stdCap{\addr'}, \untainted, \tcap)$ and transfers
control to the benign $\addr'$. Finally, $w \gets w$ fires
Rule~\ref{isa:isa-assign}, producing the empty observation $\onone$
and advancing the program counter.

\subsection{Memory Operation Semantics}
\label{sec:memop}
\begin{figure*}
  \begin{rules}
      \Infer[helpersys][write-data]{
        \sz = \szAddr
      }{
        \stdWriteRes{(v, \taint, c), \taint_c, \sz} = (v, \tval)
      }
      \Infer[helpersys][write-cap]{
        \stdPackTaint{\stdCapMeta{v}, \taint} = v_m
      }{
        \stdWriteRes{(\stdCap{v}, \taint, \tcap), \taint_c, \szCap} = ((v, v_m), \tcap)
      }
      \Infer[helpersys][spill-value]{
        (v_m', c') = (v_m = 0 \land \taint = \untainted) \;?\; (\stdMagic, \tcap) : (v_m, \tval)
      }{
        \stdWriteRes{((v, v_m), \taint, \tval), \taint_c, \szCap} = ((v, v_m'), c')
      }\\
      \Infer[helpersys][read-data]{
        \sz = \szAddr \lor c = \tval\\
      }{
        \stdReadRes{v, \taint_c, c, \sz} =  (v, \taint_c, \tval)
      }
      \Infer[helpersys][unspill-value]{
        v_m = \stdMagic
      }{
        \stdReadRes{(v, v_m), \taint_c, \tcap, \szCap} = ((v, 0), \untainted, \tval)
      }
      \Infer[helpersys][read-cap]{
        v_m \neq \stdMagic &
        \stdUnpackTaint{v_m} = (\stdCapMeta{v}, \taint)
      }{
         \stdReadRes{(v, v_m), \taint_c, \tcap, \szCap} = ((v, \stdCapMeta{v}), \taint \sqcap \taint_c, \tcap)
      }
    \end{rules}
  \caption{\sys's memory operation semantics.}
\label{fig:mem-op-semantics}
\end{figure*}

At a high level, our memory semantics combines three mechanisms:
\begin{enumerate}
[leftmargin=*]
\item On reads, the capability used for the access determines the taint
  of the loaded value. Hence, values stored via a capability to
  untainted data are treated as untainted by any following load, and
  thereby declassified.

\item 2-word-wide public values with high-order bits equal to zero are
  represented in memory by \emph{Untainted Register Records} (URRs),
  i.e., pairs $(v, v_{\mathit{magic}})$ tagged as capabilities, where
  $v_{\mathit{magic}}$ is a magic value that does not collide with any
  valid capability metadata. This allows URRs to be distinguished from
  genuine capabilities. Unlike \blackout, where the BRR mechanism only
  applies to register spills, the URR mechanism
  applies to every untainted value store in \sys.
  To improve memory efficiency, future work could refine this mechanism.

\item When a capability is written to memory, its \emph{outer} taint
  is saved as part of its stored metadata and later recovered when the
  capability is read back. If this encoded taint is $\untainted$,
  the outer taint is restored to $\untainted$
  regardless of the taint of the loading capability.  This mechanism
  avoids the performance degradation caused by conservatively
  tainting originally untainted capabilities.
\end{enumerate}

The semantics of memory accesses is defined in
\Cref{fig:mem-op-semantics}, via two helper functions that convert
between the ISA-level representation of values and their in-memory
representation. The function
$\stdWriteRes{(v, \taint, c), \taint_c, \sz}$ takes as arguments the
value $(v, \taint, c)$ being stored---i.e., the ISA-level value $v$,
together with its taint $\taint$ and capability tag $c$---the inner taint
$\taint_c$ of the capability used to perform the store, and the store
size $\sz$; it returns the pair $(v', c')$ to be written respectively
to $\memData$ and $\memTag$. Dually, the function
$\stdReadRes{v, \taint_c, c, \sz}$ takes as arguments the raw value
$v$ read from $\memData$, the inner taint $\taint_c$ of the capability
used to perform the load, the tag $c$ read from $\memTag$ at the
corresponding address, and the load size $\sz$; it returns the ISA-level value,
taint, and capability tag.

For writes, operations of size $\sz=1$ are handled by
Rule~\ref{helpersys:write-data}, which tags them as non-capabilities,
since capabilities have size 2. Capability writes are handled via
Rule~\ref{helpersys:write-cap}, which encodes the taint bit into the
metadata via the function $\stdPackTaint\cdot$ before storing it in
memory. This way the capability's taint can be restored at load time
(Rule~\ref{helpersys:read-cap}) via function
$\stdUnpackTaint\cdot$. For generality, we do not further specify
these two functions: we only assume that
$\stdUnpackTaint {\stdPackTaint\cdot}$ is the identity.  Finally,
Rule~\ref{helpersys:spill-value} ensures that untainted values are spilled
to memory using URRs, as previously discussed. The restriction to
pairs whose second element is zero ensures no information loss.

For reads, Rule~\ref{helpersys:read-data} specifies that data of size
$1$, or data not tagged as a capability, is forwarded to the
ISA level as a non-capability: its value is unchanged, and its taint is
set to $\taint_c$, i.e., to the taint of the
capability used for loading, rather than to any taint recorded in
memory. If the loaded data is tagged as a capability and the access
size is $2$, then two cases arise, distinguished by
the value $v_m$ of the second memory word: if $v_m = \stdMagic$,
i.e., the loaded data is a URR,
Rule~\ref{helpersys:unspill-value} returns the
spilled value and marks it as untainted.
Otherwise, Rule~\ref{helpersys:read-cap} applies, and the value is
reconstructed as a capability, using $\stdUnpackTaint\cdot$ to recover
its metadata and stored taint from $v_m$. The
two taints are combined via $\sqcap$, so the resulting capability is
untainted if the stored taint is $\untainted$---meaning that the
capability was untainted before being stored---or if the accessing
capability is untainted.

\paragraph{Running Example} We illustrate memory operations with
two stores. First, when $\mathtt P$ stores the tainted $x$,
Rule~\ref{helpersys:write-data} applies, storing the value as-is and
tagged as a non-capability; when $\mathtt Q$ later loads it via
$\load {y} {\mathit{s}} 1$, the tag in $\memTag$ is $\tval$, so
Rule~\ref{helpersys:read-data} returns $(v, \tainted, \tval)$---the
observed taint is entirely determined by that of $s$. Second, if $x$
holds a public 2-word value $(v, 0)$, it is spilled via
Rule~\ref{helpersys:spill-value} as a URR:
$\stdWriteRes{((v, 0), \untainted, \tval), \taint, \szCap} = ((v,
\stdMagic), \tcap)$, i.e., tagged as a \emph{capability} in $\memTag$
with second word $\stdMagic$, \emph{independently} of the storing
capability's taint. When the value is read back the tag is $\tcap$ and the second
word is $\stdMagic$, so Rule~\ref{helpersys:unspill-value} applies and
the value is read back untainted, \emph{independently} of the loading
capability's taint.

\subsection{Hardware Semantics}
\label{sec:sem}

\paragraph{Reorder buffers}
Following a common approach~\cite{CTfoundations,guarnieri2021hardware,prospect},
out-of-order and speculative execution are modeled by using a
\emph{reorder buffer} (ROB): a partial map
$\buf: \Nat \rightharpoonup \Peinstr$, expressing the position of a partially
evaluated instruction in the buffer. In turn, partially evaluated instructions are
described by the following BNF:
\begin{align*}
  \Tg \ni \tg &\coloneqq \varepsilon \mid \stdCap{\addr} \mid (\stdCap{\addr}, \sz, z)\\
  \Peinstr \ni \peinstr &\Coloneqq x \gets e \specat{\tg}\mid x \gets (v, t, c) \specat{\tg} \mid \load{x}{e}{\sz} \specat{\tg} \mid \\
  & \quad \store{x}{e}{\sz} \specat{\tg} \mid \store{(v, t, c)}{\stdCap{\addr}}{\sz} \specat{\tg}
\end{align*}
where $e\in \Exprs$ is an expression that has not yet been evaluated,
$(v, t, c)\in \Val^*\times \Taint\times \Tag$ represents an evaluated
expression for assignments, or a load/store value,
and $\stdCap{\addr}\in \Cap$ is an evaluated capability for memory access.
In order to resolve speculative choices,
partially evaluated instructions are tagged with a speculation tag $\tg \in \Tg$.
In turn, the speculation tag can be:
\begin{itemize}
[leftmargin=*]
\item the constant $\varepsilon$, if the partial execution of the corresponding instruction has so far made no unresolved speculative choice;
\item the $\pc$ capability from which the branch instruction was
  fetched, in the case of control-flow speculation;
\item a triple
  $(\stdCap\addr, \sz, z)\in \Cap\times \{1, 2\}\times \Nat_\bot$,
  which is used to tag assignments resulting from loads, in the
  presence of speculative store-to-load forwarding. Precisely, $\stdCap\addr$
  is the load capability, $\sz$ is the load size, and $z$ is the index
  in the ROB of the aliasing store that forwarded the
  data, or $\bot$ if the value is fetched from the memory.
\end{itemize}

Given a ROB $\buf$ and an index $i$, the ROB
$\buf \setminus i$ is pointwise identical to $\buf$, but is undefined
on $i$. Similarly, the ROB $\pref i \buf$ is pointwise identical to
$\buf$,
but is undefined at indices greater than or equal to $i$.
More formally, their domains are:
\begin{align*}
  \dom{\buf \setminus i} &\defsym \dom{\buf} \setminus \{i\},\\
  \dom{\pref i \buf} &\defsym \dom{\buf} \cap \{0, \ldots, i-1\},
\end{align*}
and, for every $j$ in the respective domain:
\begin{align*}
  (\buf \setminus i)(j) &\defsym \buf(j), &
  (\pref i \buf)(j) &\defsym \buf(j).
\end{align*}
\paragraph{Buffer application} Instructions in the ROB
cannot be evaluated directly under the architectural register file, as
it may be stale with respect to updates performed by earlier buffered
instructions. Instead, evaluation must use a register file that
reflects the updates of the preceding portion of the buffer. This is
captured by the function $\apl(\buf, \reg)$, defined as follows:
\begin{equation*}
\apl(\buf, \reg) \defsym
\begin{cases}
  \reg & \dom{\buf} = \emptyset\\
  \apl(\buf\setminus i, \aplinst(\buf(i), r)) & i = \min \dom{\buf},
\end{cases}
\end{equation*}
\noindent
where:
\begin{align*}
\aplinst(x \gets (v, t, c) \specat{\tg}, \reg) & \defsym r[\pc \mapsto \sem{\stdNextPc{\pc}}_\reg][x \mapsto (v, t, c)] \\
\aplinst(x \gets e \specat{\tg}, \reg) & \defsym r[\pc \mapsto \sem{\stdNextPc{\pc}}_\reg][x \mapsto \bot] \\
\aplinst(\load{x}{e}{\sz} \specat{\tg}, \reg) & \defsym r[\pc \mapsto \sem{\stdNextPc{\pc}}_\reg][x \mapsto \bot] \\
\aplinst(\store{x}{e}{\sz} \specat{\tg}, \reg) & \defsym r[\pc \mapsto \sem{\stdNextPc{\pc}}_\reg].
\end{align*}
This function applies the updates of $\buf$ to $\reg$ by processing
instructions in program order, starting from the smallest index in the
buffer.
It also updates the program counter for each instruction.
Note that the result of a buffer application in general is a \emph{partial
  register file}. Using \emph{partial} register files is also
convenient to model how our design prevents the leak of tainted data
during speculative execution. Following~\cite{prospect}, this is done
by evaluating instructions using exclusively untainted values during
speculative execution. We model this via an \emph{untainted projection},
 $\lowproj \cdot$, defined as follows:
\[
  \lowproj{w} \defsym
  \begin{cases}
    w & \text{if } w = (v, \untainted, c) \text{ for some } v \in \Val^*, c \in \Tag,\\
    \bot & \text{otherwise.}
  \end{cases}
\]
The untainted projection is extended pointwise to register files by stipulating:
\(
  \lowproj{\reg}(\vx) \defsym \lowproj{\reg(\vx)}.
\)
The untainted projection of a register file can be used to sanitize
speculatively tainted values, as follows:
\begin{equation*}
\aplsan(\buf, r) =
\begin{cases}
  r & \text{if }\dom{\buf} = \emptyset \\
  \lowproj{\apl(\buf, r)} & \text{otherwise.}
\end{cases}
\end{equation*}
During out-of-order execution, when the value of an expression at the $i$-th entry of a ROB $\buf$ might be leaked through side channels,
that expression is evaluated with the register file
$\aplsan(\pref i \buf, \reg)$.
The case $\dom{\pref i \buf} = \emptyset$ captures instructions at the
very beginning of the ROB: such instructions are non-speculative,
evaluated directly against the architectural register file
$r$. All other instructions have at least one predecessor in the
buffer, and are evaluated against $\lowproj{\apl(\pref i \buf, \reg)}$, i.e.,
against the sanitized register file obtained by applying their
predecessors' updates and discarding any value that is tainted.
This prevents leaking tainted values during speculation.

For instance, take a ROB $\buf$ where index $0$
holds $x \gets (\mathit{secret}, \tainted, c)$, and index $1$ holds
a subsequent instruction $\peinstr$ that references and leaks $x$. When $\peinstr$
is evaluated, it is evaluated in
$\aplsan(\pref 1 \buf, \reg) = \aplsan(\{0 \mapsto x \gets
(\mathit{secret}, \tainted, c)\}, \reg)$. Since $\pref 1 \buf$ is
not empty, this reduces to $\lowproj{\apl(\pref 1 \buf, \reg)}$:
first, $\apl$ propagates the update to $x$; then, $\lowproj \cdot$
maps the tainted value of $x$ to $\bot$. Therefore, when $\peinstr$
is evaluated, $x$ is unavailable, preventing its transient leak.
Only once the first instruction retires, can $\peinstr$
 access $x$'s value.

\paragraph{Microarchitectural contexts} Having discussed how our
hardware model supports speculative and out-of-order execution, we
turn to the speculative semantics itself.
Following~\cite{guarnieri2021hardware,prospect}, we model the attacker by including in our
configurations a \emph{microarchitectural context}, which abstracts
both the observations available to the adversary and its influence on
execution.

\begin{definition}[Microarchitectural context]
  A microarchitectural context is a structure
  $(\Ctx, \stdUpdateMu{\cdot, \cdot},\stdPredPc{\cdot},
  \stdNextMu{\cdot})$ with non-empty carrier $\Ctx$, where:
  \begin{enumerate}
  [leftmargin=*]
  \item $\stdUpdateMu {\cdot, \cdot} : \Ctx \times A \to \Ctx$ returns
    a microarchitectural context that is obtained by updating the
    input context with information leaked during evaluation,
    taken from a set $A$ which we do not further specify;
  \item $\stdPredPc {\cdot} : \Ctx \to \Addr$ takes the
    current microarchitectural context and returns a jump target (address)
    prediction;
  \item $\stdNextMu {\cdot} : \Ctx \to \Dir$ takes the
    current context and returns a directive
    from $\Dir \Coloneqq \dfetch \mid \dexecute i \mid \dcommit$, where $i \in \Nat$.
  \end{enumerate}
\end{definition}

The semantics interacts with the microarchitectural context through
the operations $\stdUpdateMu {\cdot, \cdot}, \stdPredPc\cdot$ and
$\stdNextMu\cdot$. First, all observable microarchitectural effects
are reported via $\stdUpdateMu{\cdot,\cdot}$, ensuring that any
information that may leak is reflected in the context. Second,
speculative choices are generated using $\stdPredPc{\cdot}$, which is
invoked to predict the target of jump instructions. Finally, the
execution order is governed by $\stdNextMu{\cdot}$, whose output
determines whether the processor fetches a new instruction (directive
$\dfetch$), executes the $i$-th instruction in the ROB
(directive $\dexecute i$), or commits the oldest instruction
($\dcommit$).  In the following, we use a fixed arbitrary
microarchitectural context. This ensures that all of our results hold
independently of the specific adversary.

\begin{figure*}[t]
  \small
  \begin{rules}
\Infer[hw][step]{
  \begin{array}[b]{c}
  \mu'= \stdUpdateMu{\mu, \lowproj{\buf}} \quad
  d = \stdNextMu{\mu'}\\
  \cf{\mem, \reg, \buf, \mu'} \sstep{d}{} \cf{\mem', \reg', \buf', \mu''}
  \end{array}
}{
  \cf{\mem, \reg, \buf, \mu} \sstep{}{} \cf{\mem', \reg', \buf', \mu''}
}
\Infer[hw][fetch-branch-predict-pc]{
  \begin{array}[b]{c}
  (\stdCap{\addr}, t, c)=\sem{\pc}_{\aplsan (\buf, \reg)}\quad
  \instr \in \qty{\beqz x \addr'', \jmp e}\quad
  \instr = \stdReadMemInst{\mem, (\stdCap{\addr}, t, c)}\\
  i = \sup\dom{\buf}\quad
  \buf' = \buf[i + 1\mapsto \pc\gets((\stdPredPc{\mu}, \stdCapMeta{\addr}), t, c)\specat{\stdCap{\addr}}]
  \end{array}
}{
  \cf{\mem, \reg, \buf, \mu}
  \sstep{\dfetch}{}
  \cf{\mem, \reg,
  \buf',
  \stdUpdateMu{\mu, \sem{\pc}_{\aplsan (\buf, \reg)}}}
}
\\
\Infer[hw][fetch-other]{
  \begin{array}[b]{c}
  \instr = \stdReadMemInst{\mem, \sem{\pc}_{\aplsan (\buf, \reg)}}\\
  \instr \not\in \qty{\beqz x \addr'', \jmp e}\quad
  \buf' = \buf[\sup\dom{\buf} + 1\mapsto \instr\specat{\varepsilon}]
  \end{array}
}{
  \cf{\mem, \reg, \buf, \mu}
  \sstep{\dfetch}{}
  \cf{\mem, \reg, \buf', \stdUpdateMu{\mu, \sem{\pc}_{\aplsan (\buf, \reg)}}}
}
\Infer[hw][execute-assign]{
  \buf(i) = x\gets e \specat{\varepsilon}&
  x \neq \pc&
  (v, t, c) = \sem{e}_{\apl (\pref i \buf, \reg)}
}{
  \cf{\mem, \reg, \buf, \mu}
  \sstep{\dexecute{i}}{}
  \cf{\mem, \reg, \buf[i \mapsto x \gets (v, t, c) \specat{\varepsilon}], \mu}
}\\
\Infer[hw][execute-jmp-ok]{
  \begin{array}[b]{c}
  \buf(i) = \pc\gets \addr' \specat{\stdCap{\addr}}\hfill
  \buf' = \buf[i\mapsto \pc\gets \addr' \specat{\varepsilon}]~~
  \addr' = (\addr_0, \_, \tcap)\\
  \jmp{e} = \stdReadMemInst{m, \sem{\pc}_{\aplsan(\pref i \buf, \reg)}}~~
  (\addr_0, \_, \tcap) = \sem{e}_{\aplsan (\pref i \buf, r)}\\
  \end{array}
}{
  \cf{\mem, \reg, \buf, \mu}
  \sstep{\dexecute{i}}{}
  \cf{\mem, \reg, \buf', \stdUpdateMu{\mu, \addr_0}}
}
\Infer[hw][execute-jmp-hazard]{
  \begin{array}[b]{c}
  \buf(i) = \pc\gets (\addr_1, \_, \_)\specat{\stdCap{\addr}}~~
    \jmp{e} = \stdReadMemInst{m, \sem{\pc}_{\aplsan(\pref i \buf, \reg)}} \\
    \addr_1 \neq \addr_0 \hfill (\addr_0, \_, \tcap) = \sem{e}_{\aplsan (\pref i \buf, r)} \hfill
  \buf' = \pref {i + 1} \buf[i\mapsto \pc\gets (\addr_0, \textcolor{red}{\untainted}, \tcap) \specat{\varepsilon}]
  \end{array}
}{
  \cf{\mem, \reg, \buf, \mu}
  \sstep{\dexecute{i}}{}
  \cf{\mem, \reg, \buf', \stdUpdateMu{\mu, \addr_0}}
}\\
\Infer[hw][execute-load-fwd]{
  \begin{array}[b]{c}
    \buf(i) = \load{x}{e}{\sz} \specat{\varepsilon}\hfill
    x\neq \pc \hfill
    [\addr, \addr + \sz) = [\addr', \addr' + \sz')  \hfill
    ((\addr, (\perm, \rbase, \rend, t_c)), \_, \tcap) = \sem{e}_{\aplsan (\pref i \buf, \reg)} \hfill
    \stdCheckCap{(\addr, (\perm, \rbase, \rend, t_c)), \sz, \readonly}\\
    (j = \max \lbrace
    j < i: \buf(j) = \store{\_}{\stdCap{\addr'}}{\sz'}\specat{\tg}
    \land [\addr, \addr + \sz) \cap [\addr', \addr' + \sz') \neq \emptyset
  \rbrace)
    \quad
    t_0 \sqsubseteq t_c \quad
  \buf(j) = \store{(v_0, t_0, c_0)}{\stdCap{\addr'}}{\sz'}\specat{\tg}
    \\
  \end{array}
}{
  \cf{\mem, \reg, \buf, \mu}
  \sstep{\dexecute{i}}{}
  \cf{\mem, \reg, \buf[i \mapsto (x\gets \stdReadRes{v_0, t_c, c_0, \sz}\specat{(\stdCap{\addr}, \sz, j)})], \stdUpdateMu{\mu, (\addr, (\perm, \rbase, \rend, t_c))}}
}\\
\Infer[hw][execute-store-ok]{
  \begin{array}[b]{c}
  \buf(i) = \store{x}{e}{\sz}\specat{\varepsilon}\hfill
  ((\addr, (\perm, \rbase, \rend, \taint_c)), \_, \tcap) = \sem{e}_{\aplsan (\pref i \buf, \reg)} \hfill
  \stdCheckCap{(\addr, (\perm, \rbase, \rend, \taint_c)), \sz, \readwrite}\hfill
  (v, t, c) = \sem{x}_{\apl(\pref i \buf, \reg)}\\
  (\forall j > i, \buf(j) = x \gets v \specat{(\stdCap{\addr'}, \sz', k)} \land k \neq i \Rightarrow
  (k > i \lor [\addr, \addr + \sz) \cap [\addr', \addr + \sz') = \emptyset))\quad
  (v', c') = \stdWriteRes{(v, t, c), \taint_c, \sz}\\
  \end{array}
}{
  \cf{\mem, \reg, \buf, \mu}
  \sstep{\dexecute{i}}{}
  \cf{\mem, \reg, \buf[i \mapsto (\store{(v', t, c')}{(\addr, (\perm, \rbase, \rend, \taint_c))}{\sz})\specat{\varepsilon}], \stdUpdateMu{\mu, (\addr, (\perm, \rbase, \rend, \taint_c))}}
}
\\
\Infer[hw][commit-assign]{
i = \min\dom{\buf}&
    x = \pc \Rightarrow \tg= \varepsilon &
  \buf(i) = x \gets v \specat{\tg}&
}{
  \cf{\mem, \reg, \buf, \mu}
  \sstep{\dcommit}{}
  \cf{\mem, \reg[\pc \mapsto \sem{\stdNextPc{\pc}}_\reg][x\mapsto v], \buf \setminus i, \mu}
}\\
\Infer[hw][commit-store]{
  i = \min\dom{\buf}\quad
  \buf(i) = \store{(v, t, c)}{\stdCap{\addr}}{\sz} \specat{\varepsilon}\quad
  \mem' = (\memData[[\addr, \addr + \sz) \mapsto v], \memTag[\addr/\szCap \mapsto c])
}{
  \cf{(\memData, \memTag), \reg, \buf, \mu}
  \sstep{\dcommit}{}
  \cf{\mem', \reg[\pc \mapsto \sem{\stdNextPc{\pc}}_\reg], \buf \setminus i, \stdUpdateMu{\mu, \stdCap{\addr}}}
}
\end{rules}
  \caption{Hardware semantics of \sys, excerpt.}
  \label{fig:hwsem}
\end{figure*}

\paragraph{Hardware semantics} Hardware configurations are
tuples of the form $(\mem, \reg, \buf, \ctx)$, where $\mem$ is a pair of a
data and a tag memory, $\reg$ is a register file, $\buf$ is a ROB,
and $\ctx$ is a microarchitectural context. Our hardware-level
semantics of \sys is a small-step operational semantics. An excerpt of
the transition rules is in \Cref{fig:hwsem}.
\iftechreport
The complete set of rules is deferred to \Cref{sec:apprules}.
\else
The complete set of rules is in \cite{scheri-proof}.
\fi
The hardware semantics proceeds at two levels. Externally, the
behavior is described by Rule~\ref{hw:step}: at each step, the
microarchitectural context is updated with the untainted projection
of the ROB, ensuring that scheduling and prediction
decisions depend only on untainted information. The context then
produces a directive via $\stdNextMu{\cdot}$, which determines the
next directive $\dir$, governing an internal transition.

If the
directive is $\dfetch$, and the program counter points to a branch
instruction ($\instr \in \qty{\beqz x \addr'', \jmp e}$),
Rule~\ref{hw:fetch-branch-predict-pc} applies. It uses
$\stdPredPc{\cdot}$ to obtain a predicted target, and appends a
speculative update of the program counter to the ROB by setting:
\[
  \buf' = \buf[i + 1\mapsto \pc\gets((\stdPredPc{\mu}, \stdCapMeta{\addr}), t, c)\specat{\stdCap{\addr}}].
\]
Here, the branch instruction is tagged by its $\pc$ value,
indicating the speculation on its target, which is resolved at
execute time.
Finally,
the rule leaks the current program counter to the microarchitectural
context $\mu$ via
$\stdUpdateMu{\mu, \sem{\pc}_{\aplsan (\buf, \reg)}}$.  If the program
counter points to a non-control-flow instruction,
Rule~\ref{hw:fetch-other} applies. It appends the fetched instruction
to the ROB with an empty speculation tag, and leaks the
current program counter to the context.

When the issued directive is $\dexecute i$, different rules can be
applied, depending on the $i$-th entry of the reorder
buffer. Rule~\ref{hw:execute-assign} evaluates assignment instructions
$\vx \leftarrow e$ in the buffer by computing the value of $e$ under
the partially applied register file ${\apl (\pref i \buf, \reg)}$, and
storing it in the buffer. Rule~\ref{hw:execute-jmp-ok} resolves
indirect jumps. The target $e$ is evaluated under the sanitized
register file $\aplsan (\pref i \buf, r)$. This ensures that the value
of the jump target is not stale and that the leaked jump target
depends only on untainted values. If the target resolves as correct,
the speculative annotation is cleared and the resulting target is
reported to the context. If instead the resolved target $\addr_0$ does
not match the predicted one, Rule~\ref{hw:execute-jmp-hazard} applies:
the corresponding entry of the ROB is updated with the
correct target $\addr_0$ (retagged as $\untainted$), the speculation
tag is cleared, and all younger instructions are dropped from the
buffer.  Conditional branches are executed analogously to indirect
jumps by Rules~\rlref{hw}{execute-beqz-ok}
and \rlref{hw}{execute-beqz-hazard}, whose details are in
\iftechreport
\Cref{sec:apprules}.
\else
\cite{scheri-proof}.
\fi

Rule~\ref{hw:execute-load-fwd} evaluates load instructions. It first checks that the target register is not the
program counter, as this can only be updated via
$\kwd{beqz}$ and $\kwd{jmp}$ instructions, then it evaluates the load
address to a capability
$(\stdCap{\addr}, \_, \tcap) = \sem{e}_{\aplsan (\pref i \buf,
  \reg)}$.  Using $\aplsan$ ensures that the address only depends on
untainted values in speculative execution.  The rule then checks
whether the resulting capability has sufficient permissions for
performing the load in the
$\stdCheckCap{\stdCap{\addr}, \sz, \readonly}$ premise. Then, if
an aliasing evaluated store exists at index $j < i$, the stored value
is forwarded by setting the $i$-th entry in the ROB to an
assignment of that value to $x$. Note that this forwarding is a form
of SSB speculation, as the ROB may contain a not-yet-evaluated
aliasing store at some index $j'$ with $j < j' < i$.
Thus, the
speculation tag of the assignment carries the load
capability $\stdCap{\addr}$, the load size $\sz$, and the index $j$ in
the ROB of the store instruction that forwarded its
value.

Rule~\ref{hw:execute-store-ok} executes store instructions and
resolves store-to-load dependency speculation. As for the load rule,
the store address is first evaluated under the sanitized register file
to ensure that no tainted values are leaked during speculative
execution. The resulting capability is then checked for sufficient
permissions via $\stdCheckCap{\cdot, \sz, \readwrite}$. The value to
be stored is obtained by evaluating $x$ under the partially applied
register file $\apl(\pref i \buf, \reg)$.
The value to be written is then computed via $\stdWriteRes{\cdot}$, the
corresponding entry in the ROB is updated, and the accessed
address is leaked to the microarchitectural context.
Crucially, the rule additionally enforces consistency with prior
store-to-load forwarding.
Recall that Rule~\ref{hw:execute-load-fwd} may speculatively forward a
value from an older store to a younger load. To ensure that such
forwarding is sound, the store rule checks that no younger load has
speculated on a conflicting store. This is captured by the side
condition on indices $j > i$: if a younger load has been resolved by
forwarding from a store at position $k \neq i$, then either that
store is younger than the current one ($k > i$), or the accessed
memory regions do not overlap (otherwise, the pipeline is flushed).

Rule~\ref{hw:commit-assign} commits assignments by updating the
architectural register file with the computed value and removing the
instruction from the buffer. This rule explicitly checks that if the
target register is $\pc$, the speculation tag has to be $\varepsilon$,
to ensure that speculative jumps have been validated at execution
time. Rule~\ref{hw:commit-store} commits store instructions by
updating the data and tag memories, leaking the accessed
address to the context a second time and removing the instruction from the buffer.
\begin{figure*}[t]
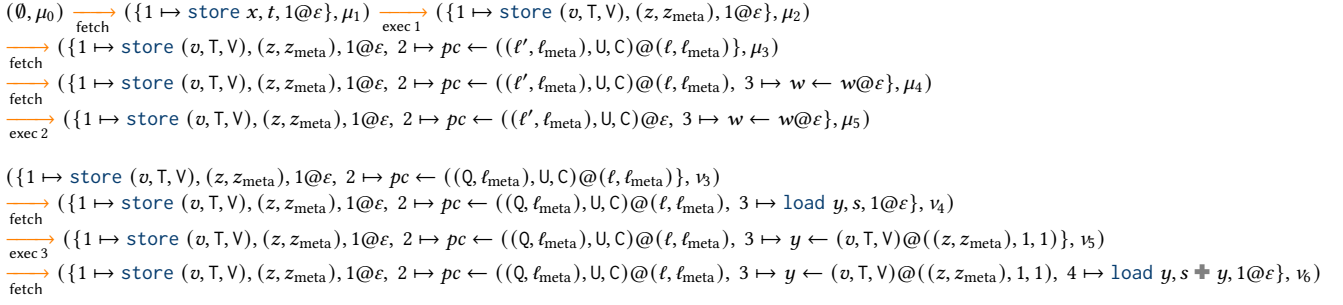

  \centering
  \small
  \[
    \begin{array}{l}
      \cf{\emptyset,\mu_0}
      \sstep{\dfetch}{} \cf{\{1 \mapsto \store x t 1 \specat{\varepsilon}\},\mu_1}
      \sstep{\dexecute 1}{} \cf{\{1 \mapsto \store{(\val,\tainted,\tval)}{\stdCap{z}}1 \specat{\varepsilon}\},\mu_2}\\
      \sstep{\dfetch}{} \cf{\{1 \mapsto \store{(\val,\tainted,\tval)}{\stdCap{z}}1 \specat{\varepsilon},\ 2 \mapsto \pc \gets ((\addr',\stdCapMeta{\addr}), \untainted, \tcap) \specat{\stdCap{\addr}}\},\mu_3}\\
      \sstep{\dfetch}{} \cf{\{1 \mapsto \store{(\val,\tainted,\tval)}{\stdCap{z}}1 \specat{\varepsilon},\ 2 \mapsto \pc \gets ((\addr',\stdCapMeta{\addr}), \untainted, \tcap) \specat{\stdCap{\addr}},\ 3 \mapsto w \gets w \specat{\varepsilon}\},\mu_4}\\
      \sstep{\dexecute 2}{} \cf{\{1 \mapsto \store{(\val,\tainted,\tval)}{\stdCap{z}}1 \specat{\varepsilon},\ 2 \mapsto \pc \gets ((\addr',\stdCapMeta{\addr}), \untainted, \tcap) \specat{\varepsilon},\ 3 \mapsto w \gets w \specat{\varepsilon}\},\mu_5}
    \end{array}
\]

\[
  \begin{array}{l}
    \cf{\{1 \mapsto \store{(\val,\tainted,\tval)}{\stdCap{z}}1 \specat{\varepsilon},\ 2 \mapsto \pc \gets ((\mathtt Q,\stdCapMeta{\addr}), \untainted, \tcap) \specat{\stdCap{\addr}}\},\nu_3}\\
    \sstep{\dfetch}{} \cf{\{1 \mapsto \store{(\val,\tainted,\tval)}{\stdCap{z}}1 \specat{\varepsilon},\ 2 \mapsto \pc \gets ((\mathtt Q,\stdCapMeta{\addr}), \untainted, \tcap) \specat{\stdCap{\addr}},\ 3 \mapsto \load y s 1 \specat{\varepsilon}\},\nu_4}\\
    \sstep{\dexecute 3}{} \cf{\{1 \mapsto \store{(\val,\tainted,\tval)}{\stdCap{z}}1 \specat{\varepsilon},\ 2 \mapsto \pc \gets ((\mathtt Q,\stdCapMeta{\addr}), \untainted, \tcap) \specat{\stdCap{\addr}},\ 3 \mapsto y \gets (\val, \tainted, \tval) \specat{(\stdCap{z},1,1)}\},\nu_5}\\
    \sstep{\dfetch}{} \cf{\{1 \mapsto \store{(\val,\tainted,\tval)}{\stdCap{z}}1 \specat{\varepsilon},\ 2 \mapsto \pc \gets ((\mathtt Q,\stdCapMeta{\addr}), \untainted, \tcap) \specat{\stdCap{\addr}},\ 3 \mapsto y \gets (\val, \tainted, \tval) \specat{(\stdCap{z},1,1)},\ 4 \mapsto \load y {s \stdCapOff y} 1 \specat{\varepsilon}\},\nu_6}
  \end{array}
\]
\caption{Two hardware executions of \eqref{eq:iss1ex} from the same
  initial configuration: execution under a correct prediction (top), which matches the architectural execution
  of \Cref{sec:isasem}; execution under the attacker's
  misprediction reaches $\mathtt Q$ (bottom).}
\label{fig:runningex}
\end{figure*}
\paragraph{Running example}
In \Cref{fig:runningex}, we
illustrate the execution in \sys of the vulnerable
program~\eqref{eq:iss1ex}, modeling \attackref{par:exploit-stack}{issue:stack}.
First, the store
$\store x t 1$ is fetched into entry $1$ by
Rule~\ref{hw:fetch-other} and executed by
Rule~\ref{hw:execute-store-ok}, resolving its store capability to $\stdCap{z}$ and value to $(\val, \tainted, \tval)$ without yet updating memory.
Then, the jump $\jmp g$ is fetched by Rule~\ref{hw:fetch-branch-predict-pc}.
Recall that we assume $g$ points to a benign, skip-like target at $\addr'$ performing $w \gets w$.
Under the correct prediction (top), the predicted target is $\stdPredPc\mu = \addr'$, so entry $2$ is speculatively set to $\pc \gets ((\addr', \stdCapMeta{\addr}), \untainted, \tcap)$, tagged with the branch's $\pc$ value $\stdCap\addr$.
Next, the benign $w \gets w$ is fetched into entry $3$ by Rule~\ref{hw:fetch-other}.
Finally, $\dexecute 2$ resolves the jump
(Rule~\ref{hw:execute-jmp-ok}): the evaluated target matches the
prediction, so entry $2$'s speculation tag is cleared, and the run
reproduces the architectural execution.

Under the attacker's misprediction
$\stdPredPc\mu = \mathtt Q$
(bottom), the store and jump execute as above
except that entry $2$ records the mispredicted target, steering
speculation into the gadget $\mathtt Q$. The first load $\load y s 1$
is fetched into entry $3$ by Rule~\ref{hw:fetch-other} and executed
by Rule~\ref{hw:execute-load-fwd}.
The load capability is resolved to $\stdCap{z}$, aliasing the store capability at entry $1$,
so the store value is forwarded to the load, yielding
$y \gets (\val, {\color{red}\tainted}, \tval)$.
The instruction $\load y {s \stdCapOff y} 1$ is fetched into entry $4$ by
Rule~\ref{hw:fetch-other}, but $\dexecute 4$ is \emph{disabled}: its
address evaluates to
$\sem{s \stdCapOff y}_{\aplsan (\pref 4 \buf, \reg)} = \bot$ because
$\aplsan$ drops the tainted $y$.

\section{Correctness and Security of \sys}
\label{sec:security}

In this section, we establish the two main properties of \sys, its correctness and security.
\emph{Correctness} ensures that every evaluation that is performed at the hardware level can be simulated at the ISA level.
\emph{Security} states that \sys is \emph{relative speculative constant time}. Intuitively, this property guarantees
that if a program does not leak secret data at the ISA level, it also
does not leak secret data under the hardware semantics of \sys.  For
brevity's sake, proofs and formal definitions are deferred to
\iftechreport
\Cref{sec:proofs}.
\else
\cite{scheri-proof}.
\fi

\paragraph{Correctness of \sys} The following result
establishes that \sys is correct with respect to its ISA-level
semantics.

\begin{theorem}[Functional Correctness]
  \label{thm:fc}
  Let $S = (\mem, \reg)$ and $C =(\mem, \reg, \buf, \mu)$
  be initial ISA and HW configurations, respectively.
  For every $n\in \Nat$, and configuration $C' = (\mem', \reg', \buf', \mu')$ such that
  \(
    C \step {}^n C',
  \)
  there exist $t\in \Nat$ and  $O \in \Obs^*$ such that $S \step{O}\!{}^{t} (\mem', \reg')$.
\end{theorem}

The proof of \Cref{thm:fc} is in
\iftechreport
\Cref{sec:fc-proof},
\else
\cite{scheri-proof},
\fi
and is carried out by induction on the length of the execution trace.

\paragraph{Security of \sys}
Our main result shows that the hardware semantics preserves ISA-level
security guarantees in the presence of speculation. Formally, we
require that hardware configurations that contain the same public data
and produce the same ISA-level leakage also produce the same leakage at
the hardware level.
We first define two ISA states $(\mem_0, \reg_0)$ and $(\mem_1, \reg_1)$ containing the same public values,
written $(\mem_0, \reg_0)\pubeq (\mem_1, \reg_1)$:
untainted values coincide in the two states.
Next, we define how two hardware configurations can leak
the same data.  Since the hardware-level semantics updates the
microarchitectural contexts at each transition with leaked values, it
suffices to require the identity of the target microarchitectural
contexts at each transition step. More formally, we say that two
hardware states
$C_0 = (\mem_0, \reg_0, \buf_0, \mu), C_1 = (\mem_1, \reg_1, \buf_1,
\mu)$ produce the same leakage traces (written
$C_0 \equiv_{\mathit{HW}} C_1$) whenever for every $n \in \Nat$ and
adversarial context $\mu'$, we have

\footnotesize
\[
  \exists \mem_0', \reg_0', \buf_0'. C_0 \sstep{}\!{}^{n}\, (\mem_0', \reg_0', \buf_0', \mu') \
  \Leftrightarrow\ \exists \mem_1', \reg_1', \buf_1'. C_1 \sstep{}\!{}^{n}\,
  (\mem_1', \reg_1', \buf_1', \mu').
\]
\normalsize

Note that this condition imposes the identity
of the target microarchitectural states that are updated at each step
with microarchitectural leakage.  With these definitions at hand,
security of \sys can now be stated as follows:

\begin{theorem*}[Relative Speculative Constant Time]{thm:sct}
Let
\begin{align*}
  S_0 & = \cf{\mem_0, \reg_0} & C_0 & = \cf{\mem_0, \reg_0, \bufinit, \mu} \\
  S_0' & = \cf{\mem_0', \reg_0'} & C_0' & = \cf{\mem_0', \reg_0', \bufinit, \mu}
\end{align*}
be initial ISA and HW configurations where $\dom \bufinit = \emptyset$.
If $S_0 \pubeq S_0'$ and $S_0 \equiv_{\mathit{ISA}}S_0'$, then we
have $C_0 \equiv_{\mathit{HW}} C_0'$.
\end{theorem*}

The proof is carried out by induction on
the length of the execution trace and relies on key invariants.
The main one is the preservation of public equivalence at the hardware level,
established via the functional correctness of \sys (\Cref{thm:fc}),
which helps relate leakage at the ISA level and the HW level.
Another is the absence of memory
regions referenced simultaneously by tainted and untainted
capabilities, which prevents using untainted capabilities to access
and leak tainted data.

\section{Comparison with Other Designs}
\label{sec:comparison}

In this section, we compare \sys with different designs. To this end,
we start in \Cref{sec:modelingtooba} by modeling CHERI-Toooba's design
issues outlined
in~\cite{fuchsMasterThesis,fuchs2024csc,fuchsPhdThesis}, and showing
that the resulting system does not satisfy the CSC contract. We then
model \blackout~\cite{blackout} formally in
\Cref{sec:modelingbo}, and show that the
resulting system can leak secrets through microarchitectural
side channels.

\subsection{Modeling CHERI-Toooba's CSC Violations}
\label{sec:modelingtooba}

Fuchs et al.~\cite{fuchsMasterThesis,fuchs2024csc,fuchsPhdThesis} describe several
violations of the CSC contract in CHERI-Toooba.
These vulnerabilities fall into two categories:
violations caused by control-flow speculation crossing the program counter capability boundaries,
and violations arising from capability manipulation,
where a modified capability can be used transiently before the corresponding monotonicity checks are enforced.
In both cases, transient execution may perform memory accesses that violate capability bounds or permissions.
In the following, we demonstrate that our formal semantics can capture these classes of vulnerabilities with only minimal modifications.

\paragraph{Speculative control-flow violation}
This vulnerability~\cite{fuchsMasterThesis} occurs when the
branch target predictor predicts a jump target that is not allowed by
the current program counter capability, violating the CSC contract. Such a
vulnerability can be modeled by modifying
Rule~\ref{hw:fetch-branch-predict-pc} as follows:
\[
\Infer[hw][fetch-branch-predict-pc-cap]{
  \begin{array}[b]{c}
  (\stdCap{\addr}, t, c)=\sem{\pc}_{\aplsan (\buf, \reg)}\quad
  \instr \in \qty{\beqz x \addr'', \jmp e}\\
  \instr = \stdReadMemInst{\mem, (\stdCap{\addr}, t, c)}\hfill
  i = \sup\dom{\buf}\\
  \buf' = \buf[i + 1\mapsto \pc\gets\issuemath{\stdPredPcCap{\mu}} \specat{\stdCap{\addr}}]\\
  \end{array}
}{
  \cf{\mem, \reg, \buf, \mu}
  \sstep{\dfetch}{}
  \cf{\mem, \reg, \buf', \stdUpdateMu{\mu, \sem{\pc}_{\aplsan (\buf, \reg)}}}
}
\]
The main difference between this rule and
Rule~\ref{hw:fetch-branch-predict-pc} lies in the
$\stdPredPcCap{\mu}$ function highlighted in red: while in our system, the function
predicts the offset in the $\pc$ register, here it can predict the
whole $\pc$ capability, potentially violating the architectural
boundaries of the $\pc$, and thereby the CSC contract (and also our security property).

\paragraph{Speculative capability manipulation} This class of
vulnerability~\cite{fuchs2024csc,fuchsPhdThesis} is caused
by race conditions in the hardware between capability usage and validation that can
occur with operations modifying capability offsets
or, more generally, constructing capabilities from existing ones. Specifically,
the constructed capability can be used during transient execution
before the hardware ensures that such a capability has smaller range
and permissions than an existing one, thus violating the CSC contract.
To model these vulnerabilities,
we extend the language of expressions with an additional operator $\stdBuildCap{\cdot, \cdot}$
corresponding to CHERI's \texttt{cbuildcap} instruction.
It has the following interpretation:

\small
\begin{align*}
\sem{\stdBuildCap{e_1, e_2}}_\reg & \defsym
\begin{cases}
\bot & \text{if } \sem{e_1}_\reg = \bot \lor \sem{e_2}_\reg = \bot \\
(v_2, \taint_1 \sqcup \taint_2, \tcap) &
\text{if }\sem{e_1}_\reg = (v_1, \taint_1, \tcap)\\
& \land \sem{e_2}_\reg = (v_2, \taint_2, \tval)\\
& \land v_2 \sqsubseteq v_1\\
(0, \untainted, \tval) & \text{otherwise}.
\end{cases}
\end{align*}
\normalsize

The primitive satisfies the axioms for capability operators in \Cref{sec:exprsem}.
First, the interpretation above satisfies Axiom~(\ref{ax:capop1}) by construction.
Second, Axioms~(\ref{ax:capop2}) and (\ref{ax:capop3}) are satisfied given the following definition of
\iftechreport
$v_2\sqsubseteq v_1$, for $v_1, v_2\in \Cap$ (\Cref{def:cap-mono}):
\else
$v_2\sqsubseteq v_1$, for $v_1, v_2\in \Cap$~\cite{scheri-proof}:
\fi
it requires monotonicity of permissions and bounds, as well as equality of the inner taints of $v_1$ and $v_2$.

The primitive is sufficiently expressive to capture arbitrary monotone capability manipulations:
$e_1$ is a valid capability derived in accordance with the capability monotonicity requirement,
which we call the \emph{guarantor} capability;
$e_2$ expresses the arithmetic operations to build the binary representation of the new capability.

To model the race conditions that occur in CHERI-Toooba with capability manipulation,
we extend the HW semantics with a rule for executing an instruction that builds a capability.
\[
\Infer[hw][execute-forge]{
  \begin{array}[b]{c}
  \buf(i) = x\gets \stdBuildCap{e_1, e_2} \specat{\varepsilon}\quad
  x \neq \pc\\
  (v_1, \taint_1, \tcap) = \sem{e_1}_{\apl (\pref i \buf, \reg)}\quad
  (v_2, \taint_2, \tval) = \sem{e_2}_{\apl (\pref i \buf, \reg)}
  \end{array}
}{
  \cf{\mem, \reg, \buf, \mu}
  \sstep{\dexecute{i}}{}
  \cf{\mem, \reg, \buf[i \mapsto x \gets (v_2, \taint_1 \sqcup \taint_2, \tcap) \specat{(v_1, v_2)}], \mu}
}
\]
In this rule, the value of $e_2$ is promoted to a capability
without enforcing the monotonicity constraints between $e_1$ and $e_2$;
these are instead deferred to commit time, as specified by the next rule:

\small
\[
  \Infer[hw][commit-forge]{
    i = \min\dom{\buf}&
    x\neq\pc&
    \buf(i) = x \gets (v_2, \taint, \tcap) \specat{(v_1, v_2)}&
    v_2 \sqsubseteq v_1
  }{
    \cf{\mem, \reg, \buf, \mu}
    \sstep{\dcommit}{}
    \cf{\mem, \reg[\pc \mapsto \sem{\stdNextPc{\pc}}_\reg][x\mapsto (v_2, \taint, \tcap)], \buf \setminus i, \mu}
  }
\]
\normalsize

Here, the rule uses the speculation tag $(v_1, v_2)$---corresponding to the values of the \emph{guarantor} and the new capability---to check whether the new capability satisfies the monotonicity constraints.

By leveraging Rule~\ref{hw:execute-forge}, an attacker can issue a
$\stdBuildCap$ instruction during speculation, and use the resulting
capability to perform memory accesses that violate the CSC contract.
Such problematic memory accesses would also violate our relative speculative constant time property, since the resulting capability may allow accesses to memory locations containing secrets.

These violations also illustrate how much control an attacker
has over \emph{which} secrets are leaked. Because CHERI-Toooba
admits transient accesses that violate capability bounds, an
attacker can steer speculation towards secrets that are never
accessed architecturally, exactly as in classical
Spectre attacks on non-capability architectures.
A CSC-compliant design reduces
this control: the attacker can only dereference
capabilities reachable from the committed architectural state and
leak the memory they authorize. \sys pushes this even
further, by preventing leakage of \emph{any} secret data.

\subsection{Modeling \blackout}
\label{sec:modelingbo}

To model \blackout \cite{blackout}, there are two largely orthogonal differences to consider.
First, \blackout attempts to eagerly detect violations of the data-oblivious programming model and will generate faults when tainted values are stored to memory through non-blinded capabilities (except in the case of spilling tainted registers into BRRs, see below).
In our formal model, this corresponds to an alternative version of the rules in \Cref{fig:blackout-mem-op-semantics}.
The additional condition $t \sqsubseteq \taint_c$ in rule \ref{helpersys:blackout-write-data} (highlighted) will make execution get stuck when a tainted value is written to memory through a non-blinded capability, modeling the exception that \blackout{} would generate.
Additionally, the extra conditions in rules \ref{memory:blackout-read-mem} and \ref{memory:blackout-write-mem} (highlighted) will make execution get stuck when memory is accessed through a tainted capability.

This eager detection of data-obliviousness violations is a consequence of the different threat model targeted by both systems.
As explained in \Cref{sec:blackout}, \blackout{} strives to enforce speculative constant time under five assumptions on user code.
The check $t \sqsubseteq \taint_c$
in \Cref{fig:blackout-mem-op-semantics} enforces the first (I1) of these assumptions, while other assumptions only hold for code generated by a trusted compiler.
Our threat model is different: for assembly code that is constant time (i.e., data oblivious) in the architectural semantics, we ensure that speculation will not create new leaks and preserves constant-time behavior.
As such, we make no attempt to enforce architectural security of the code and restrict ourselves to correctly propagating but not enforcing taint in the architectural semantics.

\begin{figure}
  \centering
  \small
  \begin{rules}[Eager detection of data-obliviousness violations]
    \Infer[helpersys][blackout-write-data]{
      \sz = \szAddr & \issuemath{\taint \sqsubseteq \taint_c}
    }{
      \stdWriteRes{(v, \taint, c), \taint_c, \sz} = (v, \tval)
    }\\
    \Infer[memory][blackout-read-mem]{
      \begin{array}[b]{c}
        \stdCheckCap{(\addr, (\perm, \rbase, \rend, \taint_c)), \sz, \readonly}\\
        \stdReadRes{\memData[\addr, \addr + \sz), \taint_c, \memTag[\addr / \szCap], \sz} = (v, t, c)
      \end{array}
    }{
      \stdReadMem{(\memData, \memTag), ((\addr, (\perm, \rbase, \rend, \taint_c)), \issuemath{\untainted}, \tcap), \sz} = (v, t, c)
    }\\
    \Infer[memory][blackout-write-mem]{
      \begin{array}[b]{c}
        \stdCheckCap{(\addr, (\perm, \rbase, \rend, \taint_c)), \sz, \readwrite}\\
        \stdWriteRes{(v, t, c), \taint_c, \sz} = (v', c')\\
        \mem' = (\memData[[\addr, \addr + \sz) \mapsto v'], \memTag[\addr/\szCap \mapsto c'])
      \end{array}
    }{
      \stdWriteMem{(\memData, \memTag), ((\addr, (\perm, \rbase, \rend, \taint_c)), \issuemath{\untainted}, \tcap), \sz, (v, t, c)} = \mem'
    }
  \end{rules}
  \begin{rules}[Treatment of BRRs]
    \Infer[helpersys][blackout-unspill-value]{
      v_m = \stdMagic &
      c = \tcap
    }{
      \stdReadRes{(v, v_m), \taint_c, c, \szCap} = ((v, v_m), \issuemath{\tainted}, \tval)
    }\\
    \Infer[helpersys][blackout-spill-value]{
      (v_m', c') = (\taint = \issuemath{\tainted}) \;?\; (\stdMagic, \tcap) : (v_m, c)
    }{
      \stdWriteRes{((v, v_m), \taint, c), \taint_c, \szCap} = ((v, v_m'), c')
    }
  \end{rules}
  \caption{Rules modeling eager detection of data-obliviousness violations and BRRs in \blackout.}
  \label{fig:blackout-mem-op-semantics}
\end{figure}

The second important change with respect to \blackout{} is that we adopt a treatment of the stack that is \emph{dual} to \blackout{}'s.
We treat the stack capability as always tainted,
and mark spills of \emph{public} values using a distinguished encoding, URR,
that sets the capability tag and uses a reserved marker value in the metadata field.
\Cref{fig:blackout-mem-op-semantics} shows modified rules that revert to \blackout{}'s original treatment of BRRs.
A consequence of \blackout{}'s design is that secrets are stored in memory that is accessible through non-blinded capabilities, only guarded by the BRR's capability tag and magic value.
To preserve the tainted status of such values, it is crucial to enforce atomicity of BRR writes (the magic value or tag must not be overwritten without erasing the secret value) and to ensure all reads (including partial reads) of the BRR secret value taint the result.
\blackout{} does not enforce these two properties in hardware, but instead relies on invariants of software.
\Cref{sec:motivating-examples} shows that constructing software that satisfies these invariants is highly non-trivial, especially because the invariants must hold in transient execution as well as architecturally.

By reversing the treatment of BRRs, atomicity of URR writes becomes less important: if the magic value or the tag is overwritten, then an untainted value will be treated as tainted, which only risks degrading performance.
Similarly, if a URR value is read using an instruction (e.g., a partial read) that does not correctly take into account the URR, this will also only result in unnecessarily
marking the value that has been read as tainted.
Moreover, URRs are unforgeable, just like \cheri{}'s capabilities.
A URR is created only when the hardware sets the capability tag and the magic value while spilling a public value.
Since software cannot set capability tags itself (a property inherited from \cheri{}), writing the magic value alone leaves the tag clear and yields no URR.
Conversely, writing a secret over an existing URR clears its tag and thereby destroys it.
Hence, \sys{} guarantees that URRs hold public values, and achieves sound taint tracking without any assumption on software,
preventing \blackout{}'s Issue~\ref{issue:brr-all}.

\blackout narrows the attacker's control on speculative memory accesses compared to CSC-compliant systems, but less
than \sys. In \blackout, only secrets that are reachable through a
\emph{non-blinded} capability can leak.
However, \blackout keeps the stack capability non-blinded,
while still allowing programs to create secret stack variables,
as long as these variables are accessed only through blinded capabilities during non-speculative execution.
As a result, an attacker who controls the stack capability's offset can speculatively reach a leak gadget (such as the one in Issue~\ref{issue:stack}) and leak \emph{arbitrary} stack secrets.
Such leakage is impossible in \sys{}, since it avoids any overlap between blinded and non-blinded capabilities, preventing \blackout{}'s Issue~\ref{issue:overlap}.

\section{Related Work}
\label{sec:rw}

Since CSC~\cite{fuchs2024csc} and BLACKOUT \cite{blackout} are discussed extensively throughout the paper, we focus this section on other work.

\paragraph{Formal models of transient execution.}
A large body of prior work develops formal semantics for speculative processors to detect or prove absence of transient leaks. These cover variants~\cite{surveytransient,machineclears} such as Spectre-PHT, Spectre-BTB, Spectre-STL, and have enabled tools for leak detection, symbolic analysis, and formal noninterference proofs~\cite{guarnieri2020spectector,CTfoundations,catspectre,binsechaunted,HighAssurance,survey}.
Our semantics builds on this line of work, but extends it with capability state, capability provenance, and capability-specific authorization checks that are central in CHERI-like systems.
A foundational direction~\cite{guarnieri2021hardware} formalizes security guarantees as contracts between hardware and software. Several prior works characterize the attacker-visible behavior of speculative processors and show how software guarantees must be matched to hardware behavior to obtain end-to-end confidentiality. This contract-based view is closely aligned with our approach.

\paragraph{Relative security}
Prior works~\cite{cheang2019formal,guanciale2020inspectre,dongol2024relative} introduce general security notions to model speculative execution vulnerabilities, which relate execution traces produced with and without speculation to determine whether speculation introduces additional leakage.
We compare our security property (\Cref{thm:sct}) against the state-of-the-art notion in \cite{dongol2024relative}.
First, our threat model is similar to theirs: both include control-flow and memory-access addresses in the attacker's observation.
We additionally leak the value written by a store whose capability has untainted inner taint, i.e., when storing to public memory regions.
We adopt a simpler setup,
not modeling secret input introduced during execution or interactive attacker actions.
Second, we require an additional assumption of public equivalence on initial states,
which states that all secrets in the initial state are already properly tracked by taint bits on registers and capabilities.
Third, the security property of \cite{dongol2024relative}, instantiated to our setting, reads:
if the attacker observes a difference between a pair of HW traces, then there must exist a pair of ISA traces that carry the same secrets as the HW traces (respectively) and are also attacker-distinguishable.
\Cref{thm:sct} is a stronger, less general version of this statement: rather than asserting the existence of such ISA traces, we construct them by executing the ISA semantics from the same initial ISA state as the corresponding HW trace.

\paragraph{Provably secure speculation for constant time.}
Several works study how hardware support can maintain the classic constant-time programming model during speculative execution. ProSpeCT~\cite{prospect} develops a processor model with secrecy tracking in the pipeline and proves that constant-time software remains secure under speculative and out-of-order execution across a range of predictors and speculation policies.
Importantly, it assumes a fixed partitioning of memory into secret and public regions,  which forces the stack to reside entirely in secret or public memory, with complementary efficiency-security trade-offs. \sys instead integrates secrecy into CHERI's capability system: capabilities keep track of the data's secrecy level, which is propagated through capability derivation, and preserved across spills and restores. This design choice enables secret and public objects to coexist within the same stack, but requires a new design of capability metadata, stack management (URR), and information-flow issues that do not arise in ProSpeCT.

More broadly, taint-based and delayed-transmitter processor designs demonstrate that dynamic secrecy tracking or speculative shadow structures can preserve performance while blocking Spectre-style leaks~\cite{nda,context,stt,spt,prospect}. Our work shares this goal, but addresses the additional challenges introduced by capabilities, capability-derived pointers, and capability-mediated memory accesses.
Relative to this line of work, our approach is closer to BLACKOUT, but BLACKOUT does not provide security proofs. To the best of our knowledge, our work is the first one to provide provably secure speculation for constant-time in CHERI.

\paragraph{Machine-level enforcement after compilation.}
A related line of work observes that source-level security guarantees may be invalidated during compilation, register allocation, and stack layout. Recent work such as SecSep~\cite{secsep} addresses this problem for speculative side-channel defenses by rewriting compiled assembly code to separate secret and public data after compilation, thereby accounting for register spills, stack slots, and other low-level effects that are difficult to control at the source level.
SecSep improves  how ProSpeCT~\cite{prospect}  handles secrets in memory.
 This perspective closely aligns with our emphasis on machine-level reasoning: in capability systems, speculative leaks may similarly arise from transient interactions with spilled values, stale stack contents, or lowered memory accesses that are invisible in higher-level models. Our work differs in targeting CHERI-like capability machines and providing formal secure-speculation guarantees for capability-aware hardware semantics rather than compiler rewriting alone.

\section{Conclusion}
\label{sec:conclusion}

Architectural isolation in CHERI does not guarantee confidentiality under speculation. We showed that secure speculation for capability machines requires reasoning jointly about authorization and information flow, and presented SCHERI with formal guarantees for preserving constant-time behavior under speculation.
Our work provides a foundation for future capability systems that  should remain secure not only architecturally, but also microarchitecturally.

\begin{acks}
This research is partially funded by the Air Force Office of Scientific Research (AFOSR) under grant FA9550-22-1-0511, the Internal Funds KU Leuven, the Cybersecurity Research Program Flanders, the Research Foundation -- Flanders (FWO) under grant number G081322N,
and a European Research Council (ERC) Starting Grant (UniversalContracts; 101040088) funded by the European Union.
Views and opinions expressed are those of the authors only and do not necessarily reflect those of the European Union or the European Research Council.
\end{acks}

\bibliographystyle{ACM-Reference-Format}
\bibliography{paper}

\appendix 

\crefalias{section}{appendix}
\crefalias{subsection}{appendix}
\crefalias{subsubsection}{appendix}

\section{Open Science} 
\label{sec:openscience}

The code of our exploits for the \blackout prototype is released at \url{https://doi.org/10.5281/zenodo.22755930}.
We additionally reported issues in the hardware,\footnote{\url{https://github.com/blindedcapabilities/blinded-toooba/issues/1}} allowing BRR writes through unblinded capabilities, and in the compiler,\footnote{\url{https://github.com/blindedcapabilities/blinded-cheri-llvm/issues/1}} which does not set blindedness properly for struct fields.



\section{Ethical Considerations} 


We responsibly disclosed all identified issues to the authors of \blackout prior to submission.
As it is a research prototype, no users are directly affected by this research.
We believe that our work improves the understanding and security of speculative execution in capability systems and contributes to the design of more secure capability-based processors.
These considerations motivated our decision to submit this work.

\iftechreport
\section{Full hardware semantics}
\label{sec:apprules}

\begin{figure*}
  \small
  \columnwidth=\linewidth
\centering
\begin{rules}[Step rule]
\vspace{-2mm}
  \Infer*[hw][step]{
  \mu'= \stdUpdateMu{\mu, \lowproj{\buf}}&
  d = \stdNextMu{\mu'}&
  \cf{\mem, \reg, \buf, \mu'} \sstep{d}{} \cf{\mem', \reg', \buf', \mu''}
}{
  \cf{\mem, \reg, \buf, \mu} \sstep{}{} \cf{\mem', \reg', \buf', \mu''}
}
\end{rules}
\vspace{-2mm}
\begin{rules}[Rules for fetch steps]
\Infer*[hw][fetch-branch-predict-pc]{
  \begin{array}[b]{c}
  (\stdCap{\addr}, t, c)=\sem{\pc}_{\aplsan (\buf, \reg)}\hfill
  \instr \in \qty{\beqz x \addr'', \jmp e}\\
  \instr = \stdReadMemInst{\mem, (\stdCap{\addr}, t, c)}\hfill
  i = \sup\dom{\buf}\\
  \buf' = \buf[i + 1\mapsto \pc\gets((\stdPredPc{\mu}, \stdCapMeta{\addr}), t, c)\specat{\stdCap{\addr}}]\\
  \end{array}
}{
  \cf{\mem, \reg, \buf, \mu}
  \sstep{\dfetch}{}
  \cf{\mem, \reg, \buf', \stdUpdateMu{\mu, \sem{\pc}_{\aplsan (\buf, \reg)}}}
}
\Infer*[hw][fetch-other]{
  \begin{array}[b]{c}
  \instr = \stdReadMemInst{\mem, \sem{\pc}_{\aplsan (\buf, \reg)}}\\
  \instr \not\in \qty{\beqz x \addr'', \jmp e}\\
  \buf' = \buf[\sup\dom{\buf} + 1\mapsto \instr\specat{\varepsilon}]
  \end{array}
}{
  \cf{\mem, \reg, \buf, \mu}
  \sstep{\dfetch}{}
  \cf{\mem, \reg, \buf', \stdUpdateMu{\mu, \sem{\pc}_{\aplsan (\buf, \reg)}}}
}
\end{rules}
\vspace{-2mm}
\begin{rules}[Rules for assignments]
  \Infer*[hw][execute-assign]{
  \buf(i) = x\gets e \specat{\varepsilon}&
  x \neq \pc&
  (v, t, c) = \sem{e}_{\apl (\pref i \buf, \reg)}
}{
  \cf{\mem, \reg, \buf, \mu}
  \sstep{\dexecute{i}}{}
  \cf{\mem, \reg, \buf[i \mapsto x \gets (v, t, c) \specat{\varepsilon}], \mu}
}
\\
\Infer*[hw][execute-jmp-ok]{
  \begin{array}[b]{c}
  \buf(i) = \pc\gets \addr' \specat{\stdCap{\addr}}\hfill
  \buf' = \buf[i\mapsto \pc\gets \addr' \specat{\varepsilon}]~~
  \addr' = (\addr_0, \_, \tcap)\\
  \jmp{e} = \stdReadMemInst{m, \sem{\pc}_{\aplsan(\pref i \buf, \reg)}}~~
  (\addr_0, \_, \tcap) = \sem{e}_{\aplsan (\pref i \buf, r)}\\
  \end{array}
}{
  \cf{\mem, \reg, \buf, \mu}
  \sstep{\dexecute{i}}{}
  \cf{\mem, \reg, \buf', \stdUpdateMu{\mu, \addr_0}}
}
\Infer*[hw][execute-jmp-hazard]{
  \begin{array}[b]{c}
  \buf(i) = \pc\gets (\addr_1, \_, \_)\specat{\stdCap{\addr}}~~
    \jmp{e} = \stdReadMemInst{m, \sem{\pc}_{\aplsan(\pref i \buf, \reg)}} \\
    \addr_1 \neq \addr_0 \hfill (\addr_0, \_, \tcap) = \sem{e}_{\aplsan (\pref i \buf, r)} \hfill
  \buf' = \pref {i + 1} \buf[i\mapsto \pc\gets (\addr_0, \textcolor{red}{\untainted}, \tcap) \specat{\varepsilon}]
  \end{array}
}{
  \cf{\mem, \reg, \buf, \mu}
  \sstep{\dexecute{i}}{}
  \cf{\mem, \reg, \buf', \stdUpdateMu{\mu, \addr_0}}
}
\\
\Infer[hw][execute-beqz-ok]{
  \begin{array}[b]{c}
    \beqz{x}{\addr''} = \stdReadMemInst{m, \sem{\pc}_{\aplsan(\pref i \buf, \reg)}}\\
    \buf(i) = \pc\gets \addr' \specat{\stdCap{\addr}}\hfill \mathit{z}{} = (v = 0) \;?\; \addr'': 1 \\
     (v, t, c) = \sem{x}_{\aplsan (\pref i \buf, r)}
    \quad
  \addr' = \sem{\pc \stdCapOff {\mathit{z}}}_{\apl(\pref i \buf, \reg)}
  \end{array}
}{
  \cf{\mem, \reg, \buf, \mu}
  \sstep{\dexecute{i}}{}
  \cf{\mem, \reg, \buf[i\mapsto \pc\gets \addr' \specat{\varepsilon}] , \stdUpdateMu{\mu, (v=0)}}
}
\Infer[hw][execute-beqz-hazard]{
  \begin{array}[b]{c}
    \beqz{x}{\addr''} = \stdReadMemInst{m, \sem{\pc}_{\aplsan(\pref i \buf, \reg)}}\\
    \buf(i) = \pc\gets \addr' \specat{\stdCap{\addr}}
    \hfill \mathit{z}{} = (v = 0) \;?\; \addr'': 1 \\
     (v, t, c) = \sem{x}_{\aplsan (\pref i \buf, r)} \quad
  \addr' \neq \addr_0 = \sem{\pc \stdCapOff {\mathit{z}}}_{\apl(\pref i \buf, \reg)}
  \end{array}
}{
  \cf{\mem, \reg, \buf, \mu}
  \sstep{\dexecute{i}}{}
  \cf{\mem, \reg, \pref {i + 1} \buf[i\mapsto \pc\gets \addr_0 \specat{\varepsilon}], \stdUpdateMu{\mu, (v=0)}}
}
\end{rules}
\vspace{-2mm}
\begin{rules}[Rules for executing loads and stores]
\Infer[hw][execute-load-mem]{
  \begin{array}[b]{c}
    \buf(i) = \load{x}{e}{\sz} \specat{\varepsilon}\hfill
    x\neq \pc
    \hfill
    ((\addr, (\perm, \rbase, \rend, \taint_c)), \_, \tcap) = \sem{e}_{\aplsan (\pref i \buf, \reg)}
    \hfill \stdCheckCap{(\addr, (\perm, \rbase, \rend, \taint_c)), \sz, \readonly} \\
    (\forall j < i,\; \buf(j) = \store{x'}{\stdCap{\addr'}}{\sz'}\specat{\tg}
    \Rightarrow [\addr, \addr + \sz) \cap [\addr', \addr' + \sz') = \emptyset)\quad
    (v, t, c) = \stdReadRes{\memData[\addr, \addr + \sz], \taint_c, \memTag[\addr / \szCap], \sz}
    \\
  \end{array}
}{
  \cf{(\memData, \memTag), \reg, \buf, \mu}
  \sstep{\dexecute{i}}{}
  \cf{(\memData, \memTag), \reg, \buf[i \mapsto (x\gets (v, t, c)\specat{((\addr, (\perm, \rbase, \rend, \taint_c)), \sz, \bot)})], \stdUpdateMu{\mu, (\addr, (\perm, \rbase, \rend, \taint_c))}}
}\\
\Infer*[hw][execute-load-fwd]{
  \begin{array}[b]{c}
    \buf(i) = \load{x}{e}{\sz} \specat{\varepsilon}\hfill
    x\neq \pc \hfill
    [\addr, \addr + \sz) = [\addr', \addr' + \sz')  \hfill
    ((\addr, (\perm, \rbase, \rend, t_c)), \_, \tcap) = \sem{e}_{\aplsan (\pref i \buf, \reg)} \hfill
    \stdCheckCap{(\addr, (\perm, \rbase, \rend, t_c)), \sz, \readonly}\\
    (j = \max \lbrace
    j < i: \buf(j) = \store{\_}{\stdCap{\addr'}}{\sz'}\specat{\tg}
    \land [\addr, \addr + \sz) \cap [\addr', \addr' + \sz') \neq \emptyset
  \rbrace)
    \quad
    t_0 \sqsubseteq t_c \quad
  \buf(j) = \store{(v_0, t_0, c_0)}{\stdCap{\addr'}}{\sz'}\specat{\tg}
    \\
  \end{array}
}{
  \cf{\mem, \reg, \buf, \mu}
  \sstep{\dexecute{i}}{}
  \cf{\mem, \reg, \buf[i \mapsto (x\gets \stdReadRes{v_0, t_c, c_0, \sz}\specat{(\stdCap{\addr}, \sz, j)})], \stdUpdateMu{\mu, (\addr, (\perm, \rbase, \rend, t_c))}}
}\\
\Infer*[hw][execute-store-ok]{
  \begin{array}[b]{c}
  \buf(i) = \store{x}{e}{\sz}\specat{\varepsilon}\hfill
  ((\addr, (\perm, \rbase, \rend, \taint_c)), \_, \tcap) = \sem{e}_{\aplsan (\pref i \buf, \reg)} \hfill
  \stdCheckCap{(\addr, (\perm, \rbase, \rend, \taint_c)), \sz, \readwrite}\hfill
  (v, t, c) = \sem{x}_{\apl(\pref i \buf, \reg)}\\
  (\forall j > i, \buf(j) = x \gets v \specat{(\stdCap{\addr'}, \sz', k)} \land k \neq i \Rightarrow
  (k > i \lor [\addr, \addr + \sz) \cap [\addr', \addr + \sz') = \emptyset))\quad
  (v', c') = \stdWriteRes{(v, t, c), \taint_c, \sz}\\
  \end{array}
}{
  \cf{\mem, \reg, \buf, \mu}
  \sstep{\dexecute{i}}{}
  \cf{\mem, \reg, \buf[i \mapsto (\store{(v', t, c')}{(\addr, (\perm, \rbase, \rend, \taint_c))}{\sz})\specat{\varepsilon}], \stdUpdateMu{\mu, (\addr, (\perm, \rbase, \rend, \taint_c))}}
}
\\
\Infer[hw][execute-store-hazard]{
  \begin{array}[b]{c}
  \buf(i) = \store{x}{e}{\sz}\specat{\varepsilon}\hfill
  ((\addr, (\perm, \rbase, \rend, \taint_c)), \_, \tcap) = \sem{e}_{\aplsan (\pref i \buf, \reg)} \hfill
    \stdCheckCap{(\addr, (\perm, \rbase, \rend, \taint_c)), \sz, \readwrite}\hfill
      (v, t, c) = \sem{x}_{\apl(\pref i \buf, \reg)}\\
  (j = \min \lbrace j > i: \buf(j) = x\gets v \specat{(\stdCap{\addr'}, \sz', k)}
  \land (
    k < i \land [\addr, \addr+\sz) \cap [\addr', \addr'+\sz') \neq \emptyset
  ) \rbrace)\quad
  (v', c') = \stdWriteRes{(v, t, c), \taint_c, \sz}\\
  \end{array}
}{
  \cf{\mem, \reg, \buf, \mu}
  \sstep{\dexecute{i}}{}
  \cf{\mem, \reg, \pref j \buf[i \mapsto (\store{(v', t, c')}{(\addr, (\perm, \rbase, \rend, \taint_c))}{\sz})\specat{\varepsilon}], \stdUpdateMu{\mu, (\addr, (\perm, \rbase, \rend, \taint_c))}}
}
\end{rules}
\vspace{-2mm}
\begin{rules}[Rules for commit steps]
\Infer*[hw][commit-assign]{
i = \min\dom{\buf}&
    x = \pc \Rightarrow \tg= \varepsilon &
  \buf(i) = x \gets v \specat{\tg}&
}{
  \cf{\mem, \reg, \buf, \mu}
  \sstep{\dcommit}{}
  \cf{\mem, \reg[\pc \mapsto \sem{\stdNextPc{\pc}}_\reg][x\mapsto v], \buf \setminus i, \mu}
}\\
\Infer*[hw][commit-store]{
  i = \min\dom{\buf}\quad
  \buf(i) = \store{(v, t, c)}{\stdCap{\addr}}{\sz} \specat{\varepsilon}\quad
  \mem' = (\memData[[\addr, \addr + \sz) \mapsto v], \memTag[\addr/\szCap \mapsto c])
}{
  \cf{(\memData, \memTag), \reg, \buf, \mu}
  \sstep{\dcommit}{}
  \cf{\mem', \reg[\pc \mapsto \sem{\stdNextPc{\pc}}_\reg], \buf \setminus i, \stdUpdateMu{\mu, \stdCap{\addr}}}
}
\end{rules}
\vspace{-3mm}
\caption{\sys hardware semantics}
\label{fig:fullhwsem}
\end{figure*}

The full set of rules of \sys's hardware semantics is in \Cref{fig:fullhwsem}.
Many of the rules were already discussed in \Cref{sec:sem}, so we only
comment on those that were not already described.

Rule~\ref{hw:fetch-other} evaluates instructions involving the
$\dfetch$ directive when the program counter does not point to a jump
instruction. The target configuration extends the input buffer by
adding a new entry corresponding to the instruction
\[
  \instr = \stdReadMemInst{\mem, \sem{\pc}_{\aplsan (\buf,
      \reg)}}.
\]
Note that the target microarchitectural context is
updated with the leakage $\sem{\pc}_{\aplsan (\buf, \reg)}$, revealing
to the attacker the current value of the program counter.

Rule~\ref{hw:execute-jmp-hazard} is very similar to
Rule~\ref{hw:execute-jmp-ok}: the target $e$ is evaluated using
register file $\aplsan (\pref i \buf, r)$, and the result is compared
with the speculation tag, recording the jump target in the speculation
tag. This rule applies when the two do not match, and the reorder
buffer of the target configuration is set to
$\buf' = \pref {i + 1} \buf[i\mapsto \pc\gets (\addr_0,
\textcolor{red}{\untainted}, c) \specat{\varepsilon}]$, where the
$\pc$ is updated with the correct jump target and all the following
instructions are dropped.

The rules for evaluating conditional branches are similar to those for
direct jumps. Rules~\ref{hw:execute-beqz-ok}
and \ref{hw:execute-beqz-hazard} evaluate the branch condition
$(v, t, c) = \sem{x}_{\aplsan (\pref i \buf, r)}$ using the sanitized
buffer to prevent transient leakages of secret data. The rules compute
the jump offset in $z$ and the actual jump target
\( \sem{\pc \stdCapOff {\mathit{z}}}_{\apl(\pref i \buf, \reg)} \).
If the jump target matches the one stored in the buffer,
Rule~\ref{hw:execute-beqz-ok} applies and the speculation tag is
cleared. If the two targets do not match,
Rule~\ref{hw:execute-beqz-hazard} applies, which sets the correct jump
target in the ROB, and drops the stale entries. Both rules leak the
branch condition to the microarchitectural context via the
$\stdUpdateMu\cdot$ function.

Rule~\ref{hw:execute-load-mem} is similar to
Rule~\ref{hw:execute-load-fwd} and describes how load instructions are
executed when no aliasing store is found in the ROB. The rule checks
that the target register is not the program counter to ensure that
such register can only be updated via $\kwd{beqz}$ and $\kwd{jmp}$
instructions. Then the load address
$(\stdCap{\addr}, \_, \tcap) = \sem{e}_{\aplsan (\pref i \buf, \reg)}$
is evaluated and capability permissions are checked in the
$\stdCheckCap{\stdCap{\addr}, \sz, \readonly}$ premise. The rule
applies when for every $j < i$, the corresponding entry in the ROB
is not an aliasing store.  If this is the case, the value is fetched
from central memory in premise:
\[
  (v, t, c) = \stdReadRes{\memData[\addr, \addr + \sz], \taint_c, \memTag[\addr / \szCap], \sz}
\]
and the speculation tag is updated to record the loaded region, in
order to resolve aliasing, when executing store instructions.

Rule~\ref{hw:execute-store-hazard} applies if store-to-load aliasing
is detected when resolving a store instruction.
It computes the store capability using the sanitized register
file $\aplsan(\pref i \buf, \reg)$, to ensure that no tainted value
can leak to the microarchitectural context, and checks its permission
via $\stdCheckCap{\cdot, \sz, \readwrite}$. The value to be stored is
obtained by evaluating $x$ under the partially applied register file
$\apl(\pref i \buf, \reg)$, and is recorded in the store buffer.
The rule then computes the index of the first
aliasing load instruction:
\begin{multline*}
  j = \min \lbrace j > i: \buf(j) = x\gets v \specat{(\stdCap{\addr'}, \sz', k)}
  \land \\(  k < i \land [\addr, \addr+\sz) \cap [\addr', \addr'+\sz') \neq \emptyset
  ) \rbrace,
\end{multline*}
and drops all the entries in the reorder buffer at positions equal to or
greater than $j$ by setting the buffer of the target configuration to:
\[
  \pref j \buf[i \mapsto (\store{(v', t, c')}{(\addr, (\perm, \rbase,
    \rend, \taint_c))}{\sz})\specat{\varepsilon}].
\]

Finally, Rule~\ref{hw:commit-store} commits store instructions by
updating the memory with the value taken from the ROB.

\section{Proofs}
\label{sec:proofs}
This section is devoted to proving the main results presented in this
paper, namely functional correctness and security of \sys.

We start in \Cref{sec:isa-proof}, by providing lemmas about architectural (ISA-level) invariants that are helpful to reason about correctness and security of \sys's HW semantics.

First, \sys's HW semantics does not handle self-modifying code,
so we assume that in the initial configuration, there is no overlapping between memory with $\exec$ and $\readwrite$ permission. This property will be maintained as an invariant by monotonicity of capability operations.
Another key invariant of \sys necessary to its security property is that there is no overlapping between tainted and untainted memory.
This is needed since we should not allow using untainted capabilities to access tainted data and leak it, or using tainted capabilities to write secret data to the untainted region while losing protection on it.
These two invariants are formalized in \Cref{def:isa-wf} and proved in \Cref{lemma:isa-wf}.

Second, we formally define public equivalence of architectural states (\Cref{def:pub-eq-isa}) and prove that it is maintained by the ISA semantics (\Cref{lemma:isa-pubeq}).

We then turn our attention to proving functional correctness of \sys's HW semantics,
i.e., that \sys is correct with respect to its ISA-level semantics.
The proof is in \Cref{sec:fc-proof},
and is carried out by induction on the number of HW steps being executed.
The key insight is to define inductive invariants to constrain well-formedness of reorder buffer $\buf$ (\Cref{def:well-formed-buf}).
Specifically, it constrains the expected correlation between each in-flight instruction's execution outcome and its instruction type/operands.
\Cref{lemma:hw-wf-invariant} proves that the invariant holds for every step.

In \Cref{sec:sec-proof}, we prove \sys satisfies relative speculative constant time.
The proof is also done by induction on the number of HW steps being executed.
The key inductive invariant is defined via public equivalence of HW configurations (\Cref{def:pub-eq}), which requires public equivalence on both architectural states ($\mem$, $\reg$) and microarchitectural contexts ($\buf$, $\mu$).
Specifically, proving this requires us to apply \Cref{thm:fc} (functional correctness of \sys) to correlate ISA leakage traces and HW observation traces.

\subsection{ISA-level Proof}
\label{sec:isa-proof}

\subsubsection{ISA Maintains Well-formedness}

\begin{definition}[Capability Monotonicity]
We define capability monotonicity as follows:
\begin{equation*}
  (p,b,e,t) \sqsubseteq (p',b',e',t') \defsym p \sqsubseteq p' \land [b, e) \subseteq [b', e') \land t = t'.
\end{equation*}
We also employ the notation $(\addr, (p, b, e, t))\sqsubseteq (\addr', (p', b', e', t'))$ to represent the case where $(p,b,e,t) \sqsubseteq (p',b',e',t')$.
\label{def:cap-mono}
\end{definition}

\begin{definition}[Reachable Capabilities]
Given a register file $\reg$ and a memory $\mem=(\memData, \memTag)$,
the set of reachable capabilities is defined as follows:
\begin{mathpar}
\Infer[rc][base]{
  r(x) = (\stdCap{\addr}, t, \tcap)
}{
  \stdCap{\addr} \in \rc{\reg, \mem, \vx}
}\\
\Infer[rc][ind]{
  \begin{array}[b]{c}
    (\_, (p, b, e, t)) \in \rc{\reg, \mem, x} \quad [\addr', \addr'+2) \subseteq [b, e) \\
    \memData[\addr', \addr' + \szCap] = (\addr, \addr_m) \quad \memTag[\addr' / 2] = \tcap \quad \addr' \% 2 = 0\\
    \stdUnpackTaint{\addr_m} = (\stdCapMeta{\addr}, \_) \quad \addr_m \neq \stdMagic \\
  \end{array}
}{
  \stdCap{\addr} \in \rc{\reg, \mem, \vx}
}\\
\Infer[rc][shrink]{
  \stdCap{\addr'} \in \rc{\reg, \mem, \vx} \quad
  \stdCap{\addr} \sqsubseteq \stdCap{\addr'}
}{
  \stdCap{\addr} \in \rc{\reg, \mem, \vx}
}
\end{mathpar}
  We also employ the notation $\rc{\reg, \mem}$ for
  \(
  \bigcup_{\vx \in \Reg} \rc{\reg, \mem, \vx}
  \).
\label{def:reach-cap}
\end{definition}

\begin{definition}[Reachable Memory]
  Given a set of capabilities $C$, and a taint $t$, we define the set
  of reachable addresses with taint $t$ as:
  \[
    \taintAddr{C, t} \defsym \bigcup_{(\addr, (p, b, e, t))\in C}{[b, e)}
  \]
  we then define $\reachAddr{C}$ as the whole set of reachable addresses,
  \[
    \reachAddr C \defsym \taintAddr{C, \tainted} \cup \taintAddr{C, \untainted}.
  \]
  For convenience, we use \( \taintAddr{\reg, \mem, t} \) as a
  shorthand for $\taintAddr{\rc{\reg, \mem}, t}$, and
  $\reachAddr{\reg, \mem}$ as a shorthand for
  $\reachAddr{\rc{\reg, \mem}}$.  By slight abuse of notation, given a
  permission $p \in \Perm$, we also define
  \[
    \permAddr{C, p} \defsym \bigcup_{(\addr, (p, b, e, t))\in C}{[b, e)}
  \]
  and the notation $\permAddr{\reg, m, p}$ in the expected way.
  \label{def:reach-mem}
\end{definition}

\begin{definition}[Well-formed ISA State]
ISA state $\cf{\mem, \reg}$ is \emph{well-formed} (denoted as $\wf{\mem, \reg}$) if its tainted memory and untainted memory do not overlap, i.e.,
\begin{equation*}
  \taintAddr{\reg, \mem, \tainted} \cap \taintAddr{\reg, \mem, \untainted} = \emptyset
\end{equation*}
and its writable and executable memory do not overlap, i.e.,
\begin{equation*}
  \permAddr{\reg, \mem, \readwrite} \cap \permAddr{\reg, \mem, \exec} = \emptyset.
\end{equation*}
\label{def:isa-wf}
\end{definition}

\begin{lemma}[ISA Model Maintains ISA Well-formedness]
For all $\mem$, $\reg$, $\mem'$, $\reg'$, $o$, if
\begin{align*}
  & \wf{\mem, \reg} & & \cf{\mem, \reg} \step{o} \cf{\mem', \reg'},
\end{align*}
then
\begin{equation*}
  \wf{\mem', \reg'}.
\end{equation*}
\label{lemma:isa-wf}
\end{lemma}
\begin{proof}
This lemma is a direct consequence of \Cref{lemma:reach-mem},
which establishes monotonicity of $\taintAddr\cdot$ and $\permAddr\cdot$.
\end{proof}
\begin{lemma}
For all $\mem$, $\reg$, $\mem'$, $\reg'$, $o$ such that
\begin{align*}
  & \wf{\mem, \reg} \\
  & \cf{\mem, \reg} \step{o} \cf{\mem', \reg'},
\end{align*}
there should be
\begin{align*}
  \forall \taint,&&& \taintAddr{\reg', \mem', \taint} \subseteq \taintAddr{\reg, \mem, \taint} \\
  \forall \perm, &&& \permAddr{\reg', \mem', \perm} \subseteq \permAddr{\reg, \mem, \perm}.
\end{align*}
\label{lemma:reach-mem}
\end{lemma}
\begin{proof}
The proof proceeds by case analysis on the rule used for $\cf{\mem, \reg} \step{o} \cf{\mem', \reg'}$.
For each case, the key idea is to utilize the monotonicity of capability operations.
We omit the details here.
\end{proof}
\begin{proof}
  The proof goes by case analysis on our transition relation
  \begin{itemize}
  \item If the transition was established by applying
    Rule~\ref{isa:isa-assign}, then we have $\mem'= \mem$ and
    $\reg' = \reg[\pc\mapsto
    \sem{\stdNextPc{\pc}}_\reg][x \mapsto \sem{e}_{\reg}]$, therefore it suffices to observe
    that if $\sem{e}_{\reg}$ evaluates to a capability
    $c = \stdCap \addr$, by \Cref{lemma:exprmono}, there exists a
    register $z$ such that $\reg(\vz) = (\stdCap {\addr'}, t, \tcap)$,
    and $\stdCapMeta \addr \sqsubseteq \stdCapMeta \addr'$. Therefore, we have
    $\rc{\reg[\vx\mapsto c], m, \vx}\subseteq \rc{\reg, m}\cup\{c\}$,
    and thereby, by \Cref{lemma:rcupd}, we have
    \[
      \rc{\reg[\vx\mapsto c], m} \subseteq \rc{\reg, m} \cup\{c\}.
    \]
    By definition of $\taintAddr \cdot$, and
    $\stdCapMeta \addr \sqsubseteq \stdCapMeta \addr'$, we conclude
    $\taintAddr{\reg[\vx\mapsto c], m, t} \subseteq \taintAddr{\reg,
      m, t}$. By reiterating this reasoning to also cover the $\pc$
    update, we prove
    \[
      \taintAddr{\reg[\pc\mapsto
      \sem{\stdNextPc{\pc}}_\reg][x \mapsto \sem{e}_{\reg}], m, t}\subseteq \taintAddr{\reg,
      m, t}.
  \]
  This establishes the first invariant. The second invariant follows analogously.
\item If the transition was established by applying
  Rule~\ref{isa:isa-load}, then we have $\mem'= \mem$ and
  $\reg' = \reg[\pc\mapsto \sem{\stdNextPc{\pc}}_\reg][x \mapsto (w, t, c)]$,
  where $c$ is either $\tcap$ or $\tval$. In the former case, it is
  easy to observe that
  \[
    \rc{\reg[\vx\mapsto (w, t, c)], m} \subseteq \rc{\reg, m}.
  \]
  By definition of $\taintAddr \cdot$, we conclude
  \[
    \taintAddr{\reg[\pc\mapsto \sem{\stdNextPc{\pc}}_\reg][\vx\mapsto (w, t, c)], m, t'} \subseteq \taintAddr{\reg,
      m, t'},
  \]
  and similarly for $\permAddr\cdot$.  On the other hand, if
  $c=\tcap$, by introspection of the rules for memory operation, we
  observe that the capability value must have been loaded via rules
  \ref{memory:read-mem} and \ref{helpersys:read-cap}. More precisely,
  said $\mem=(\memData, \memTag)$, and $\sem e_\reg = ((\addr, \stdCapMeta \addr), \taint_c, \tcap)$ we have:
    \begin{align*}
      (w, t, \tcap) &= \stdReadMem{(\memData, \memTag), \sem e_\reg, 2}\\
                    &= \stdReadRes{\memData[\addr, \addr + 2), \taint_c, \memTag[\addr / \szCap], 2}
    \end{align*}
    so we deduce $\tcap= \memTag[\addr / \szCap]$, and said $\memData[\addr, \addr + 2) = (z_0, z_1)$,
    and $(z_{\text{meta}}, t_z)=\stdUnpackTaint{z_m}$, we deduce that $w=(z_0, z_{\text{meta}})$ and $z_{\text{meta}}\neq \stdMagic$. Furthermore, $\addr \%2 = 0$ holds by introspection of Rule~\ref{memory:read-mem}. Observe that if we are able to find a second capability in $\rc{\reg, \mem}$ such that it spans over $[\addr, \addr+2)$, we can conclude that $w=(z_0, z_{\text{meta}}) \in \rc{\reg, \mem}$ by Rule~\ref{rc:ind}, since all the other premises of the rules can be discharged by using the intermediate observations on $(z_0, z_{\text{meta}})$. The existence of such a capability spanning over $[\addr, \addr+2)$ is guaranteed by \Cref{lemma:exprmono}. This proves:
    \[
      \rc{\reg[\vx\mapsto (w, t, c)], m} \subseteq \rc{\reg, m},
    \]
    and we can conclude this sub-derivation by reasoning identically
    to the case where the loaded value was not a capability.
\item If the transition was established by applying
  Rule~\ref{isa:isa-store}, then we call
    \begin{align*}
    \sem e_\reg &= ((\addr_e, (p_e, b_e, g_e, t_e)), \taint,\tcap)\\
    \sem \vx_\reg &= (\val, \taint_x, c_x),
    \end{align*}
  and we have
  \begin{align*}
    \reg'&= \reg\\
    \mem' & = \stdWriteMem{\mem, \sem e_\reg, sz, \sem \vx_\reg}\\
         & = \stdWriteMem{\mem, \sem e_\reg, sz, \sem \vx_\reg}\\
          & = (\memData[[\addr, \addr+\sz) \mapsto \val', \memTag[\addr/2 \mapsto c']])
  \end{align*}
  where $(v', c')= \stdWriteRes{(v, t_x, c_x), \taint_e, \sz}$.
  The proof then proceeds by case analysis.
  By introspection of the rules for $\stdWriteRes{}$, either of the following cases holds:
  \begin{itemize}
    \item $c' = \tval \lor c_x = \tval$: By \Cref{lemma:storemono1}, we can easily conclude that
    \begin{equation*}
      \rc{\reg', \mem'}\subseteq\rc{\reg, \mem}.
    \end{equation*}
    \item $c_x = \tcap \land \sz = \szCap$:
    By rules for memory operations, there should be $\addr \% \szCap = 0$.
    By \Cref{lemma:storemono2}, we have
    \begin{equation*}
      \rc{\reg', \mem'} \subseteq \rc{\reg, \mem} \cup \qty{\stdCap{\addr'}: \stdCap{\addr'} \sqsubseteq v}.
    \end{equation*}
    By the definition of $v$ and \Cref{lemma:exprmono}, $v \in \rc{\reg, \mem}$.
    Then, by \ref{rc:shrink}, $\qty{\stdCap{\addr'}: \stdCap{\addr'} \sqsubseteq v} \subseteq \rc{\reg, \mem}$.
    Therefore, there should be $\rc{\reg', \mem'}\subseteq\rc{\reg, \mem}$.
  \end{itemize}
  \end{itemize}
\end{proof}

\begin{lemma}
  \label{lemma:rcupd}
  For an ISA configuration $(m, r)$ and a capability $c$, we have
  $\rc{\reg[\vx\mapsto c], m} \subseteq \rc{\reg, m}\cup \rc{\reg[\vx\mapsto c], m, \vx}$
\end{lemma}
\begin{proof}
  Direct consequence of
  \begin{equation*}
    \forall z,\; \rc{\reg[\vx\mapsto c], m, z} \subseteq \rc{\reg, m}\cup \rc{\reg[\vx\mapsto c], m, \vx},
  \end{equation*}
  which is established by case analysis on whether $z=x$.
\end{proof}

\begin{lemma}
For all $\mem$, $\reg$, $e$ such that
$\sem{e}_{\reg} = (\stdCap{\addr}, \taint, \tcap)$,
there should be $\stdCap{\addr} \in \rc{\reg, \mem}$.
\label{lemma:isa-exp-cap-mono}
\end{lemma}
\begin{proof}
By \Cref{lemma:exprmono}, there exists $\vx$ such that
\begin{align*}
  & \reg(\vx) = (\stdCap {\addr'}, \_, \tcap) && \stdCapMeta \addr \sqsubseteq \stdCapMeta \addr'.
\end{align*}
By \Cref{def:reach-cap}, $\stdCap{\addr'} \in \rc{\reg, \mem}$, so there should also be
$\stdCap{\addr} \in \rc{\reg, \mem}$ by Rule~\ref{rc:shrink}.
\end{proof}

\begin{lemma}[Monotonicity of expression evaluation]
\label{lemma:exprmono}
For every register file $\reg$ and expression $e$, if
$\sem{e}_\reg = (\stdCap \addr, \_, \tcap)$,
there exists a register $\vx$ such that
$\reg(\vx) = (\stdCap {\addr'}, \_, \tcap)$, and $\stdCapMeta \addr \sqsubseteq \stdCapMeta \addr'$.
\end{lemma}
\begin{proof}
  The proof is by induction on the syntax of the expression. If the
  expression is a register, or a value, the conclusion is trivial, and
  vacuously satisfied, respectively. We now consider the case where the
  expression is given by a composition of operators. In this case we
  have $e = \op(e_0, e_1)$.  If the outermost operator is a value
  operator, the claim is vacuously satisfied.  If the outermost
  operator is a capability operator, then the proof proceeds by case
  analysis on $\sem {e_0}_\reg$ and $\sem {e_1}_\reg$. If they are
  both tagged as values or as capabilities, then we reach a
  contradiction because by our assumption on capability operators
  (Axiom~\eqref{ax:capop1} of capability operators)
  the output $\sem{e}_\reg$ is tagged as a regular value
  while, according to our assumption, it has to be tagged as a
  capability.
  Finally, by inductive hypothesis, if $\sem {e_i}_\reg$ evaluates to a value and
  $\sem {e_{1-i}}_\reg$ evaluates to a capability $\stdCap a$, then
  there exists a register $\vx$ such that
  $\reg(\vx) = (\stdCap {\addr'}, \_, \tcap)$, and
  $\stdCapMeta a \sqsubseteq \stdCapMeta \addr'$. By monotonicity of
  capability operators (Axiom~\eqref{ax:capop2} of capability
  operators), the resulting capability $\stdCap \addr$ satisfies
  $\stdCapMeta \addr \sqsubseteq \stdCapMeta a \sqsubseteq \stdCapMeta
  \addr'$.  Finally, observe that in this case, the inner taint $t$ of
  the resulting capability $\stdCapMeta \addr$ coincides with the one
  of $\stdCapMeta a$ by assumption on operators by
  Axiom~\eqref{ax:capop3} on capability operators.
\end{proof}

\begin{lemma}[Monotonicity of reachable capabilities after store value]
  \label{lemma:storemono1}
  For every register $\reg$ and memories $m =(\memData, \memTag)$, and
  $\mem' = (\memData', \memTag') = (\memData[[\addr_0, \addr_0+\sz) \mapsto \val],
  \memTag[\addr_0/2 \mapsto c])$, where $(v, c) = \stdWriteRes{((v_0, v_m), t_0, c_0), t_c, \sz}$, $\sz = \szCap \Rightarrow \addr_0 \% \szCap = 0$, and $c=\tval \lor c_0 =\tval$, we
  have
  \[
    \rc{\reg, \mem'}
    \subseteq \rc {\reg, \mem}.
  \]
\end{lemma}
\begin{proof}
It suffices to prove that for every register $z$, $\rc{\reg, \mem', z} \subseteq \rc{\reg, \mem, z}$.
So it suffices to prove that for all $\stdCap{\addr}$ such that $\stdCap{\addr} \in \rc{\reg, \mem', z}$, there must also be $\stdCap{\addr} \in \rc{\reg, \mem, z}$.

The proof is by induction on the proof derivation.
Consider the following cases:
\begin{itemize}
  \item $\stdCap{\addr} \in \rc{\reg, \mem', z}$ was proved with \ref{rc:base}: in this case, the conclusion is trivial.
  \item $\stdCap{\addr} \in \rc{\reg, \mem', z}$ was proved with \ref{rc:ind}:
  By \ref{rc:ind}, there must exist $(\addr_1, (p, b, e, t))$ such that $(\addr_1, (p, b, e, t)) \in \rc{\reg, \mem', z}$, and by induction hypothesis, there should also be $(\addr_1, (p, b, e, t)) \in \rc{\reg, \mem, z}$.
  By \ref{rc:ind}, $\stdCap{\addr} \in \rc{\reg, \mem', z}$ implies that there exists $\addr'$ such that
  \begin{equation}
    \begin{aligned}
    & [\addr', \addr' + \szCap) \in [b, e) \\
    & \memData'[\addr', \addr' + \szCap] = (\addr, \addr_m)
    \quad \memTag'[\addr' / \szCap] = \tcap \quad \addr' \% 2 = 0 \\
    & \stdUnpackTaint{\addr_m} = (\stdCapMeta{\addr}, \_) \quad \addr_m \neq \stdMagic.
    \end{aligned}
    \label{eq:mono-reach-cap-hyp-1}
  \end{equation}
  Furthermore, there should be
  \begin{equation}
    [\addr', \addr' + \szCap) \cap [\addr_0, \addr_0 + \sz) = \emptyset.
    \label{eq:mono-reach-cap-non-overlap}
  \end{equation}
  This can be proved by considering the following cases:
  \begin{itemize}
    \item $\sz = \szAddr$: by \ref{helpersys:write-data}, $\memTag'[\addr_0/\szCap] = c = \tval$. Note that $\memTag'[\addr / \szCap] = \tcap$ and $\addr' \% \szCap = 0$, so $ \addr_0 \not\in[\addr', \addr' + \szCap)$. Thus, \eqref{eq:mono-reach-cap-non-overlap} holds.
    \item $\sz = \szCap$, $v_m = 0 \land \taint_0 = \untainted$:
    in this case, there should be $\addr_0 \% \szCap = 0$.
    Furthermore, by \ref{helpersys:spill-value}, there should be
    \begin{align*}
      & \memData'[\addr_0, \addr_0 + \szCap] = (v_0, \stdMagic) &
      & \memTag'[\addr_0/\szCap] = \tcap.
    \end{align*}
    Note that $\memData'[\addr', \addr' + \szCap] = (\addr, \addr_m)$ and $\addr_m \neq \stdMagic$, so $\addr' \neq \addr_0$.
    Furthermore, $\addr' \% \szCap = \addr_0 \% \szCap = 0$, so \eqref{eq:mono-reach-cap-non-overlap} holds.
    \item $\sz = \szCap$, $v_m \neq 0 \lor \taint_0 = \tainted$:
    in this case, there should be $\addr_0 \% \szCap = 0$.
    By \ref{helpersys:spill-value}, $\memTag'[\addr_0/\szCap] = c = \tval$.
    Note that $\memTag'[\addr' / \szCap] = \tcap$ and $\addr' \% \szCap = 0$, so \eqref{eq:mono-reach-cap-non-overlap} holds.
  \end{itemize}
  By \eqref{eq:mono-reach-cap-non-overlap}, $\memData[\addr', \addr' + \szCap] = \memData'[\addr', \addr' + \szCap]$ and $\memTag[\addr' / \szCap] = \memTag'[\addr' / \szCap]$.
  This implies that \eqref{eq:mono-reach-cap-hyp-1} still holds if we replace $\memData'$ with $\memData$, and $\memTag'$ with $\memTag$. Therefore, by \ref{rc:ind}, there should be $\stdCap{\addr} \in \rc{\reg, \mem, z}$.
  \item $\stdCap{\addr} \in \rc{\reg, \mem', z}$ was proved with \ref{rc:shrink}:
  By \ref{rc:shrink}, there exists $\stdCap{\addr'}$ such that $\stdCap{\addr'} \in \rc{\reg, \mem', z}$ and $\stdCap{\addr} \sqsubseteq \stdCap{\addr'}$.
  By induction hypothesis, there should also be $\stdCap{\addr'} \in \rc{\reg, \mem, z}$.
  Therefore, by \ref{rc:shrink}, $\stdCap{\addr} \in \rc{\reg, \mem, z}$ holds.
\end{itemize}
\end{proof}
\begin{lemma}[Monotonicity of reachable capabilities after store capability]
  \label{lemma:storemono2}
  For every register $\reg$ and memories $m =(\memData, \memTag)$, and
  $\mem' = (\memData', \memTag') = (\memData[[\addr_0, \addr_0+2) \mapsto \val], \memTag[\addr_0/2 \mapsto c])$, where $\addr_0 \%2 = 0$ and
  $(v, c)= \stdWriteRes{((v_0, \stdCapMeta{v}), t_0, \tcap), \taint_c, \szCap}$, we
  have
  \[
    \rc{\reg, \mem'}
    \subseteq \rc {\reg, \mem} \cup
    \qty{\stdCap{\addr}: \stdCap{\addr} \sqsubseteq (v_0, \stdCapMeta{v})}.
  \]
\end{lemma}
\begin{proof}
Denote $R = \qty{\stdCap{\addr}: \stdCap{\addr} \sqsubseteq (v_0, \stdCapMeta{v})}$.
It suffices to prove that for every register $z$, $\rc{\reg, \mem', z} \subseteq \rc{\reg, \mem, z} \cup R$.
So it suffices to prove that for all $\stdCap{\addr}$ such that $\stdCap{\addr} \in \rc{\reg, \mem', z}$, there must also be $\stdCap{\addr} \in \rc{\reg, \mem, z} \cup R$.

The proof is by induction on the proof derivation. Consider the following cases:
\begin{itemize}
  \item $\stdCap{\addr} \in \rc{\reg, \mem', z}$ was proved with \ref{rc:base}: in this case, the conclusion is trivial.
  \item $\stdCap{\addr} \in \rc{\reg, \mem', z}$ was proved with \ref{rc:ind}:
  By \ref{rc:ind}, there must exist $(\addr_1, (p, b, e, t))$ such that $(\addr_1, (p, b, e, t)) \in \rc{\reg, \mem', z}$, and by induction hypothesis, there should also be $(\addr_1, (p, b, e, t)) \in \rc{\reg, \mem, z}$.
  By \ref{rc:ind}, $\stdCap{\addr} \in \rc{\reg, \mem', z}$ implies that there exists $\addr'$ such that
  \begin{equation}
    \begin{aligned}
    & [\addr', \addr' + \szCap) \in [b, e) \\
    & \memData'[\addr', \addr' + \szCap] = (\addr, \addr_m)
    \quad \memTag'[\addr' / \szCap] = \tcap \quad \addr' \% 2 = 0 \\
    & \stdUnpackTaint{\addr_m} = (\stdCapMeta{\addr}, \_) \quad \addr_m \neq \stdMagic.
    \end{aligned}
    \label{eq:mono-reach-cap-hyp-2}
  \end{equation}
  Consider the following cases:
  \begin{itemize}
    \item $\addr' = \addr_0$: in this case, there should be
    \begin{equation*}
      ((\addr, \addr_m), \tcap) = (v, c) = \stdWriteRes{((v_0, \stdCapMeta{v}), \taint_0, \tcap), \taint_c, \szCap}.
    \end{equation*}
    By \ref{helpersys:write-cap}, this implies that
    \begin{align*}
      & \addr = v_0 && \addr_m = \stdPackTaint{\stdCapMeta{v}, \taint_0}.
    \end{align*}
    Hence,
    \begin{equation*}
      (\stdCapMeta{v}, \taint_0) = \stdUnpackTaint{\addr_m} = (\stdCapMeta{\addr}, \_).
    \end{equation*}
    In short, the above analysis shows that $\stdCap{\addr} = (v_0, \stdCapMeta{v})$. Thus, $\stdCap{\addr} \in R \subseteq \rc{\reg, \mem, z} \cup R$.
    \item $\addr' \neq \addr_0$: Since $\addr' \% \szCap = \addr_0 \% \szCap = 0$, then $[\addr', \addr' + \szCap) \cap [\addr_0, \addr_0 + \szCap) = \emptyset$.
    Thus, $\memData[\addr', \addr' + \szCap] = \memData'[\addr', \addr' + \szCap]$ and $\memTag[\addr' / \szCap] = \memTag'[\addr' / \szCap]$.
    This implies that \eqref{eq:mono-reach-cap-hyp-2} still holds if we replace $\memData'$ with $\memData$, and $\memTag'$ with $\memTag$. Therefore, by \ref{rc:ind}, there should be $\stdCap{\addr} \in \rc{\reg, \mem, z} \subseteq \rc{\reg, \mem, z} \cup R$.
  \end{itemize}
  \item $\stdCap{\addr} \in \rc{\reg, \mem', z}$ was proved with \ref{rc:shrink}:
  By \ref{rc:shrink}, there exists $\stdCap{\addr'}$ such that $\stdCap{\addr'} \in \rc{\reg, \mem', z}$ and $\stdCap{\addr} \sqsubseteq \stdCap{\addr'}$.
  By induction hypothesis, there should also be $\stdCap{\addr'} \in \rc{\reg, \mem, z}$.
  Therefore, by \ref{rc:shrink}, $\stdCap{\addr} \in \rc{\reg, \mem, z}$ holds.
\end{itemize}
\end{proof}
\subsubsection{ISA Maintains Public Equivalence}
\label{sec:isa-sec-proof}
\begin{definition}[Untainted Projection]
We define untainted projection of a value associated with taint and capability tag, i.e., $(v, \taint, c)$ as
\begin{equation*}
  \lowproj{(v, \taint, c)} =
  \begin{cases}
    (v, \untainted, c) & \taint = \untainted \\
    \bot & \taint = \tainted.
  \end{cases}
\end{equation*}
\label{def:low-proj-val}
\end{definition}

\begin{definition}[Public Equivalence for ISA States]
ISA states
$S = \cf{\mem, \reg}$, $S' = \cf{\mem', \reg'}$
are \emph{publicly equivalent}, i.e., $S \pubeq S'$, if
\begin{align*}
\forall x \in \Reg, &&&
\lowproj{\reg(x)} = \lowproj{\reg'(x)}.
\end{align*}
and
\begin{multline*}
\forall \addr, \sz \in \qty{\szAddr, \szCap}, \taint,\\
\begin{aligned}
& ((\sz = \szCap \Rightarrow \addr \% \szCap = 0) \\
& \land [\addr, \addr + \sz) \subseteq \taintAddr{\reg, \mem, \taint} \cap \taintAddr{\reg', \mem', \taint}) \\
\Rightarrow \quad &
(\lowproj{\stdReadRes{\memData[\addr, \addr + \sz], t, \memTag[\addr / \szCap], \sz}} \\
& = \lowproj{\stdReadRes{\memData'[\addr, \addr + \sz], t, \memTag'[\addr / \szCap], \sz}})
\end{aligned}
\end{multline*}
where $\mem = (\memData, \memTag)$, $\mem' = (\memData', \memTag')$.
\label{def:pub-eq-isa}
\end{definition}

\begin{lemma}[ISA Model Maintains Public Equivalence]
Let $S_0$, $S_0'$ be well-formed initial ISA states such that $S_0\pubeq S_0'$ and $S_0 \equiv_{\mathit{ISA}}S_0'$.
Then, for all $n$ such that
\begin{equation*}
  S_0 \step{o_1} S_1 \step{o_2} S_2 \dots \step{o_n} S_n,
\end{equation*}
there must exist $S_1', \dots, S_n'$ such that
\begin{equation*}
  S_0' \step{o_1} S_1 \step{o_2} S_2' \dots \step{o_n} S_n',
\end{equation*}
and $S_n \pubeq S_n'$.
\label{lemma:isa-pubeq}
\end{lemma}
\begin{proof}
The proof of the inductive claim is by induction on $n$.
As the base case is trivial, we directly go to the inductive case.
Assuming $S_n\pubeq S_n'$, we must prove that the target configurations $S_{n+1}$ and $S_{n+1}'$ satisfy $S_{n+1} \pubeq S_{n+1}'$.
Recall that we denote $S_n \step{o_{n+1}} S_{n+1}$ and $S_n' \step{o_{n+1}} S_{n+1}'$.
We also denote $S_n=\cf{\mem_n, \reg_n}$, $S_n' = \cf{\mem_n', \reg_n'}$.
By repeatedly applying \Cref{lemma:isa-wf}, we can get $\wf{S_n}$ and $\wf{S_n'}$.

By rules for ISA semantics and \ref{memory:read-mem-inst}, register $\pc$ should be untainted in both configurations $S_n$ and $S_n'$, i.e.,
\begin{align*}
  \sem{\pc}_{\reg_n} & = (\_, \untainted, \tcap) & \sem{\pc}_{\reg_n'} & = (\_, \untainted, \tcap).
\end{align*}
Since $S_n \pubeq S_n'$, then by \Cref{def:pub-eq-isa} there should be $\sem{\pc}_{\reg_n} = \sem{\pc}_{\reg_n'}$.
Furthermore, by \Cref{lemma:isa-exp-cap-mono}, we can get
\begin{align*}
  \sem{\pc}_{\reg_n} & \in \rc{\reg_n. \mem_n} & \sem{\pc}_{\reg_n'} & \in \rc{\reg_n', \mem_n'}.
\end{align*}
Then, by \Cref{lemma:read-mem-inst}, there should be an $\instr$ such that
\begin{align*}
  \instr & = \stdReadMemInst{\mem_n, \sem{\pc}_{\reg_n}} \\
  & = \stdReadMemInst{\mem_n', \sem{\pc}_{\reg_n'}}
\end{align*}
The proof proceeds by case analysis on $\instr$.
We consider the following cases:

\pgheading{Case $\instr = x \gets e$}
By \ref{isa:isa-assign}, the next states are
\begin{align*}
  \reg_{n+1} & = \reg_n[\pc \mapsto \sem{\stdNextPc{\pc}}_{\reg_n}][x \mapsto \sem{e}_{\reg_n}] \\
  \reg_{n+1}' & = \reg_n'[\pc \mapsto \sem{\stdNextPc{\pc}}_{\reg_n'}][x \mapsto \sem{e}_{\reg_n'}]\\
  S_{n+1} & = \cf{\mem_n, \reg_{n+1}} \\
  S_{n+1}' & = \cf{\mem_n', \reg_{n+1}'}.
\end{align*}
By \Cref{def:pub-eq-isa}, $S_n \pubeq S_n'$ implies that
\begin{align*}
  \lowproj{\sem{e}_{\reg_n}} & = \lowproj{\sem{e}_{\reg_n'}} &
  \lowproj{\sem{\stdNextPc{\pc}}_{\reg_n}} & = \lowproj{\sem{\stdNextPc{\pc}}_{\reg_n'}}.
\end{align*}
Hence,
\begin{align*}
\forall x \in \Reg, &&&
\lowproj{\reg_{n+1}(x)} = \lowproj{\reg_{n+1}'(x)}.
\end{align*}
Denote $\mem_n = (\memData, \memTag)$, $\mem_n' = (\memData', \memTag')$.
To prove $S_{n+1} \pubeq S_{n+1}'$, it suffices to prove that
for all $\addr$, $\sz \in \qty{\szAddr, \szCap}$, $\taint$ such that
\begin{equation*}
\begin{aligned}
  & (\sz = \szCap \Rightarrow \addr \% \szCap = 0) \\
  \land \; & [\addr, \addr + \sz) \subseteq \taintAddr{\reg_{n+1}, \mem_n, \taint} \cap \taintAddr{\reg_{n+1}', \mem_n', \taint},
\end{aligned}
\label{eq:isa-assign-assume}
\end{equation*}
there must be
\begin{equation}
\begin{aligned}
  & \lowproj{\stdReadRes{\memData[\addr, \addr + \sz], t, \memTag[\addr / \szCap], \sz}} \\
  = \; & \lowproj{\stdReadRes{\memData'[\addr, \addr + \sz], t, \memTag'[\addr / \szCap], \sz}}.
\end{aligned}
\label{eq:isa-assign-goal}
\end{equation}
Fix $\addr$, $\sz$, $\taint$.
By \Cref{lemma:reach-mem}, for all $\taint$,
\begin{equation}
\begin{aligned}
[\addr, \addr + \sz]
& \subseteq (\taintAddr{\reg_{n+1}, \mem_n, \taint} \cap \taintAddr{\reg_{n+1}', \mem_n', \taint})\\
& \subseteq
(\taintAddr{\reg_n, \mem_n, \taint} \cap \taintAddr{\reg_n', \mem_n', \taint})
\end{aligned}
\end{equation}
By induction hypothesis, i.e., $S_n \pubeq S_n'$, \eqref{eq:isa-assign-goal} holds.

\pgheading{Case $\instr = \load{x}{e}{\sz}$}
Denote
\begin{align*}
  & (\stdCap{\addr}, \_, \tcap) = \sem{e}_{\reg_n} &
  & (\stdCap{\addr'}, \_, \tcap) = \sem{e}_{\reg_n'}.
\end{align*}
By \ref{isa:isa-load}, there should be
\begin{equation*}
  o_{n+1} = \stdLoad{\addr, \stdCapMeta{\addr}} = \stdLoad{\addr', \stdCapMeta{\addr'}}.
\end{equation*}
By \ref{isa:isa-load}, we can denote
\begin{align*}
  (v, t, c) & = \stdReadMem{\mem_n, (\stdCap{\addr}, \_, \tcap), \sz} \\
  (v', t', c') & = \stdReadMem{\mem_n', (\stdCap{\addr'}, \_, \tcap), \sz}.
\end{align*}
and get
\begin{align*}
  S_{n+1} & = \cf{\mem_n, \reg_n[\pc \mapsto \sem{\stdNextPc{\pc}}_{\reg_n}][x \mapsto (v, t, c)]}\\
  S_{n+1}' & = \cf{\mem_n', \reg_n'[\pc \mapsto \sem{\stdNextPc{\pc}}_{\reg_n'}][x \mapsto (v', t', c')]}.
\end{align*}
Hence, it suffices to prove that $\lowproj{(v, t, c)} = \lowproj{(v', t', c')}$.
Denote
\begin{align*}
& \mem_n = (\memData, \memTag) \qquad \mem_n' = (\memData', \memTag') \\
& \stdCapMeta{\addr} = \stdCapMeta{\addr'} = (\perm, \rbase, \rend, t_c).
\end{align*}
By \Cref{lemma:addr-same-taint} and \ref{memory:check-cap},
\begin{align*}
& \sz = \szCap \Rightarrow \addr \% \szCap = 0 \\
& [\addr, \addr + \sz) \subseteq \taintAddr{\reg_n, \mem_n, \taint_c} \cap \taintAddr{\reg_n', \mem_n', \taint_c}.
\end{align*}

Thus, by \Cref{def:pub-eq-isa} and \ref{memory:read-mem},
we have
\begin{align*}
& \lowproj{(v, t, c)} = \lowproj{\stdReadRes{\memData[\addr, \addr + \sz], \taint_c, \memTag[\addr / \szCap], \sz}}\\
= \; & \lowproj{(v', t', c')} = \lowproj{\stdReadRes{\memData'[\addr, \addr + \sz], \taint_c, \memTag'[\addr / \szCap], \sz}}.
\end{align*}

\pgheading{Case $\instr = \store{x}{e}{\sz}$}
Denote
\begin{align*}
  & (\stdCap{\addr}, \_, \tcap) = \sem{e}_{\reg_n} &
  & (\stdCap{\addr'}, \_, \tcap) = \sem{e}_{\reg_n'},
\end{align*}
where $\stdCapMeta{\addr} = (\perm, \rbase, \rend, \taint_c)$ and
$\stdCapMeta{\addr'} = (\perm', \rbase', \rend', \taint_c')$.
Since the two configurations generate the same observation,
then by \ref{isa:isa-store},
\begin{align*}
& o_{n+1} = \stdStore{v, (\addr, \stdCapMeta{\addr})} = \stdStore{v', (\addr', \stdCapMeta{\addr'})}\\
& v = (\taint_c = \untainted) \;?\; \sem{x}_{\reg_n} : \bot \qquad
  v' = (\taint_c' = \untainted) \;?\; \sem{x}_{\reg_n'} : \bot.
\end{align*}
Hence, $\stdCap{\addr} = \stdCap{\addr'}$ and $v = v'$; in particular, $\taint_c = \taint_c'$.
Denote
\begin{align*}
\sem{x}_{\reg_n} & = (v_x, \taint_x, c_x)
& \sem{x}_{\reg_n'} & = (v_x', \taint_x', c_x').
\end{align*}
Then, by discussing the value of $\taint_c$, it is straightforward to derive that
\begin{equation*}
\lowproj{(v_x, \taint_c, c_x)} = \lowproj{(v_x', \taint_c, c_x')}.
\end{equation*}
By \Cref{def:pub-eq-isa}, $S_n \pubeq S_n'$ implies that
\begin{equation*}
\lowproj{(v_x, \taint_x, c_x)} = \lowproj{(v_x', \taint_x', c_x')}.
\end{equation*}

By \ref{isa:isa-store}, the next states can be derived as
\begin{align*}
\mem_{n+1} & = \stdWriteMem{\mem_n, ((\addr, (\perm, \rbase, \rend, \taint_c)), \_, \tcap), \sz, \sem{x}_{\reg_n}} \\
\mem_{n+1}' & = \stdWriteMem{\mem_n', ((\addr, (\perm, \rbase, \rend, \taint_c)), \_, \tcap), \sz, \sem{x}_{\reg_n'}} \\
S_{n+1} & = \cf{\mem_{n+1}, \reg_n[\pc \mapsto \sem{\stdNextPc{\pc}}_{\reg_n}]}\\
S_{n+1}' & = \cf{\mem_{n+1}', \reg_n'[\pc \mapsto \sem{\stdNextPc{\pc}}_{\reg_n'}]}.
\end{align*}
It is straightforward to derive that
\begin{equation*}
  \lowproj{\reg_n[\pc \mapsto \sem{\stdNextPc{\pc}}_{\reg_n}]} =
  \lowproj{\reg_n'[\pc \mapsto \sem{\stdNextPc{\pc}}_{\reg_n'}]}.
\end{equation*}
Hence, we focus on proving the statement about memory in \Cref{def:pub-eq-isa} holds for $\mem_{n+1}$, $\mem_{n+1}'$.
Denote
\begin{align*}
\mem_n & = (\memData, \memTag) & \mem_n' & = (\memData', \memTag')\\
\mem_{n+1} & = (\memDataN{1}, \memTagN{1}) & \mem_{n+1}' & = (\memDataN{1}', \memTagN{1}').
\end{align*}
It suffices to prove that for all $\addr_0$, $\sz_0\in \qty{\szAddr, \szCap}$, $\taint_0$ where
\begin{align*}
  & (\sz_0 = \szCap \Rightarrow \addr_0 \% \szCap = 0)\\
  \land \; &
  [\addr_0, \addr_0 + \sz_0) \subseteq \taintAddr{\reg_{n+1}, \mem_{n+1}, \taint_0} \cap \taintAddr{\reg_{n+1}', \mem_{n+1}', \taint_0}.
\end{align*}
the following statement holds:
\begin{equation}
\begin{aligned}
& \lowproj{\stdReadRes{\memDataN{1}[\addr_0, \addr_0 + \sz_0], \taint_0, \memTagN{1}[\addr_0 / \szCap], \sz_0}} \\
= \; &\lowproj{\stdReadRes{\memDataN{1}'[\addr_0, \addr_0 + \sz_0], \taint_0, \memTagN{1}'[\addr_0 / \szCap], \sz_0}}.
\end{aligned}
\label{eq:isa-store-goal}
\end{equation}
Fix $\addr_0$, $\sz_0\in \qty{\szAddr, \szCap}$, $\taint_0$.

By \Cref{lemma:reach-mem}, for all $\taint$,
\begin{equation}
\begin{aligned}
& (\taintAddr{\reg_{n+1}, \mem_{n+1}, \taint} \cap \taintAddr{\reg_{n+1}', \mem_{n+1}', \taint})\\
\subseteq \; &
(\taintAddr{\reg_n, \mem_n, \taint} \cap \taintAddr{\reg_n', \mem_n', \taint})
\end{aligned}
\end{equation}
By inductive hypothesis,
\begin{equation}
\begin{aligned}
& \lowproj{\stdReadRes{\memData[\addr_0, \addr_0 + \sz_0], \taint_0, \memTag[\addr_0 / \szCap], \sz_0}} \\
= \; &\lowproj{\stdReadRes{\memData'[\addr_0, \addr_0 + \sz_0], \taint_0, \memTag'[\addr_0 / \szCap], \sz_0}}.
\end{aligned}
\label{eq:isa-store-hyp}
\end{equation}

Our proof goes by case discussion on $\addr$, $\sz$, $\addr_0$, and $\sz_0$.
By \ref{memory:check-cap}, $\sz \% 0 \Rightarrow \addr \% 2 = 0$.
Consider the following cases. Note that we use $\addr / \szCap$ as a shortcut for $\lfloor \addr / \szCap \rfloor$.

\begin{itemize}
\item $\sz = \sz_0 = \szCap$, $\addr = \addr_0$:
By assumption, in this case there should be $\addr \% \szCap = \addr_0 \% \szCap = 0$.
By \Cref{lemma:addr-same-taint}, $\taint_0 = \taint_c$.
By \ref{memory:write-mem},
\begin{align*}
& (\memDataN{1}[\addr_0, \addr_0 + \sz_0], \memTagN{1}[\addr_0 / \szCap])\\
= \; & (\memDataN{1}[\addr, \addr + \szCap], \memTagN{1}[\addr / \szCap])\\
= \; & \stdWriteRes{\sem{x}_{r_n}, \taint_c, \szCap}\\
= \; & \stdWriteRes{(v_x, \taint_x, c_x), \taint_c, \szCap}.
\end{align*}
Similarly,
\begin{align*}
& (\memDataN{1}'[\addr_0, \addr_0 + \sz_0], \memTagN{1}'[\addr_0 / \szCap])\\
= \; & \stdWriteRes{(v_x', \taint_x', c_x'), \taint_c, \szCap}.
\end{align*}
By \Cref{lemma:read-write-lowproj}, \eqref{eq:isa-store-goal} holds.

\item $\sz = \sz_0 = \szCap$, $\addr \neq \addr_0$:
Similarly, $\addr \% \szCap = \addr_0 \% \szCap = 0$.
Then, $[\addr, \addr + \szCap) \cap [\addr_0, \addr + \szCap) = \emptyset$ and $\addr/\szCap \neq \addr_0/\szCap$.
Thus, by \ref{memory:write-mem}, \eqref{eq:isa-store-hyp} implies \eqref{eq:isa-store-goal}.

\item $\sz = \szAddr$, $\sz_0 = \szCap$, $\addr \in [\addr_0, \addr_0 + \szCap)$:
Note that in this case, $\addr_0 \% \szCap = 0$, so $\addr / \szCap = \addr_0 / \szCap$.
By \ref{memory:write-mem} and \ref{helpersys:write-data},
\begin{align*}
  (\memDataN{1}[\addr], \memTagN{1}[\addr / \szCap])
  = \stdWriteRes{(v_x, \taint_x, c_x), \taint_c, \szAddr}
  = (v_x, \tval).
\end{align*}
Similarly,
\begin{equation*}
  (\memDataN{1}'[\addr], \memTagN{1}'[\addr / \szCap]) = (v_x', \tval).
\end{equation*}
Then, by \ref{helpersys:read-data},
\begin{align*}
& \stdReadRes{\memDataN{1}[\addr_0, \addr_0 + \sz_0], \taint_0, \memTagN{1}[\addr_0 / \szCap], \sz_0}\\
= \; & \stdReadRes{\memDataN{1}[\addr_0, \addr_0 + \sz_0], \taint_0, \memTag[\addr_0 / \szCap], \sz_0} \\
= \; & \stdReadRes{\memDataN{1}[\addr_0, \addr_0 + \sz_0], \taint_0, \tval, \sz_0} \\
= \; & (\memDataN{1}[\addr_0, \addr_0 + \szCap], \taint_0, \tval).
\end{align*}
Similarly,
\begin{align*}
& \stdReadRes{\memDataN{1}'[\addr_0, \addr_0 + \sz_0], \taint_0, \memTagN{1}'[\addr_0 / \szCap], \sz_0}\\
= \; & (\memDataN{1}'[\addr_0, \addr_0 + \szCap], \taint_0, \tval).
\end{align*}
It suffices to prove
\begin{equation*}
  \lowproj{(\memDataN{1}[\addr_0, \addr_0 + \szCap], \taint_0, \tval)} = \lowproj{(\memDataN{1}'[\addr_0, \addr_0 + \szCap], \taint_0, \tval)}
\end{equation*}
Since $\addr \in [\addr_0, \addr_0 + \szCap)$, there should be either $\addr = \addr_0$ or $\addr = \addr_0 + \szAddr$.
Here we only prove the first case; the second one can be proved similarly.
When $\addr = \addr_0$,
\begin{align*}
  \memDataN{1}[\addr_0, \addr_0 + \szCap] & = (v_x, \memData[\addr_0 + \szAddr]) \\
  \memDataN{1}'[\addr_0, \addr_0 + \szCap] & = (v_x', \memData'[\addr_0 + \szAddr]).
\end{align*}
By inductive hypothesis,
\begin{align*}
  & \lowproj{\stdReadRes{\memData[\addr_0 + \szAddr], \taint_0, \memTag[\addr_0  / \szAddr], \szAddr}} \\
  = \;
  & \lowproj{\stdReadRes{\memData'[\addr_0 + \szAddr], \taint_0, \memTag'[\addr_0  / \szAddr], \szAddr}}.
\end{align*}
By \ref{helpersys:read-data}, this implies that
\begin{equation*}
  \lowproj{(\memData[\addr_0 + \szAddr], \taint_0, \tval)} = \lowproj{(\memData'[\addr_0 + \szAddr], \taint_0, \tval)}.
\end{equation*}
Furthermore, recall that $\lowproj{(v_x, \taint_c, c_x)} = \lowproj{(v_x', \taint_c, c_x')}$.
And since $\addr \in [\addr_0, \addr_0 + \szCap)$, it is straightforward to derive that $\taint_0 =\taint_c$ using \Cref{lemma:addr-same-taint}.
Therefore,
\begin{align*}
  & \lowproj{(\memDataN{1}[\addr_0, \addr_0 + \szCap], \taint_0, \tval)} \\
  = \; & \lowproj{((v_x, \memData[\addr_0 + \szAddr]), \taint_0, \tval)}
  = \lowproj{((v_x', \memData'[\addr_0 + \szAddr]), \taint_0, \tval)} \\
  = \; & \lowproj{(\memDataN{1}'[\addr_0, \addr_0 + \szCap], \taint_0, \tval)}
\end{align*}

\item $\sz = \szAddr$, $\sz_0 = \szCap$, $\addr \not\in [\addr_0, \addr_0 + \szCap)$:
Note that in this case $\addr_0 \% \szCap = 0$, so $\addr / \szCap \neq \addr_0 / \szCap$.
By \ref{memory:write-mem},
\begin{align*}
& \stdReadRes{\memDataN{1}[\addr_0, \addr_0 + \sz_0], \taint_0, \memTagN{1}[\addr_0 / \szCap], \sz_0}\\
= \; & \stdReadRes{\memData[\addr_0, \addr_0 + \sz_0], \taint_0, \memTag[\addr_0 / \szCap], \sz_0}\\
& \stdReadRes{\memDataN{1}'[\addr_0, \addr_0 + \sz_0], \taint_0, \memTagN{1}'[\addr_0 / \szCap], \sz_0}\\
= \; & \stdReadRes{\memData'[\addr_0, \addr_0 + \sz_0], \taint_0, \memTag'[\addr_0 / \szCap], \sz_0}.
\end{align*}
Therefore, \eqref{eq:isa-store-hyp} implies \eqref{eq:isa-store-goal}.

\item $\sz_0 = \szAddr$:
By \ref{helpersys:read-data},
\begin{align*}
  & \stdReadRes{\memDataN{1}[\addr_0, \addr_0 + \sz_0], \taint_0, \memTagN{1}[\addr_0 / \szCap], \sz_0}\\
  =\; & (\memDataN{1}[\addr_0], \taint_0, \tval) \\
  & \stdReadRes{\memDataN{1}'[\addr_0, \addr_0 + \sz_0], \taint_0, \memTagN{1}'[\addr_0 / \szCap], \sz_0}\\
  =\; & (\memDataN{1}'[\addr_0], \taint_0, \tval).
\end{align*}
Consider the following two cases:
\begin{itemize}
  \item $\addr_0 \in [\addr, \addr + \szCap)$:
  Note that by \ref{memory:write-mem},
  \begin{align*}
    (\memDataN{1}[\addr, \addr + \szCap], \memTagN{1}[\addr / \szCap])
    & = \stdWriteRes{(v_x, \taint_x, c_x), \taint_c, \szCap}\\
    (\memDataN{1}'[\addr, \addr + \szCap], \memTagN{1}'[\addr / \szCap])
    & = \stdWriteRes{(v_x', \taint_x', c_x'), \taint_c, \szCap}.
  \end{align*}
  By \Cref{lemma:read-write-lowproj}, this implies that
  \begin{equation*}
    \lowproj{(\memDataN{1}[\addr, \addr + \szCap], \taint_c, \memTagN{1}[\addr / \szCap])} =
    \lowproj{(\memDataN{1}'[\addr, \addr + \szCap], \taint_c, \memTagN{1}'[\addr / \szCap])}.
  \end{equation*}
  Note that $\addr_0 \in [\addr, \addr + \szCap)$ and by \Cref{lemma:addr-same-taint} there should be $\taint_0 = \taint_c$. Thus, $\lowproj{(\memDataN{1}[\addr_0], \taint_0, \tval)} = \lowproj{(\memDataN{1}'[\addr_0], \taint_0, \tval)}$, so the statement holds.
  \item $\addr_0 \not \in [\addr, \addr + \szCap)$:
  In this case,
  \begin{align*}
  & \stdReadRes{\memDataN{1}[\addr_0, \addr_0 + \sz_0], \taint_0, \memTagN{1}[\addr_0 / \szCap], \sz_0}\\
  = \; & (\memDataN{1}[\addr_0], \taint_0, \tval) = (\memData[\addr_0], \taint_0, \tval)\\
  = \; & \stdReadRes{\memData[\addr_0, \addr_0 + \sz_0], \taint_0, \memTag[\addr_0 / \szCap], \sz_0}.
  \end{align*}
  Similarly,
  \begin{align*}
  & \stdReadRes{\memDataN{1}'[\addr_0, \addr_0 + \sz_0], \taint_0, \memTagN{1}'[\addr_0 / \szCap], \sz_0}\\
  = \; & \stdReadRes{\memData'[\addr_0, \addr_0 + \sz_0], \taint_0, \memTag'[\addr_0 / \szCap], \sz_0}.
  \end{align*}
  Therefore, \eqref{eq:isa-store-hyp} implies \eqref{eq:isa-store-goal}.
\end{itemize}
\end{itemize}

\pgheading{Case $\instr = \jmp{e}$ or $\instr = \beqz{x}{\addr'}$}
The two cases can be proved similarly to the case where $\instr = x \gets e$.
\end{proof}

\begin{lemma}
For all $\mem$, $\reg$, $e$, $\sz$, $\perm'$ such that
\begin{align*}
  & \wf{\mem, \reg} \\
  & \sem{e}_{\reg} = ((\addr, (\perm, \rbase, \rend, \taint_c)), \_, \tcap) \\
  & \stdCheckCap{(\addr, (\perm, \rbase, \rend, \taint_c)), \sz, \perm'},
\end{align*}
then $[\addr, \addr + \sz) \subseteq \taintAddr{\reg, \mem, \taint_c}$.
Furthermore, if for some $\taint$,
\begin{equation*}
  [\addr, \addr + \sz) \cap \taintAddr{\reg, \mem, \taint} \neq \emptyset,
\end{equation*}
then there must be $\taint = \taint_c$.
\label{lemma:addr-same-taint}
\end{lemma}
\begin{proof}
By \Cref{lemma:isa-exp-cap-mono},
$(\addr, (\perm, \rbase, \rend, \taint_c)) \in \rc{\reg, \mem}$.
Then, by \Cref{def:reach-mem}, $[\rbase, \rend) \subseteq \taintAddr{\reg, \mem, \taint_c}$.
Furthermore, by \ref{memory:check-cap}, $\stdCheckCap{(\addr, (\perm, \rbase, \rend, \taint_c)), \sz, \perm'}$ implies that $[\addr, \addr + \sz) \subseteq [\rbase, \rend)$.
Therefore, $[\addr, \addr + \sz) \subseteq \taintAddr{\reg, \mem, \taint_c}$.
By \Cref{def:isa-wf}, there should be $\taint = \taint_c$.
\end{proof}

\begin{lemma}
For all $v, v_m, \taint, c, \taint_c$, denote
\begin{equation*}
((v', v_m'), c') = \stdWriteRes{((v, v_m), \taint, c), \taint_c, \szCap}.
\end{equation*}
Then,
\begin{equation*}
((v, v_m), \taint', c) = \stdReadRes{(v', v_m'), \taint_c, c', \szCap},
\end{equation*}
where
\begin{equation*}
\taint' =
\begin{cases}
\taint_c & c = \tval \land v_m \neq 0 \land \taint = \untainted\\
\taint \sqcap \taint_c & \text{Otherwise}
\end{cases}
\end{equation*}
\label{lemma:read-write-res-szcap}
\end{lemma}
\begin{proof}
Consider the following cases:
\begin{itemize}
\item $c = \tcap$:
Note that $\stdPackTaint{v_m, \taint} \neq \stdMagic$ and
$(v_m, \taint) = \stdUnpackTaint{\stdPackTaint{v_m, \taint}}$.
Then, by \ref{helpersys:write-cap} and \ref{helpersys:read-cap},
\begin{align*}
& \stdReadRes{(v', v_m'), \taint_c, c', \szCap} \\
= \; & \stdReadRes{(v, \stdPackTaint{v_m, \taint}), \taint_c, \tcap, \szCap} \\
= \; & ((v, v_m), \taint \sqcap \taint_c, \tcap) = ((v, v_m), \taint \sqcap \taint_c, c).
\end{align*}
\item $c = \tval \land v_m = 0 \land \taint = \untainted$:
By \ref{helpersys:spill-value},
\begin{equation*}
((v', v_m'), c') = ((v, \stdMagic), \tcap).
\end{equation*}
Then, by \ref{helpersys:unspill-value},
\begin{align*}
& \stdReadRes{(v', v_m'), \taint_c, c', \szCap} \\
=\; & \stdReadRes{(v, \stdMagic), \taint_c, \tcap, \szCap} \\
=\; & ((v, 0), \untainted, \tval) = ((v, v_m), t \sqcap \taint_c, c).
\end{align*}
\item $c = \tval \land (v_m \neq 0 \lor \taint = \tainted)$:
By \ref{helpersys:spill-value},
\begin{equation*}
((v', v_m'), c') = ((v, v_m), \tval).
\end{equation*}
Then, by \ref{helpersys:read-data},
\begin{align*}
& \stdReadRes{(v', v_m'), \taint_c, c', \szCap} \\
=\; & \stdReadRes{(v, v_m), \taint_c, \tval, \szCap} \\
=\; & ((v, v_m), \taint_c, \tval).
\end{align*}
If $\taint = \tainted$, then $\taint_c = \taint \sqcap \taint_c$, so the statement holds.

If $\taint = \untainted$, then by assumption $v_m \neq 0$, thus the statement also holds.
\end{itemize}
\end{proof}

\begin{lemma}
For all $v, v', \taint, \taint', c, c', \taint_c, \sz$ such that
\begin{align*}
  \lowproj{(v, \taint, c)} & = \lowproj{(v', \taint', c')}\\
  \lowproj{(v, \taint_c, c)} & = \lowproj{(v', \taint_c, c')},
\end{align*}
denoting
\begin{align*}
  (v_1, c_1) & = \stdWriteRes{(v, \taint, c), \taint_c, \sz} \\
  (v_1', c_1') & = \stdWriteRes{(v', \taint', c'), \taint_c, \sz},
\end{align*}
then,
\begin{align}
\lowproj{(v_1, \taint_c, c_1)} & = \lowproj{(v_1', \taint_c, c_1')} \label{eq:read-write-goal-1}\\
\lowproj{\stdReadRes{v_1, \taint_c, c_1, \sz}} & =
\lowproj{\stdReadRes{v_1', \taint_c, c_1', \sz}}. \label{eq:read-write-goal-2}
\end{align}
\label{lemma:read-write-lowproj}
\end{lemma}
\begin{proof}
By \Cref{def:low-proj-val}, $\lowproj{(v, \taint, c)} = \lowproj{(v', \taint', c')}$ implies that $\taint = \taint'$.
Consider the following cases:
\begin{itemize}
\item $\sz = \szAddr$:
By \ref{helpersys:write-data},
\begin{align*}
  (v_1, c_1) & = (v, \tval) & (v_1', c_1') & = (v', \tval).
\end{align*}
Thus, $\lowproj{(v, \taint_c, c)} = \lowproj{(v', \taint_c, c')}$ implies \eqref{eq:read-write-goal-1}.

By \ref{helpersys:read-data},
\begin{align*}
  \stdReadRes{v_1, \taint_c, c_1, \sz} & = (v_1, \taint_c, \tval) = (v, \taint_c, \tval)\\
  \stdReadRes{v_1', \taint_c, c_1', \sz} & = (v_1', \taint_c, \tval) = (v', \taint_c, \tval).
\end{align*}
Thus, \eqref{eq:read-write-goal-2} holds.
\item $\sz = \szCap$ and $\taint = \taint' = \untainted$:
In this case, $\lowproj{(v, \taint, c)} = \lowproj{(v', \taint', c')}$ implies that $v = v'$, and $c = c'$.
Then, there should be $v_1 = v_1'$ and $c_1 = c_1'$, which implies \eqref{eq:read-write-goal-1}.

Denote $v = v' = (v_0, v_m)$.
Consider the following two cases:
\begin{itemize}
\item $c = c' = \tval$ and $v_m \neq 0$:
By \Cref{lemma:read-write-res-szcap},
\begin{align*}
(v, \taint_c, c) & = \stdReadRes{v_1, \taint_c, c_1, \sz}\\
(v', \taint_c, c') & = \stdReadRes{v_1', \taint_c, c_1', \sz}.
\end{align*}
Since $\lowproj{(v, \taint_c, c)} = \lowproj{(v', \taint_c, c')}$, \eqref{eq:read-write-goal-2} holds.
\item $c = c' = \tcap$ or $v_m = 0$:
By \Cref{lemma:read-write-res-szcap},
\begin{align*}
(v, \taint \sqcap \taint_c, c) & = \stdReadRes{v_1, \taint_c, c_1, \sz}\\
(v', \taint' \sqcap \taint_c, c') & = \stdReadRes{v_1', \taint_c, c_1', \sz}.
\end{align*}
Since $\lowproj{(v, \taint, c)} = \lowproj{(v', \taint', c')}$ and
$\taint = \taint' = \untainted = \taint \sqcap \taint_c = \taint' \sqcap \taint_c$, \eqref{eq:read-write-goal-2} holds.
\end{itemize}
\item $\sz = \szCap$ and $\taint = \taint' = \tainted$:
First, we prove \eqref{eq:read-write-goal-1}.
If $\taint_c = \tainted$, by the definition of low projection, the statement holds.
If $\taint_c = \untainted$, then $\lowproj{(v, \taint_c, c)} = \lowproj{(v', \taint_c, c)}$ implies that $v=v'$ and $c=c'$. Then, $v_1 = v_1'$, $c_1 = c_1'$, so \eqref{eq:read-write-goal-1} holds.

Second, we prove \eqref{eq:read-write-goal-2}.
By \Cref{lemma:read-write-res-szcap},
\begin{align*}
(v, \taint \sqcap \taint_c, c) & = \stdReadRes{v_1, \taint_c, c_1, \sz}\\
(v', \taint' \sqcap \taint_c, c') & = \stdReadRes{v_1', \taint_c, c_1', \sz}.
\end{align*}
Since $\lowproj{(v, \taint_c, c)} = \lowproj{(v', \taint_c, c')}$ and
$\taint \sqcap \taint_c = \taint' \sqcap \taint_c = \tainted \sqcap \taint_c = \taint_c$, \eqref{eq:read-write-goal-2} holds.
\end{itemize}
\end{proof}

\begin{lemma}[Lemma for $\stdReadMemInst{\cdot}$]
Let $\cf{\mem, \reg}$ and $\cf{\mem', \reg'}$ be two ISA states such that $\cf{\mem, \reg}\pubeq \cf{\mem', \reg'}$.
Let $(\stdCap{\addr}, t, c)$ be a capability (associated with taint and capability tag) such that $\stdCap{\addr}\in \rc{\reg, \mem}$, $\stdCap{\addr} \in \rc{\reg', \mem'}$, and there exists $\instr$ such that
\begin{equation*}
  \instr = \stdReadMemInst{\mem, (\stdCap{\addr}, t, c)}.
\end{equation*}
Then, there must be
\begin{equation*}
  \instr = \stdReadMemInst{\mem', (\stdCap{\addr}, t, c)}.
\end{equation*}
\label{lemma:read-mem-inst}
\end{lemma}
\begin{proof}
By \ref{memory:read-mem-inst}, there should be
\begin{align*}
  & \stdCapMeta{\addr} = (\perm, \rbase, \rend, \untainted) \quad
  t = \untainted \quad
  c = \tcap \\
  & \stdCheckCap{\stdCap{\addr}, \szInst, \exec}
\end{align*}
By \ref{memory:check-cap}, this implies that $\addr \in [\rbase, \rend)$.
Furthermore, by \Cref{def:reach-mem}, $\stdCap{\addr}\in \rc{\reg, \mem}$ implies that $[\rbase, \rend) \subseteq \taintAddr{\reg, \mem, \untainted}$, so $\addr \in \taintAddr{\reg, \mem, \untainted}$.
Similarly, since $\stdCap{\addr}\in \rc{\reg', \mem'}$, there should also be $\addr \in \taintAddr{\reg', \mem', \untainted}$.

Denote $\mem = (\memData, \memTag)$, $\mem' = (\memData', \memTag')$.
Since $\cf{\mem, \reg} \pubeq \cf{\mem', \reg'}$, by \Cref{def:pub-eq-isa}, there should be.
\begin{align*}
  & \stdReadRes{\memData[\addr], \untainted, \memTag[\addr / \szCap], \szAddr} \\
  =\; & \stdReadRes{\memData'[\addr], \untainted, \memTag'[\addr / \szCap], \szAddr}.
\end{align*}
By \ref{helpersys:read-data}, this implies that $\memData[\addr] = \memData'[\addr]$.
Therefore, there should be $\instr = \stdReadMemInst{\mem', (\stdCap{\addr}, t, c)}$.
\end{proof}

\subsection{HW Functional Correctness}
\label{sec:fc-proof}

\begin{definition}[Well-formed HW Configuration]
A HW configuration $C=\cf{\mem, \reg, \buf, \mu}$ is \emph{well-formed} if for all $i \in \dom{\buf}$, entry $\buf(i)$ is well-formed (denoted as $\wf{\mem, \reg, \buf, i}$) so that the following statement holds:

Let $\reg_i = \apl(\pref i \buf, \reg)$.
There must exist $\instr\in \Instrs$ such that $\instr = \stdReadMemInst{\mem, \sem{\pc}_{\reg_i}}$ such that one of the following cases holds:
\begin{itemize}
\item $\instr \not\in \qty{\beqz{x}{\addr}, \jmp{e}}$ and $\buf(i) = \instr\specat{\varepsilon}$.
\item $\instr = x\gets e$, $x\neq \pc$, and $\buf(i) = x\gets v\specat{\varepsilon}$ and $v =\sem{e}_{\reg_i} \neq \bot$.
\item $\instr = \load{x}{e}{\sz}$, $x\neq \pc$, then
\begin{align*}
  & \buf(i) = x\gets (v, t, c)\specat{(\stdCap{\addr}, \sz, j)} \\
  & (\stdCap{\addr}, \_, \tcap) = \sem{e}_{\reg_i} \\
  & \stdCheckCap{\stdCap{\addr}, \sz, \readonly},
\end{align*}
and either
\begin{align*}
  & \buf(j) = \store{(v', t', c')}{\stdCap{\addr'}}{\sz'}\specat{\varepsilon} \\
  & [\addr, \addr + \sz) = [\addr', \addr' + \sz') \\
  & (v, t, c) = \stdReadRes{v', t_c, c', \sz} \\
  & \stdCapMeta{\addr} = (\_, \_, \_, t_c) \\
  & \forall j < k < i,\; \buf(k) = \store{\_}{\stdCap{\addr'}}{\sz'}\specat{\tg} \Rightarrow \\
  & \qquad \qquad \qquad \qquad \qquad [\addr, \addr + \sz)\cap [\addr', \addr' + \sz') = \emptyset
\end{align*}
or
\begin{align*}
  & j = \bot \lor j < \min\dom{\buf} \\
  & (v, t, c) = \stdReadRes{\memData[\addr, \addr + \sz], t_c, \memTag[\addr/ 2], \sz} \\
  & \stdCapMeta{\addr} = (\_, \_, \_, t_c) \\
  & \forall k < i,\; \buf(k) = \store{\_}{\stdCap{\addr'}}{\sz'}\specat{\tg} \Rightarrow \\
  & \qquad \qquad \qquad \qquad \qquad [\addr, \addr + \sz)\cap [\addr', \addr' + \sz') = \emptyset\\
  & \mem = (\memData, \memTag).
\end{align*}
\item $\instr = \store{x}{e}{\sz}$, then
\begin{align*}
  & \buf(i) = \store{(v', t, c')}{\stdCap{\addr}}{\sz}\specat{\varepsilon} \\
  & (\stdCap{\addr}, \_, \tcap) = \sem{e}_{\reg_i} \\
  & \stdCheckCap{\stdCap{\addr}, \sz, \readwrite} \\
  & (v, t, c) = \sem{x}_{\reg_i} \\
  & (v', c') = \stdWriteRes{(v, t, c), t_c, \sz} \\
  & \stdCapMeta{\addr} = (\_, \_, \_, t_c).
\end{align*}
\item $\instr = \jmp{e}$, $\buf(i) = \pc \gets (\addr', \untainted, \tcap)\specat{\tg}$, and $\tg=\varepsilon \Rightarrow \sem{e}_{\reg_i} = (\addr', \_, \tcap)$.
\item $\instr = \beqz{x}{\addr''}$, $\buf(i) = \pc \gets (\addr', t, c)\specat{\tg}$, and $\tg = \varepsilon$ implies that
\begin{align*}
  & (v, \_, \_) = \sem{x}_{\reg_i}
  & & z = (v = 0) \;?\; \addr'' : 1
  & & (\addr', t, c) = \sem{\pc \stdCapOff z}_{\reg_i}.
\end{align*}
\end{itemize}
\label{def:well-formed-buf}
\end{definition}

\begin{definition}[Initial HW State]
An initial HW state is of the form $(\mem, \reg, \bufinit, \muinit)$ where
\begin{itemize}
  \item $\mem$ and $\muinit$ are arbitrary memory and microarchitectural contexts
  \item $\bufinit$ is empty, i.e., $\dom{\buf} = \emptyset$
  \item $\reg$ is a \emph{total} register map (i.e., for all registers $x$, $\reg(x)\neq \bot$),  and $\reg(\pc) = (\texttt{ep}, \untainted, \tcap)$ where $\texttt{ep}$ is the program capability pointing to the entrypoint of the program.
  \item The state is well-formed in the sense of \Cref{def:well-formed-buf}.
  \end{itemize}
\end{definition}

\begin{lemma}[$\apl(\pref i \buf, \reg)$ invariant]
Let $C_i = \cf{\mem_i, \reg_i, \buf_i, \mu_i}$, $i\in \qty{0, 1}$, be two HW configurations such that $C_0$ is well-formed and $C_0 \sstep{}{}C_1$.
For all $i\in \dom{\buf_0}\cap\dom{\buf_1}$, denote $\reg_{i0} = \apl(\pref i {\buf_0}, \reg_0)$, $\reg_{i1} = \apl(\pref i {\buf_1}, \reg_1)$, then
\begin{equation*}
  \forall x,\; \reg_{i0}[x] \neq \bot \Rightarrow \reg_{i0}[x] = \reg_{i1}[x].
\end{equation*}
For convenience, we denote $\reg_{i0} \nextreg \reg_{i1}$ if the above statement holds.
\label{lemma:apl-invariant}
\end{lemma}
\begin{proof}
By \ref{hw:step}, there should be
\begin{mathpar}
\inferrule{
  \mu' = \stdUpdateMu{\mu_0, \lowproj{\buf_0}} \\
  d = \stdNextMu{\mu'} \\
  \cf{\mem_0, \reg_0, \buf_0, \mu'} \sstep{d}{} \cf{\mem_1, \reg_1, \buf_1, \mu_1}
}{
  \cf{\mem_0, \reg_0, \buf_0, \mu_0} \sstep{d}{} \cf{\mem_1, \reg_1, \buf_1, \mu_1}
}
\end{mathpar}
Consider the following cases for $d$:

\pgheading{Case $d = \dfetch$}
By rules for fetch steps, there should be
\begin{align*}
  & \mem_1 = \mem_0 \qquad \reg_1 = \reg_0 \\
  & \buf_1 = \buf_0[\sup\dom{\buf_0} + 1 \mapsto \instr \specat{\tg}].
\end{align*}
Note that $\dom{\buf_1} = \dom{\buf_0} \cup \qty{\sup\dom{\buf_0} + 1}$.
For all $i\in \dom{\buf_0} \cap \dom{\buf_1}$, there should be $i < \sup\dom{\buf_0} + 1$.
Furthermore, for all $j < i$, $\buf_0(j) = \buf_1(j)$.
By the definition of $\apl$, there should be $\apl(\pref i {\buf_0}, \reg_0) = \apl(\pref i {\buf_1}, \reg_1)$. Thus, $\apl(\pref i {\buf_0}, \reg_0) \nextreg \apl(\pref i {\buf_1}, \reg_1)$.

\pgheading{Case $d = \dexecute{i}$}
By rules for execute steps, there should be
\begin{equation*}
  \mem_1 = \mem_0 \qquad \reg_1 = \reg_0.
\end{equation*}
We prove a stronger statement: for all $j \leq \max(\dom{\buf_0} \cap \dom{\buf_1})$, $\apl(\pref j {\buf_0}, \reg_0) \prec \apl(\pref j {\buf_1}, \reg_1)$.

If $j \leq i$, then for all $k < j \leq i$, $\buf_0(k) = \buf_1(k)$.
Then, $\apl(\pref j {\buf_0}, \reg_0) = \apl(\pref j {\buf_1}, \reg_1)$. Thus, there should be $\apl(\pref j {\buf_0}, \reg_0) \nextreg \apl(\pref j {\buf_1}, \reg_1)$.

Next, we use induction to prove that the statement holds for all $j$ where $i < j \leq \max(\dom{\buf_0} \cap \dom{\buf_1})$.

\pgheading{Base case}
Consider $j = i + 1$. by rules for execute steps, there should be $i \in \dom{\buf_0} $ and $i \in \dom{\buf_1}$. Then, by unfolding the recursive definition of $\apl$, we can get
\begin{align*}
  \apl(\pref {i + 1} {\buf_0}, \reg_0) & = \aplinst(\buf_0(i), \apl(\pref i {\buf_0}, \reg_0)) \\
  \apl(\pref {i + 1} {\buf_1}, \reg_1) & = \aplinst(\buf_1(i), \apl(\pref i {\buf_1}, \reg_1)).
\end{align*}
Denote
\begin{align*}
  \reg_{i0} & = \apl(\pref i {\buf_0}, \reg_0) & \reg_{i0}' & = \apl(\pref {i + 1} {\buf_0}, \reg_0) \\
  \reg_{i1} & = \apl(\pref i {\buf_1}, \reg_1) & \reg_{i1}' & = \apl(\pref {i + 1} {\buf_1}, \reg_1).
\end{align*}
According to what we have proved, $\reg_{i0}\nextreg \reg_{i1}$. We aim to show that $\reg_{i0}' \nextreg \reg_{i1}'$.

By rules for execute steps, we just need to consider the following subcases:

\pgheading{Subcase $\buf_0(i) = x \gets e\specat{\varepsilon}$ where $x \neq \pc$}
By \ref{hw:execute-assign}, $\buf_1(i) = x \gets (v, t, c)\specat{\varepsilon}$, where $(v, t, c) = \sem{e}_{\apl(\pref i {\buf_0}, \reg_0)}$.
Thus,
\begin{align*}
  \reg_{i0}' & = \aplinst(\buf_0(i), \reg_{i0}) = \reg_{i0}[\pc \mapsto \sem{\stdNextPc{\pc}}_{\reg_{i0}}][x \mapsto \bot] \\
  \reg_{i1}' & = \aplinst(\buf_1(i), \reg_{i1}) = \reg_{i1}[\pc \mapsto \sem{\stdNextPc{\pc}}_{\reg_{i1}}][x \mapsto (v, t, c)].
\end{align*}
We aim to prove that for all $x_0$, $\reg_{i0}'[x_0] \neq \bot \Rightarrow \reg_{i0}'[x_0] = \reg_{i1}'[x_0]$.
For all registers $x_0$, consider the following cases
\begin{itemize}
\item $x_0 = \pc$: Note that $\reg_{i0}'[\pc] = \sem{\stdNextPc{\pc}}_{\reg_{i0}}=\bot$ iff $\reg_{i0}[\pc] = \sem{\pc}_{\reg_{i0}} = \bot$.
If $\reg_{i0}'[\pc] \neq \bot$, then $\reg_{i0}[\pc] \neq \bot$. Thus, $\reg_{i0}\nextreg \reg_{i1}$ implies that $\reg_{i0}[\pc] = \reg_{i1}[\pc]$. Thus, by the definition of $\sem{\cdot}$, $\sem{\stdNextPc{\pc}}_{\reg_{i0}} = \sem{\stdNextPc{\pc}}_{\reg_{i1}}$. Therefore, $\reg_{i0}'[\pc] = \reg_{i1}'[\pc]$.
\item $x_0 = x$: Note that $\reg_{i0}'[x] = \bot$, so the statement holds.
\item $x_0 \not\in \qty{\pc, x}$: In this case, $\reg_{i0}'[x_0] = \reg_{i0}[x_0]$ and $\reg_{i1}'[x_0] = \reg_{i1}[x_0]$. Note that $\reg_{i0} \nextreg \reg_{i1}$ implies that $\reg_{i0}[x_0] \neq \bot \Rightarrow \reg_{i0}[x_0] = \reg_{i1}[x_0]$. Thus, the statement holds.
\end{itemize}
Therefore, $\reg_{i0}' \nextreg \reg_{i1}'$.

\pgheading{Subcase $\buf_0(i) = \pc \gets \addr'\specat{\stdCap{\addr}}$}
By \ref{hw:execute-jmp-ok}, \ref{hw:execute-jmp-hazard}, \ref{hw:execute-beqz-ok}, and \ref{hw:execute-beqz-hazard}, either of the following cases should hold:
\begin{itemize}
\item $\buf_1 = \buf_0[i \mapsto \pc \gets \addr'\specat{\varepsilon}]$ (branch is not mispredicted): in this case, there should be
\begin{align*}
  \reg_{i0}' & = \aplinst(\buf_0(i), \reg_{i0}) = \reg_{i0}[\pc \mapsto \addr'] \\
  \reg_{i1}' & = \aplinst(\buf_1(i), \reg_{i1}) = \reg_{i1}[\pc \mapsto \addr'].
\end{align*}
We aim to prove that for all $x_0$, $\reg_{i0}'[x_0] \neq \bot \Rightarrow \reg_{i0}'[x_0] = \reg_{i1}'[x_0]$.
For all registers $x_0$, consider the following cases
\begin{itemize}
\item $x_0 = \pc$: Note that $\reg_{i0}'[\pc] = \reg_{i1}'[\pc] = \addr'$, so the statement holds.
\item $x_0 \neq \pc$: Note that $\reg_{i0}'[x_0] = \reg_{i0}[x_0]$, $\reg_{i1}'[x_0] = \reg_{i1}[x_0]$. Furthermore, $\reg_{i0} \nextreg \reg_{i1}$ implies that $\reg_{i0}[x_0]\neq \bot \Rightarrow \reg_{i0}[x_0] = \reg_{i1}[x_0]$. Thus, the statement holds.
\end{itemize}
\item $\buf_1 = \pref {i + 1} {\buf_0}[i \mapsto \pc \gets \addr_0 \specat{\varepsilon}]$ where $\addr_0 \neq \addr'$ (branch is mispredicted):
Note that in this case, $\max(\dom{\buf_0} \cap \dom{\buf_1}) = \max\dom{\buf_1} = i < j$. As we only aim to prove the statement for $i < j \leq \max(\dom{\buf_0} \cap \dom{\buf_1})$, we do not need to reason about this case.
\end{itemize}

\pgheading{Subcase $\buf_0(i) = \load{x}{e}{\sz}\specat{\varepsilon}$}
By \ref{hw:execute-load-fwd} and \ref{hw:execute-load-mem}, there always exists $(v, t, c)$ and $\tg$ such that
$\buf_1 = \buf_0[i \mapsto (x \gets (v, t, c))\specat{\tg}]$.
Thus,
\begin{align*}
  \reg_{i0}' & = \aplinst(\buf_0(i), \reg_{i0}) = \reg_{i0}[\pc \mapsto \sem{\stdNextPc{\pc}}_{\reg_{i0}}][x \mapsto \bot] \\
  \reg_{i1}' & = \aplinst(\buf_1(i), \reg_{i1}) = \reg_{i1}[\pc \mapsto \sem{\stdNextPc{\pc}}_{\reg_{i1}}][x \mapsto (v, t, c)].
\end{align*}
Thus, similar to the subcase for execute an assign instruction, we can derive that $\reg_{i0}' \nextreg \reg_{i1}'$.

\pgheading{Subcase $\buf_0(i) = \store{x}{e}{\sz}\specat{\varepsilon}$}
By \ref{hw:execute-store-ok} and \ref{hw:execute-store-hazard}, either of the following cases will hold:
\begin{itemize}
\item $\buf_1 = \buf_0[i \mapsto (\store{v}{\stdCap{\addr}}{\sz})\specat{\varepsilon}]$ (no load needs to be squashed): in this case, there should be
\begin{align*}
  \reg_{i0}' & = \aplinst(\buf_0(i), \reg_{i0}) = \reg_{i0}[\pc \mapsto \sem{\stdNextPc{\pc}}_{\reg_{i0}}] \\
  \reg_{i1}' & = \aplinst(\buf_1(i), \reg_{i1}) = \reg_{i1}[\pc \mapsto \sem{\stdNextPc{\pc}}_{\reg_{i1}}].
\end{align*}
We aim to prove that for all $x_0$, $\reg_{i0}'[x_0] \neq \bot \Rightarrow \reg_{i0}'[x_0] = \reg_{i1}'[x_0]$. For all registers $x_0$, consider the following cases:
\begin{itemize}
\item $x_0 = \pc$: Note that $\reg_{i0}'[\pc] = \sem{\stdNextPc{\pc}}_{\reg_{i0}}=\bot$ iff $\reg_{i0}[\pc] = \sem{\pc}_{\reg_{i0}} = \bot$.
If $\reg_{i0}'[\pc] \neq \bot$, then $\reg_{i0}[\pc] \neq \bot$. Thus, $\reg_{i0}\nextreg \reg_{i1}$ implies that $\reg_{i0}[\pc] = \reg_{i1}[\pc]$. Thus, by the definition of $\sem{\cdot}$, $\sem{\stdNextPc{\pc}}_{\reg_{i0}} = \sem{\stdNextPc{\pc}}_{\reg_{i1}}$. Therefore, $\reg_{i0}'[\pc] = \reg_{i1}'[\pc]$.
\item $x_0 \neq \pc$: In this case, $\reg_{i0}'[x_0] = \reg_{i0}[x_0]$ and $\reg_{i1}'[x_0] = \reg_{i1}[x_0]$. Note that $\reg_{i0} \nextreg \reg_{i1}$ implies that $\reg_{i0}[x_0] \neq \bot \Rightarrow \reg_{i0}[x_0] = \reg_{i1}[x_0]$. Thus, the statement holds.
\end{itemize}
\item $\buf_1 = \pref j {\buf_0}[i \mapsto (\store{v}{\stdCap{\addr}}{\sz})\specat{\varepsilon}]$ for some $j > i$ (some load is squashed): Since $j > i$, $\reg_{i1}'$ has the same format as in the last case. Thus, the statement still holds.
\end{itemize}

\pgheading{Induction Step}
Consider $j$ such that $i < j \leq \max(\dom{\buf_0} \cap \dom{\buf_1})$.
Denote $\reg_{j0} = \apl(\pref j {\buf_0}, \reg_0)$, $\reg_{j1} = \apl(\pref j {\buf_1}, \reg_1)$ and assume $\reg_{j0} \nextreg \reg_{j1}$.
We aim to prove that $\apl(\pref {j + 1} {\buf_0}, \reg_0) \nextreg \apl(\pref {j + 1} {\buf_1}, \reg_1)$.
If $j \not \in \dom{\buf_0} = \dom{\buf_1}$, then $\pref{j+1}{\buf_0} = \pref j {\buf_0}$ and $\pref{j+1}{\buf_1} = \pref j {\buf_1}$. Thus, the statement holds.

By rules for execute step, there should be
\begin{itemize}
\item either $\dom{\buf_0} = \dom{\buf_1}$ (if no hazard)
\item or there exists $k$ such that $\qty{j: j < k \land j \in \dom{\buf_0}} = \dom{\buf_1}$ (if hazard).
\end{itemize}
This implies that $j\in \dom{\buf_0} \land j \leq \min(\dom{\buf_0} \cap \dom{\buf_1})$ iff $j \in \dom{\buf_j}$.

Then, if $j \in \dom{\buf_0}$, since we only consider the case where $j \leq \max(\dom{\buf_0} \cap \dom{\buf_1})$, we have $j \in \dom{\buf_1}$. Similarly, by unfolding the recursive definition of $\apl$, we can get
\begin{align*}
  \apl(\pref {j + 1} {\buf_0}, \reg_0) & = \aplinst(\buf_0(j), \apl(\pref j {\buf_0}, \reg_0)) \\
  & = \aplinst(\buf_0(j), \reg_{j0}) \\
  \apl(\pref {j + 1} {\buf_1}, \reg_1) & = \aplinst(\buf_1(j), \apl(\pref j {\buf_1}, \reg_1)) \\
  & = \aplinst(\buf_1(j), \reg_{j1}).
\end{align*}
Since $j > i$, then by rules for execution steps, $\buf_0(j) = \buf_1(j)$.
By discussing the instruction type of $\buf_0(j) = \buf_1(j)$, it is straightforward to derive that $\apl(\pref {j + 1} {\buf_0}, \reg_0) \nextreg \apl(\pref {j + 1} {\buf_1}, \reg_1)$.

\pgheading{Case $d = \dcommit$}
By rules for commit steps, $\buf_0$ is not empty.
Denote $i = \min\dom{\buf_0}$. Then, there should be
$\buf_1 = \buf_0 \backslash i$, which implies that for all $j$, $\pref j {\buf_1} = \pref j {\buf_0} \backslash i$.
Furthermore, there should be
\begin{equation*}
  \dom{\buf_0} \cap \dom{\buf_1} = \dom{\buf_0} \backslash \qty{i} = \dom{\buf_1}.
\end{equation*}
If $\dom{\buf_0} \cap \dom{\buf_1} = \emptyset$, the statement automatically hold.

We consider the case where $\dom{\buf_0} \cap \dom{\buf_1} \neq \emptyset$.
It suffices to prove that for all $j\in \dom{\buf_0} \cap \dom{\buf_1}$,
\begin{equation*}
  \apl(\pref j {\buf_0}, \reg_0) = \apl(\pref j {\buf_1}, \reg_1).
\end{equation*}
Note that there should be $j > i$, so
\begin{align*}
  \apl(\pref j {\buf_0}, \reg_0)
  & = \apl(\pref j {\buf_0} \backslash i, \aplinst(\buf_0(i), \reg_0)) \\
  & = \apl(\pref j {\buf_1}, \aplinst(\buf_0(i), \reg_0)).
\end{align*}
Hence, we just need to prove that $\reg_1 = \aplinst(\buf_0(i), \reg_0)$.
Depending on the type of the committed instruction, there will be the following two cases:

\pgheading{Subcase $\buf_0(i) = x \gets v \specat{\tg}$}
By rule \ref{hw:commit-assign} and the definition of $\aplinst$, there should be
\begin{equation*}
  \reg_1 = \reg_0[\pc \mapsto \sem{\stdNextPc{\pc}}_{\reg_0}][x \mapsto v]
  = \aplinst(\buf_0(i), \reg_0).
\end{equation*}

\pgheading{Subcase $\buf_0(i) = \store{(v, t, c)}{\stdCap{\addr}}{\sz}\specat{\varepsilon}$}
By rule \ref{hw:commit-store} and the definition of $\aplinst$, there should be
\begin{equation*}
  \reg_1 = \reg_0[\pc \mapsto \sem{\stdNextPc{\pc}}_{\reg_0}] = \aplinst(\buf_0(i), \reg_0).
\end{equation*}
\end{proof}

\begin{lemma}
Let $\reg_0$ and $\reg_1$ be two register files such that $\reg_0 \nextreg \reg_1$.
For all $e$, if $\sem{e}_{\reg_0} \neq \bot$, then $\sem{e}_{\reg_0} = \sem{e}_{\reg_1}$.
\label{lemma:sem-nextreg-invariant}
\end{lemma}
\begin{proof}
  By induction on the syntax of $e$.
\end{proof}

\begin{lemma}
For all $\buf$, $\reg$, $e$ such that $\sem{e}_{\aplsan(\buf, \reg)} \neq \bot$, then there must be $\sem{e}_{\aplsan(\buf, \reg)} = \sem{e}_{\apl(\buf, \reg)}$.
\label{lemma:sem-apl-aplsan-eq}
\end{lemma}
\begin{proof}
If $\dom{\buf} = \emptyset$, then by definition, $\aplsan(\buf, \reg) = \apl(\buf, \reg)$. Thus, the statement holds.

If $\dom{\buf} \neq \emptyset$, then $\aplsan(\buf, \reg) = \lowproj{\apl(\buf, \reg)}$.
By induction on $e$, it is straightforward to derive that $\sem{e}_{\aplsan(\buf, \reg)} = \lowproj{\sem{e}_{\apl(\buf, \reg)}}$.
Consider the following cases:
\begin{itemize}
  \item $\sem{e}_{\apl(\buf, \reg)} = (v, \tainted, c)$: Then, $\sem{e}_{\aplsan(\buf, \reg)} = \bot$, and we do not consider this case.
  \item $\sem{e}_{\apl(\buf, \reg)} = (v, \untainted, c)$: Then, $\sem{e}_{\aplsan(\buf, \reg)} = (v, \untainted, c) = \sem{e}_{\apl(\buf, \reg)}$.
  \item $\sem{e}_{\apl(\buf, \reg)} = \bot$: Then, $\sem{e}_{\aplsan(\buf, \reg)} = \bot$, and we do not consider this case.
\end{itemize}
Therefore, the statement holds.
\end{proof}

\begin{lemma}[Well-formed Prefix Buf]
For all $\mem$, $\reg$, $\buf$, $i$,
\begin{equation*}
  \wf{\mem, \reg, \pref {i + 1} \buf, i} \qquad \Leftrightarrow \qquad \wf{\mem, \reg, \buf, i}
\end{equation*}
\label{lemma:wf-prefix}
\end{lemma}
\begin{proof}
It is straightforward to derive this from \Cref{def:well-formed-buf}.
\end{proof}

\begin{lemma}[HW State Well-formed Invariant]
Let $C_0$, $C_1$ be two HW configurations such that $C_0$ is a well-formed HW config and $C_0 \sstep{}{}C_1$.
Furthermore, denote $C_0 = \cf{\mem_0, \reg_0, \buf_0, \mu_0}$, assume there is also $\wf{\mem_0, \reg_0}$.
Then, $C_1$ is also a well-formed HW config.
\label{lemma:hw-wf-invariant}
\end{lemma}
\begin{proof}
Denote $C_1 = \cf{\mem_1, \reg_1, \buf_1, \mu_1}$.
By \ref{hw:step}, there should be
\begin{mathpar}
\inferrule{
  \mu' = \stdUpdateMu{\mu_0, \lowproj{\buf_0}} \\
  d = \stdNextMu{\mu'} \\
  \cf{\mem_0, \reg_0, \buf_0, \mu'} \sstep{d}{} \cf{\mem_1, \reg_1, \buf_1, \mu_1}
}{
  \cf{\mem_0, \reg_0, \buf_0, \mu_0} \sstep{d}{} \cf{\mem_1, \reg_1, \buf_1, \mu_1}
}
\end{mathpar}
Consider the following cases for $d$:

\pgheading{Case $d = \dfetch$}
By rules for fetch steps, there must exist $\instr$ such that $\instr = \stdReadMemInst{\mem_0, \sem{\pc}_{\aplsan(\buf_0, \reg_0)}}$.
Note that this implies that $\sem{\pc}_{\aplsan(\buf_0, \reg_0)} \neq \bot$.
By \Cref{lemma:sem-apl-aplsan-eq}, there should be
\begin{equation*}
  \instr = \stdReadMemInst{\mem_0, \sem{\pc}_{\apl(\buf_0, \reg_0)}}.
\end{equation*}
Furthermore, there should also be $\mem_1 = \mem_0$, $\reg_1 = \reg_0$.
Denote $i = \sup\dom{\buf}$.
Consider the following two cases:

\pgheading{Subcase $\instr \not\in \qty{\beqz{x}{\addr''}, \jmp{e}}$}
By \ref{hw:fetch-other}, we have
\begin{equation*}
  \buf_1 = \buf_0[i + 1 \mapsto \instr \specat{\varepsilon}].
\end{equation*}
This implies that $\buf_0 = \pref {i + 1} {\buf_1}$.
For all $j \in \dom{\buf_1}$, consider the following two cases:
\begin{itemize}
\item $j \leq i$: By \Cref{def:well-formed-buf}, $S_0$ is well-formed implies that $\wf{\mem_0, \reg_0, \buf_0, j}$. Thus, by the equivalence we have shown above, there should be $\wf{\mem_1, \reg_1, \pref {i + 1} {\buf_1}, j}$.
By \Cref{lemma:wf-prefix}, there should be $\wf{\mem_1, \reg_1, \buf_1, j}$.
\item $j = i$: Since $\buf_1(i) = \instr \specat{\varepsilon}$, then by \Cref{def:well-formed-buf}, there should be $\wf{\mem_1, \reg_1, \buf_1, i}$.
\end{itemize}

\pgheading{Subcase $\instr \in \qty{\beqz{x}{\addr''}, \jmp{e}}$}
By \ref{hw:fetch-branch-predict-pc} and \ref{memory:read-mem-inst}, we have
\begin{align*}
  & \buf_1 = \buf_0[i + 1 \mapsto \pc \gets ((\stdPredPc{\mu'}, \stdCapMeta{\addr}), t, c)\specat{\stdCap{\addr}}] \\
  & (\stdCap{\addr}, t, c) = \sem{\pc}_{\aplsan(\buf_0, \reg_0)} \\
  & t = \untainted.
\end{align*}
This implies that $\buf_0 = \pref {i + 1} {\buf_1}$.
For all $j \in \dom{\buf_1}$, consider the following two cases:
\begin{itemize}
\item $j \leq i$: we can prove $\wf{\mem_1, \reg_1, \buf_1, j}$ using the similar reasoning for the non-branch case.
\item $j = i$: Note that
\begin{align*}
  & \buf_1(i) = \pc \gets ((\stdPredPc{\mu'}, \stdCapMeta{\addr}), t, c)\specat{\stdCap{\addr}} \\
  & t = \untainted.
\end{align*}
By \Cref{def:well-formed-buf}, there should be $\wf{\mem_1, \reg_1, \buf_1, i}$.
\end{itemize}

\pgheading{Case $d = \dexecute{i}$}
We aim to prove that for all $j\in\dom{\buf_1}$, $\wf{\mem_1, \reg_1, \buf_1, j}$.
We discuss based on the range of $j$ ($j < i$, $j = i$, and $j > i$).

\pgheading{Subcase $j < i$}
By rules for execute steps, there should be
\begin{align*}
  \mem_0 & = \mem_1 & \reg_0 & = \reg_1 & \pref i {\buf_0} & = \pref i {\buf_1}.
\end{align*}
Then, for all $j < i$, $j \in \dom{buf_1}$ (naturally, there should also be $j \in \dom{\buf_0}$), we have
\begin{align*}
  \wf{\mem_0, \reg_0, \buf_0, j}
  & \Leftrightarrow \wf{\mem_0, \reg_0, \pref{j+1}{\buf_0}, j} \\
  & \Leftrightarrow  \wf{\mem_1, \reg_1, \pref{j+1}{\buf_1}, j} \\
  & \Leftrightarrow \wf{\mem_1, \reg_1, \buf_1, j}.
\end{align*}
Since $S_0$ is well-formed, then $\wf{\mem_0, \reg_0, \buf_0, j}$. Thus, there should be $\wf{\mem_1, \reg_1, \buf_1, j}$.

\pgheading{Subcase $j = i$}
By rules for execute steps, there should be
\begin{align*}
  & \pref i {\buf_0} = \pref i {\buf_1} && i \in \buf_0,\buf_1 \\
  & \reg_0 = \reg_1 && \mem_0 = \mem_1.
\end{align*}
Hence, we can denote
\begin{equation*}
  \reg_i = \apl(\pref i {\buf_0}, \reg_0) = \apl(\pref i {\buf_1}, \reg_1).
\end{equation*}
Since $S_0$ is well-formed, then $\wf{\mem_0, \reg_0, \buf_0, i}$.
Hence, there exists $\instr \in \Instrs$ such that $\instr = \stdReadMemInst{\mem_0, \sem{\pc}_{\reg_i}} = \stdReadMemInst{\mem_1, \sem{\pc}_{\reg_i}}$.

To prove $\wf{\mem_1, \reg_1, \buf_1, j}$, by rules for execute steps, we just need to consider the following cases:
\begin{itemize}
\item $\buf_0(i) = x \gets e \specat{\varepsilon}$ and $x\neq \pc$:
$\wf{\mem_0, \reg_0, \buf_0, i}$ implies that $\instr = x \gets e$.
By \ref{hw:execute-assign}, there should be
\begin{align*}
  \buf_1(i) & = x \gets (v, t, c) \specat{\varepsilon}
  & (v, t, c) & = \sem{e}_{\apl(\pref i {\buf_0}, \reg_0)} \neq \bot.
\end{align*}
Thus, by \Cref{def:well-formed-buf}, $\wf{\mem_1, \reg_1, \buf_1, i}$.

\item $\buf_0(i) = \pc \gets (\addr', t', c')\specat{\stdCap{\addr}}$:
$\wf{\mem_0, \reg_0, \buf_0, i}$ implies that either of the following cases holds:
\begin{itemize}
  \item $\instr = \jmp{e}$: By \ref{hw:execute-jmp-ok}, \ref{hw:execute-jmp-hazard}, and \Cref{lemma:sem-apl-aplsan-eq}, there should always be
  \begin{align*}
    & \buf_1(i) = \pc \gets (\addr_0, \untainted, c)\specat{\varepsilon} \\
    & (\addr_0, t, c) = \sem{e}_{\aplsan(\pref i {\buf_0}, \reg_0)} = \sem{e}_{\reg_i}.
  \end{align*}
  By \Cref{def:well-formed-buf}, $\wf{\mem_1, \reg_1, \buf_1, i}$ holds.
  \item $\instr = \beqz{x}{\addr''}$: By \ref{hw:execute-beqz-ok}, \ref{hw:execute-beqz-hazard}, and \Cref{lemma:sem-apl-aplsan-eq}, there should always be
  \begin{align*}
    & \buf_1(i) = \pc \gets \addr_0 \specat{\varepsilon} \\
    & (v, \_, \_) = \sem{x}_{\aplsan(\pref i {\buf_0}, \reg_0)}= \sem{x}_{\reg_i} \\
    & z = (v = 0) \;?\; \addr'' : 1 \\
    & \addr_0 = \sem{\pc \stdCapOff z}_{\reg_i}.
  \end{align*}
  By \Cref{def:well-formed-buf}, $\wf{\mem_1, \reg_1, \buf_1, i}$ holds.
\end{itemize}

\item $\buf_0(i) = \load{x}{e}{\sz}\specat{\varepsilon}$
By \ref{hw:execute-load-mem} and \ref{hw:execute-load-fwd}, either of the following cases should hold:
\begin{itemize}
  \item Load data is read from memory (i.e., \ref{hw:execute-load-mem} applies):
  In this case, (also by \Cref{lemma:sem-apl-aplsan-eq}) there should be
  \begin{align*}
    & \buf_1(i) = x \gets (v, t, c)\specat{((\addr, (\perm, \rbase, \rend, t_c)), \sz, \bot)} \\
    & ((\addr, (\perm, \rbase, \rend, t_c)), \_, \tcap) = \sem{e}_{\aplsan(\pref i {\buf_0}, \reg_0)} = \sem{e}_{\reg_i} \\
    & \stdCheckCap{(\addr, (\perm, \rbase, \rend, t_c)), \sz, \readonly} \\
    & (v, t, c) = \stdReadRes{\memData[\addr, \addr + sz], t_c, \memTag[\addr / 2], \sz} \\
    & \forall j < i, \buf_0(j) = \store{\_}{\stdCap{\addr'}}{\sz'}\specat{\tg} \\
    & \qquad \qquad \Rightarrow [\addr, \addr + \sz)\cap [\addr', \addr' + \sz') = \emptyset \\
    & \forall j < i, \buf_1(j) = \buf_0(j) \\
    & (\memData, \memTag) = \mem_0 = \mem_1 \\
  \end{align*}
  By \Cref{def:well-formed-buf}, $\wf{\mem_1, \reg_1, \buf_1, i}$ holds.
  \item Load data is forwarded from a store (i.e., \ref{hw:execute-load-fwd} applies):
  In this case, (also by \Cref{lemma:sem-apl-aplsan-eq}), there should be
  \begin{align*}
    & \buf_1(i) = x \gets (v, t, c)\specat{((\addr, (\perm, \rbase, \rend, t_c)), \sz, \bot)} \\
    & ((\addr, (\perm, \rbase, \rend, t_c)), \_, \tcap) = \sem{e}_{\aplsan(\pref i {\buf_0}, \reg_0)} = \sem{e}_{\reg_i} \\
    & \stdCheckCap{(\addr, (\perm, \rbase, \rend, t_c)), \sz, \readonly} \\
    & \buf(j) = \store{v', t', c'}{\stdCap{\addr'}}{\sz}\specat{\tg} \\
    & [\addr, \addr + \sz) = [\addr', \addr' + \sz') \\
    & (v, t, c) = \stdReadRes{v', t_c, c', \sz} \\
    & j = \max\lbrace j < i: \buf(j)=\store{\_}{\stdCap{\addr'}}{\sz'}\specat{\tg} \\
    & \qquad \qquad \land [\addr, \addr + \sz) \cap [\addr', \addr' + \sz') \neq \emptyset \rbrace
  \end{align*}
  Note that the last line implies that
  \begin{align*}
    & \forall j < k < i, \buf(k) = \store{\_}{\stdCap{\addr'}}{\sz'}\specat{\tg} \Rightarrow \\
    & \qquad \qquad \qquad \qquad [\addr, \addr+ \sz) \land [\addr', \addr' + \sz') = \emptyset.
  \end{align*}
  By \Cref{def:well-formed-buf}, $\wf{\mem_1, \reg_1, \buf_1, i}$ holds.
\end{itemize}

\item $\buf_0(i) = \store{x}{e}{\sz}\specat{\varepsilon}$:
By \ref{hw:execute-store-ok}, \ref{hw:execute-store-hazard}, and \Cref{lemma:sem-apl-aplsan-eq}, there should always be
\begin{align*}
  & \buf_1(i) = \store{v', t, c'}{(\addr, (\perm, \rbase, \rend, t_c))}{\sz}\specat{\varepsilon} \\
  & ((\addr, (\perm, \rbase, \rend, t_c)), \_, \tcap) = \sem{e}_{\aplsan(\pref i {\buf_0}, \reg_0)} = \sem{e}_{\reg_i} \\
  & \stdCheckCap{(\addr, (\perm, \rbase, \rend, t_c)), \sz, \readwrite} \\
  & (v, t, c) = \sem{x}_{\reg_i} \\
  & (v', c') = \stdWriteRes{(v, t, c), t_c, \sz}.
\end{align*}
By \Cref{def:well-formed-buf}, $\wf{\mem_1, \reg_1, \buf_1, i}$ holds.

\end{itemize}

\pgheading{Subcase $j > i$}
By rules for execute steps, $j>i$ and $j\in \dom{\buf_1}$ implies that $\buf_0(j) = \buf_1(j)$.
Since $S_0$ is well-formed, then $\wf{\mem_0, \reg_0, \buf_0, j}$.

Denote
\begin{align*}
  \reg_{j0} & = \apl(\pref j {\buf_0}, \reg_0) & \reg_{j1} & = \apl(\pref j {\buf_1}, \reg_1).
\end{align*}
By rules for execute steps, there should also be $\mem_0 = \mem_1$ and $\reg_0 = \reg_1$.
By \Cref{def:well-formed-buf}, there exists $\instr$ such that $\instr = \stdReadMemInst{\mem_0, \sem{\pc}_{\reg_{j0}}}$.
This implies that $\sem{\pc}_{\reg_{j0}} \neq \bot$.
Furthermore, by \Cref{lemma:apl-invariant}, $\reg_{j0} \nextreg \reg_{j1}$.
Thus, by \Cref{lemma:sem-nextreg-invariant}, $\sem{\pc}_{\reg_{j0}} = \sem{\pc}_{\reg_{j1}}$.
Therefore, $\instr = \stdReadMemInst{\mem_1, \sem{\pc}_{\reg_{j1}}}$.

Then, we just need to consider the following cases
\begin{itemize}
\item $\instr \not\in \qty{\beqz{x}{\addr''}, \jmp{e}}$ and $\buf_0(j) = \instr\specat{\varepsilon}$:
Since $\buf_0(j) = \buf_1(j)$, then $\wf{\mem_1, \reg_1, \buf_1, j}$.
\item $\instr = x \gets e$, $x \neq \pc$, $\buf_0(j) = x \gets v \specat{\varepsilon}$, and $v = \sem{e}_{\reg_{j0}} \neq \bot$:
By \Cref{lemma:sem-nextreg-invariant}, $\sem{e}_{\reg_{j0}} = \sem{e}_{\reg_{j1}}$.
Furthermore, note that $\buf_0(j) = \buf_1(j)$. Thus, $\wf{\mem_1, \reg_1, \buf_1, j}$.
\item $\instr = \load{x}{e}{\sz}$, $x \neq \pc$:
$\wf{\mem_0, \reg_0, \buf_0, j}$ further implies that either of the following cases holds:
\begin{itemize}
\item Load data is forwarded from an in-flight store: in this case, there should be
\begin{equation}
\begin{aligned}
  & \buf_0(j) = x \gets(v, t, c) \specat{(\stdCap{\addr}, \sz, j')} \\
  & (\stdCap{\addr}, \_, \tcap) = \sem{e}_{\reg_{j0}} \\
  & \stdCheckCap{\stdCap{\addr}, \sz, \readonly} \\
  & \buf_0(j') = \store{(v', t', c')}{\stdCap{\addr'}}{\sz'} \specat{\varepsilon} \\
  & [\addr, \addr + \sz) = [\addr', \addr' + \sz') \\
  & (v, t, c) = \stdReadRes{v', t_c, c', \sz} \\
  & \stdCapMeta{\addr} = (\_, \_, \_, t_c) \\
  & \forall j' < k < j,\; \buf_0(k) = \store{\_}{\stdCap{\addr'}}{\sz'}\specat{\tg} \Rightarrow \\
  & \qquad \qquad \qquad [\addr, \addr + \sz) \cap [\addr', \addr' + \sz') = \emptyset.
\end{aligned}
\label{eq:wf-exe-load-hyp-1}
\end{equation}
To prove $\wf{\mem_1, \reg_1, \buf_1, j}$, it suffices to prove that \eqref{eq:wf-exe-load-hyp-1} still holds if we replace $\buf_0$ with $\buf_1$, and $\reg_{j0}$ with $\reg_{j1}$.
Note that $\buf_0(j) = \buf_1(j)$. By rules for execute steps, since $\buf_0(j')$ is an executed store instruction, then $j' \neq i$. Hence, $\buf_0(j') = \buf_1(j')$.
Furthermore, by \Cref{lemma:sem-nextreg-invariant}, there should also be $\sem{e}_{\reg_{j0}} = \sem{e}_{\reg_{j1}}$.
Thus, we just need to prove that for all $j$ such that $j' < k < j$, the following holds.
\begin{equation}
\begin{aligned}
& \buf_1(k) = \store{\_}{\stdCap{\addr'}}{\sz'}\specat{\tg} \Rightarrow \\
& \qquad \qquad \qquad [\addr, \addr + \sz) \cap [\addr', \addr' + \sz') = \emptyset.
\end{aligned}
\label{eq:wf-exe-load-goal-1}
\end{equation}
We prove by contradiction. Suppose there exists $k$, $j' < k < j$ such that \eqref{eq:wf-exe-load-goal-1} does not hold.
Then, there must be $k = i$, otherwise $\buf_1(k) = \buf_0(k)$ and \eqref{eq:wf-exe-load-goal-1} should be implied by \eqref{eq:wf-exe-load-hyp-1}.
Furthermore, by \ref{hw:execute-store-ok} or \ref{hw:execute-store-hazard}, if there is aliasing between the store at entry $i$ and the load at entry $j$, then entry $j$ must be squashed in $\buf_1$, i.e., $j \not \in \dom{\buf_1}$. This contradicts our assumption that $j \in \dom{\buf_1}$.
Thus, \eqref{eq:wf-exe-load-goal-1} holds.
\item Load data is loaded from memory or forwarded from a committed store: in this case, there should be
\begin{equation}
\begin{aligned}
  & \buf_0(j) = x \gets(v, t, c) \specat{(\stdCap{\addr}, \sz, j')} \\
  & (\stdCap{\addr}, \_, \tcap) = \sem{e}_{\reg_{j0}} \\
  & \stdCheckCap{\stdCap{\addr}, \sz, \readonly} \\
  & j' = \bot \lor j' < \min \dom{\buf_0} \\
  & (v, t, c) = \stdReadRes{\memData[\addr, \addr + \sz], t_c, \memTag[\addr / 2], \sz} \\
  & \stdCapMeta{\addr} = (\_, \_, \_, t_c) \\
  & \forall k < j,\; \buf_0(k) = \store{\_}{\stdCap{\addr'}}{\sz'}\specat{\tg} \Rightarrow \\
  & \qquad \qquad \qquad [\addr, \addr + \sz) \cap [\addr', \addr' + \sz') = \emptyset \\
  & \mem_0 = (\memData, \memTag).
\end{aligned}
\label{eq:wf-exe-load-hyp-2}
\end{equation}
To prove $\wf{\mem_1, \reg_1, \buf_1, j}$, it suffices to prove that the above statement still holds if we replace $\buf_0$ with $\buf_1$, $\reg_{j0}$ with $\reg_{j1}$, and $\mem_0$ with $\mem_1$.
By rules for execution steps, there should be $\buf_0(j) = \buf_1(j)$, $\min \dom{\buf_0} = \min\dom{\buf_1}$, and $\mem_0 = \mem_1$.
Furthermore, by \Cref{lemma:sem-nextreg-invariant}, there should also be $\sem{e}_{\reg_{j0}} = \sem{e}_{\reg_{j1}}$.
Thus, we just need to prove that for all $j$ such that $k < j$, the following holds.
\begin{equation}
\begin{aligned}
& \buf_1(k) = \store{\_}{\stdCap{\addr'}}{\sz'}\specat{\tg} \Rightarrow \\
& \qquad \qquad \qquad [\addr, \addr + \sz) \cap [\addr', \addr' + \sz') = \emptyset.
\end{aligned}
\label{eq:wf-exe-load-goal-2}
\end{equation}
Similar to the last case, we prove by contradiction. Suppose there exists $k$, $k < j$ such that \eqref{eq:wf-exe-load-goal-2} does not hold.
Then, there must be $k = i$, otherwise $\buf_1(k) = \buf_0(k)$ and \eqref{eq:wf-exe-load-goal-2} should be implied by \eqref{eq:wf-exe-load-hyp-2}.
Furthermore, by \ref{hw:execute-store-ok} or \ref{hw:execute-store-hazard}, if there is aliasing between the store at entry $i$ and the load at entry $j$, then entry $j$ must be squashed in $\buf_1$, i.e., $j \not \in \dom{\buf_1}$. This contradicts our assumption that $j \in \dom{\buf_1}$.
Thus, \eqref{eq:wf-exe-load-goal-2} holds.
\end{itemize}

\item $\instr = \store{x}{e}{\sz}$:
$\wf{\mem_0, \reg_0, \buf_0, j}$ further implies that
\begin{align*}
  & \buf_0(j) = \store{(v', t, c')}{\stdCap{\addr}}{\sz}\specat{\varepsilon} \\
  & (\stdCap{\addr}, \_, \tcap) = \sem{e}_{\reg_{j0}} \\
  & \stdCheckCap{\stdCap{\addr}, \sz, \readwrite} \\
  & (v, t, c) = \sem{x}_{\reg_{j0}} \\
  & (v', c') = \stdWriteRes{(v, t, c), t_c, \sz} \\
  & \stdCapMeta{\addr} = (\_, \_, \_, t_c).
\end{align*}
Note that $\buf_1(j) = \buf_0(j)$.
By \Cref{lemma:sem-nextreg-invariant}, $\sem{e}_{\reg_{j0}} = \sem{e}_{\reg_{j1}}$ and $\sem{x}_{\reg_{j0}} = \sem{x}_{\reg_{j1}}$.
Furthermore, $\buf_1(j) = \buf_0(j)$.
Thus, the above statement still holds if we replace $\buf_0$ with $\buf_1$, and $\reg_{j0}$ with $\reg_{j1}$.
Therefore, $\wf{\mem_1,\reg_1, \buf_1, j}$.
\item $\instr = \jmp{e}$:
$\wf{\mem_0, \reg_0, \buf_0, j}$ further implies that
\begin{align*}
  & \buf_0(j) = \pc \gets (\addr', \untainted, c)\specat{\tg} \\
  & \tg = \varepsilon \Rightarrow \sem{e}_{\reg_{j0}} = (\addr', \_, c).
\end{align*}
By \Cref{lemma:sem-nextreg-invariant}, $\sem{e}_{\reg_{j0}} = \sem{e}_{\reg_{j1}}$.
Furthermore, note that $\buf_0(j) = \buf_1(j)$.
Thus, the above statement still holds if we replace $\buf_0$ with $\buf_1$, and $\reg_{j0}$ with $\reg_{j1}$.
Therefore, there should be $\wf{\mem_1,\reg_1, \buf_1, j}$.
\item $\instr = \beqz{x}{\addr''}$:
$\wf{\mem_0, \reg_0, \buf_0, j}$ further implies that
\begin{align*}
  & \buf_0(j) = \pc \gets (\addr', t, c)\specat{\tg} \\
  & \tg = \varepsilon \Rightarrow \\
  & \quad ((v, \_, \_) = \sem{x}_{\reg_{j0}} \land (z = (v=0)\;?\; \addr' : 1) \land (\addr', t, c) = \sem{\pc \stdCapOff z}_{\reg_{j0}})
\end{align*}
By \Cref{lemma:sem-nextreg-invariant}, $\sem{x}_{\reg_{j0}} = \sem{x}_{\reg_{j1}}$.
Furthermore, since $\buf_0(j) = \buf_1(j)$,
then the above statement still holds if we replace $\buf_0$ with $\buf_1$, and $\reg_{j0}$ with $\reg_{j1}$.
Therefore, $\wf{\mem_1,\reg_1, \buf_1, j}$.
\end{itemize}

\pgheading{Case $d = \dcommit$}
Denote $i = \min\dom{\buf}$.
By rules for commit steps, there should be $\buf_1 = \buf_0 \backslash i$.
Furthermore, for all $j \in \dom{\buf_1}$, there should be $\buf_1(j) = \buf_0(j)$.
We aim to prove that $\wf{\mem_1, \reg_1, \buf_1, j}$.
Let
\begin{align*}
  \reg_{j0} & = \apl(\pref j {\buf_0}, \reg_0) & \reg_{j1} = \apl(\pref j {\buf_1}, \reg_1).
\end{align*}
By \Cref{lemma:apl-invariant}, there should be $\reg_{j0} \nextreg \reg_{j1}$.

Since $S_0$ is well-formed, there should be $\wf{\mem_0, \reg_0, \buf_0, j}$.
By \Cref{def:well-formed-buf}, there must exist $\instr\in \Instrs$ such that $\instr = \stdReadMemInst{\mem_0, \sem{\pc}_{\reg_{j0}}}$.

We aim to show that $\instr =\stdReadMemInst{\mem_1, \sem{\pc}_{\reg_{j1}}}$.
Note that, by \Cref{lemma:sem-nextreg-invariant}, $\reg_{j0} \nextreg \reg_{j1}$ implies that $\sem{\pc}_{\reg_{j0}} = \sem{\pc}_{\reg_{j1}}$.
Furthermore, consider the following cases regarding the type of the committed instruction in entry $\buf_0(i)$:
\begin{itemize}
\item $\buf_0(i) = x \gets v \specat{\tg}$ is not a store instruction: by \ref{hw:commit-assign}, there should be $\mem_0 = \mem_1$. Thus,
\begin{align*}
  \instr & = \stdReadMemInst{\mem_0, \sem{\pc}_{\reg_{j0}}} \\
  & =\stdReadMemInst{\mem_1, \sem{\pc}_{\reg_{j1}}}
\end{align*}
\item $\buf_0(i) = \store{(v, t, c)}{\stdCap{\addr}}{\sz}\specat{\varepsilon}$ is a store instruction: by \ref{hw:commit-store}, there should be
\begin{align*}
  & \mem_0 = (\memData, \memTag) \\
  & \mem_1 = (\memData[[\addr, \addr + \sz) \mapsto v], \memTag[\addr / 2 \mapsto c]) \\
\end{align*}
Note that $S_0$ is well-formed, so $\wf{\mem_0, \reg_0, \buf_0, i}$. By \Cref{def:well-formed-buf}, this implies that
\begin{align}
  & \store{\_}{\_}{\sz} = \stdReadMemInst{\mem_0, \sem{\pc}_{\reg_0}} \label{eq:apl-commit-store-1} \\
  & \stdCheckCap{\stdCap{\addr}, \sz, \readwrite}. \label{eq:apl-commit-store-2}
\end{align}
Denote $\sem{\pc}_{\reg_0}=(\stdCap{\addr'}, \_, \_)$.
By \ref{memory:read-mem-inst}, there should be $\stdCheckCap{\stdCap{\addr'}, \szInst, \exec}$. This implies that $[\addr', \addr' + \szInst) \in \permAddr{\reg_0, \mem_0, \exec}$.
Furthermore, the above statements also implies that $[\addr, \addr + \sz) \in \permAddr{\reg_0, \mem_0, \readwrite}$.
By \Cref{def:isa-wf}, the ISA state $\cf{\mem_0, \reg_0}$ is well-formed implies that
\begin{equation*}
  \permAddr{\reg_0, \mem_0, \readwrite} \cap \permAddr{\reg_0, \mem_0, \exec} = \emptyset
\end{equation*}
Thus, $[\addr, \addr+\sz) \cap[\addr', \addr'+\szInst)=\emptyset$.
Therefore,
\begin{align*}
  \instr & = \stdReadMemInst{\mem_0, \sem{\pc}_{\reg_{j0}}} \\
  & =\stdReadMemInst{\mem_1, \sem{\pc}_{\reg_{j1}}}.
\end{align*}
\end{itemize}

To prove $\wf{\mem_1, \reg_1, \buf_1, j}$, we just need to consider the following cases:

\pgheading{Subcase $\instr \not \in \qty{\beqz{x}{\addr''}, \jmp{e}}$}
$\wf{\mem_0, \reg_0, \buf_0, j}$ implies that $\buf_0(j) = \instr\specat{\varepsilon}$.
Note that $\buf_1(j) = \buf_0(j) = \instr\specat{\varepsilon}$, so by \Cref{def:well-formed-buf} there should be $\wf{\mem_1, \reg_1, \buf_1, j}$.

\pgheading{Subcase $\instr = x \gets e$}
$\wf{\mem_0, \reg_0, \buf_0, j}$ implies that $x\neq \pc$ and $\buf_0(j) = x \gets v\specat{\varepsilon}$ where $v = \sem{e}_{\reg_{j0}}\neq \bot$.
By \Cref{lemma:sem-nextreg-invariant}, $\reg_{j0} \nextreg \reg_{j1}$ implies that $\sem{e}_{\reg_{j0}} = \sem{e}_{\reg_{j1}}$.
Thus, the above statement still holds if we replace $\buf_0$ with $\buf_1$ and $\reg_{j0}$ with $\reg_{j1}$.
Therefore, there should be $\wf{\mem_1, \reg_1, \buf_1, j}$.

\pgheading{Subcase $\instr = \load{x}{e}{\sz}$}
$\wf{\mem_0, \reg_0, \buf_0, j}$ implies that $x\neq \pc$ and either of the following cases holds:
\begin{itemize}
\item Load data is forwarded from an in-flight store: in this case there should be
\begin{equation}
\begin{aligned}
  & \buf_0(j) = x\gets (v, t, c)\specat{(\stdCap{\addr}, \sz, j')} \\
  & (\stdCap{\addr}, \_, \tcap) = \sem{e}_{\reg_{j0}} \\
  & \stdCheckCap{\stdCap{\addr}, \sz, \readonly} \\
  & \buf_0(j') = \store{(v', t', c')}{\stdCap{\addr'}}{\sz'}\specat{\varepsilon} \\
  & [\addr, \addr + \sz) = [\addr', \addr' + \sz') \\
  & (v, t, c) = \stdReadRes{v', t_c, c', \sz} \\
  & \forall j' < k < j,\; \buf_0(k) = \store{\_}{\stdCap{\addr'}}{\sz'}\specat{\tg} \Rightarrow \\
  & \qquad \qquad \qquad \qquad \qquad [\addr, \addr + \sz)\cap [\addr', \addr' + \sz') = \emptyset
\end{aligned}
\label{eq:wf-commit-load-hyp-1}
\end{equation}
By \Cref{lemma:sem-nextreg-invariant}, $\sem{e}_{\reg_{j0}} = \sem{e}_{\reg_{j1}}$.
Consider the following two cases
\begin{itemize}
  \item $j' = i$, i.e., the store where the load data was forwarded from is committed at this step:
  To prove $\wf{\mem_1, \reg_1, \buf_1, j}$, we aim to prove the following:
  \begin{equation}
  \begin{aligned}
    & \buf_1(j) = x\gets (v, t, c)\specat{(\stdCap{\addr}, \sz, j')} \\
    & (\stdCap{\addr}, \_, \tcap) = \sem{e}_{\reg_{j1}} \\
    & \stdCheckCap{\stdCap{\addr}, \sz, \readonly} \\
    & j' = \bot \lor j' < \min\dom{\buf_1} \\
    & (v, t, c) = \stdReadRes{\memData^1[\addr, \addr + \sz], t_c, \memTag^1[\addr/ 2], \sz} \\
    & \forall k < j,\; \buf_1(k) = \store{\_}{\stdCap{\addr'}}{\sz'}\specat{\tg} \Rightarrow \\
    & \qquad \qquad \qquad \qquad \qquad [\addr, \addr + \sz)\cap [\addr', \addr' + \sz') = \emptyset\\
    & \mem_1 = (\memData^1, \memTag^1)\\
    & \stdCapMeta{\addr} = (\_, \_, \_, t_c).
  \end{aligned}
  \label{eq:wf-commit-load-goal-1}
  \end{equation}
  By \ref{hw:commit-store}, as $\buf_1 = \buf_0 \backslash i$, then $j' = i < \min\dom{\buf_1}$
  Furthermore, for all $k \in \dom{\buf_1}$, $\buf_0(k) = \buf_1(k)$.
  By \Cref{lemma:sem-nextreg-invariant}, $\reg_{j0} \nextreg \reg_{j1}$ implies that $\sem{e}_{\reg_{j0}} = \sem{e}_{\reg_{j1}}$.
  Furthermore, by \ref{hw:commit-store}, we also have
  \begin{align*}
    \memData^1[\addr, \addr + \sz] & = v' & \memTag^1[\addr / 2] & = c'.
  \end{align*}
  By the above reasoning, \eqref{eq:wf-commit-load-hyp-1} implies \eqref{eq:wf-commit-load-goal-1}.
  Therefore, there should be $\wf{\mem_1, \reg_1, \buf_1, j}$.
  \item $j' \neq i$, i.e., the store where the load data was forwarded from is not committed at this step: In this case, since $i = \min\dom{\buf_0}$, there should be $i < j'$.
  Hence, for all $k$ such that $j' \leq k \leq j$, there should be $\buf_0(k) = \buf_1(k)$.
  By \Cref{lemma:sem-nextreg-invariant}, $\reg_{j0} \nextreg \reg_{j1}$ implies that $\sem{e}_{\reg_{j0}} = \sem{e}_{\reg_{j1}}$.
  Thus, the above statement still holds if we replace $\buf_0$ with $\buf_1$ and $\reg_{j0}$ with $\reg_{j1}$.
  Therefore, there should be $\wf{\mem_1, \reg_1, \buf_1, j}$.
\end{itemize}

\item Load data is loaded from memory or forwarded from a committed store: in this case, there should be
\begin{equation}
\begin{aligned}
  & \buf_0(j) = x\gets (v, t, c)\specat{(\stdCap{\addr}, \sz, j')} \\
  & (\stdCap{\addr}, \_, \tcap) = \sem{e}_{\reg_{j0}} \\
  & \stdCheckCap{\stdCap{\addr}, \sz, \readonly} \\
  & j' = \bot \lor j' < \min\dom{\buf_0} \\
  & (v, t, c) = \stdReadRes{\memData^0[\addr, \addr + \sz], t_c, \memTag^0[\addr/ 2], \sz} \\
  & \forall k < j,\; \buf_0(k) = \store{\_}{\stdCap{\addr'}}{\sz'}\specat{\tg} \Rightarrow \\
  & \qquad \qquad \qquad \qquad \qquad [\addr, \addr + \sz)\cap [\addr', \addr' + \sz') = \emptyset\\
  & \mem_0 = (\memData^0, \memTag^0)\\
  & \stdCapMeta{\addr} = (\_, \_, \_, t_c).
\end{aligned}
\label{eq:wf-commit-load-hyp-2}
\end{equation}
We consider the following cases:
\begin{itemize}
  \item The committed instruction is not a store instruction, i.e., $\buf_0(i) = x \gets v \specat{\tg}$. By \ref{hw:commit-assign}, $\mem_1 = \mem_0$.
  By \Cref{lemma:sem-nextreg-invariant}, $\sem{e}_{\reg_{j0}} = \sem{e}_{\reg_{j1}}$.
  Hence, \eqref{eq:wf-commit-load-hyp-2} still holds if replacing $\buf_0$ with $\buf_1$, $\mem_0$ with $\mem_1$, and $\reg_{j0}$ with $\reg_{j1}$.
  Therefore, there should be $\wf{\mem_1, \reg_1, \buf_1, j}$.
  \item The committed instruction is a store instruction, i.e., $\buf_0(i) = \store{\_}{\_}{\sz}\specat{\varepsilon}$.
  By \ref{hw:commit-assign} and \eqref{eq:wf-commit-load-hyp-2}, there should be
  \begin{align*}
    (v, t, c) & = \stdReadRes{\memData^0[\addr, \addr + \sz], t_c, \memTag^0[\addr/ 2], \sz} \\
    & = \stdReadRes{\memData^1[\addr, \addr + \sz], t_c, \memTag^1[\addr/ 2], \sz}.
  \end{align*}
  Specifically, \eqref{eq:wf-commit-load-hyp-2} implies that the committed store should not affect the load result since there is no aliasing between the committed store at $\buf_0(i)$ and the load at $\buf_0(j)$. Thus, the above statement can be deduced by discussing different cases to derive $\stdReadRes{}$.
  Therefore, there should be $\wf{\mem_1, \reg_1, \buf_1, j}$.
\end{itemize}
\end{itemize}

\pgheading{Subcase $\instr = \store{x}{e}{\sz}$}
$\wf{\mem_0, \reg_0, \buf_0, j}$ implies that
\begin{align*}
  & \buf_0(i) = \store{(v', t, c')}{\stdCap{\addr}}{\sz}\specat{\varepsilon} \\
  & (\stdCap{\addr}, \_, \tcap) = \sem{e}_{\reg_{j0}} \\
  & \stdCheckCap{\stdCap{\addr}, \sz, \readwrite} \\
  & (v, t, c) = \sem{x}_{\reg_{j0}} \\
  & (v', c') = \stdWriteRes{(v, t, c), t_c, \sz} \\
  & \stdCapMeta{\addr} = (\_, \_, \_, t_c).
\end{align*}
By \Cref{lemma:sem-nextreg-invariant}, $\reg_{j0} \nextreg \reg_{j1}$ implies that $\sem{e}_{\reg_{j0}} = \sem{e}_{\reg_{j1}}$ and $\sem{x}_{\reg_{j0}} = \sem{x}_{\reg_{j1}}$.
Thus, the above statement still holds if we replace $\buf_0$ with $\buf_1$ and $\reg_{j0}$ with $\reg_{j1}$.
Therefore, there should be $\wf{\mem_1, \reg_1, \buf_1, j}$.

\pgheading{Subcase $\instr = \jmp{e}$}
$\wf{\mem_0, \reg_0, \buf_0, j}$ implies that $\buf_0(j) = \pc \gets (\addr', \untainted, c)\specat{\tg}$, and $\tg=\varepsilon \Rightarrow \sem{e}_{\reg_{j0}} = (\addr', \_, c)$. Furthermore, there should also be $\addr' = (\_, \untainted, \_)$.
If $\tg=\varepsilon$, then $\sem{e}_{\reg_{j0}} \neq \bot$.
By \Cref{lemma:sem-nextreg-invariant}, $\reg_{j0} \nextreg \reg_{j1}$ implies that $\sem{e}_{\reg_{j0}} = \sem{e}_{\reg_{j1}}$.
Therefore, there should be $\wf{\mem_1, \reg_1, \buf_1, j}$.

\pgheading{Subcase $\instr = \beqz{x}{\addr''}$}
$\wf{\mem_0, \reg_0, \buf_0, j}$ implies that $\buf_0(j) = \pc \gets(\addr', t, c)\specat{\tg}$, and if $\tg=\varepsilon$, there should be
\begin{align*}
  & (v, \_, \_) = \sem{x}_{\reg_{j0}}
  & & z = (v = 0) \;?\; \addr'' : 1
  & & (\addr', t, c) = \sem{\pc + z}_{\reg_{j0}}.
\end{align*}
Suppose $\tg=\varepsilon$. Then $\sem{x}_{\reg_{j0}} \neq \bot$ and $\sem{\pc + z}_{\reg_{j0}} \neq \bot$.
By \Cref{lemma:sem-nextreg-invariant}, $\reg_{j0} \nextreg \reg_{j1}$ implies that
\begin{align*}
  \sem{x}_{\reg_{j0}} & = \sem{x}_{\reg_{j1}} & \sem{\pc + z}_{\reg_{j0}} & = \sem{\pc + z}_{\reg_{j1}}.
\end{align*}
Therefore, there should be $\wf{\mem_1, \reg_1, \buf_1, j}$.

\end{proof}

\begin{theorem}[Functional Correctness]
Let $S_0 = \qty{\mem, \reg}$ be the initial ISA state and $C_0 =\qty{\mem, \reg, \bufinit, \muinit}$ be the initial HW configuration.
Assume $S_0$ is a well-formed ISA state, and $C_0$ is a well-formed HW configuration.
Let
\begin{align*}
  \traceIsa(S_0) & = S_0 \step{o_1} S_1 \step{o_2} S_2 \dots \\
  \traceUarch(C_0) & = C_0 \sstep{}{} C_1 \sstep{}{} C_2 \dots
\end{align*}
be the ISA and HW execution traces.
The HW semantics is functionally correct, i.e.,
for all $n\in \mathbb{N}$, $\mem_n$, and $\reg_n$,
\begin{equation*}
  C_n = \cf{\mem_n, \reg_n, \_, \_} \Rightarrow
  S_{f(\qty{C_0, C_1, \dots, C_{n-1}})} = \cf{\mem_n, \reg_n},
\end{equation*}
where
\begin{multline*}
  f(\qty{C_0, \dots, C_{n-1}}) \\
  = \abs{\qty{\cf{\_, \_, \buf, \mu}\in \qty{C_0, \dots, C_{n-1}}: \stdNextMu{\stdUpdateMu{\mu, \lowproj{\buf}}} = \dcommit}}
\end{multline*}
counts the number of previous configurations in set $\qty{C_0, \dots, C_{n-1}}$ that perform a commit operation.
\label{lemma:functional-correctness}
\end{theorem}
\begin{proof}
We prove by induction on $n$.

For the base case when $n = 0$, the statement automatically holds since
\begin{align*}
  & f(\qty{C_0,\dots, C_{n-1}}) = f(\emptyset) = 0\\
  & C_0 = \qty{\mem, \reg, \bufinit, \muinit} \quad S_0 = \qty{\mem, \reg}.
\end{align*}

Suppose the statement holds for all $i \leq n$, i.e.,
\begin{equation*}
  \begin{aligned}
    & C_i = \cf{\mem_i, \reg_i, \buf_i, \mu_i} \\
    & S_{f(\qty{C_0, \dots, C_{i-1}})} = \cf{\mem_i, \reg_i}.
  \end{aligned}
\end{equation*}
Specifically, we denote $k = f(\qty{C_0, \dots, C_{n-1}})$ for convenience.

We aim to prove that the statement still holds for $n+1$, i.e.,
\begin{equation*}
  \begin{aligned}
    & C_{n+1} = \cf{\mem_{n+1}, \reg_{n+1}, \buf_{n+1}, \mu_{n+1}} \\
    & S_{k'} = \cf{\mem_{n+1}, \reg_{n+1}} \quad
    k' = f(\qty{C_0, \dots, C_{n}}).
  \end{aligned}
  \label{eq:func-corr-goal}
\end{equation*}

By applying \Cref{lemma:isa-wf} repeatedly, we can prove that $S_i$ is well-formed (i.e., $\wf{S_i}$) for $0\leq i \leq k$.
Then, by applying \Cref{lemma:hw-wf-invariant} repeatedly, we can know that $C_i$, $i \leq n$ is also well-formed.

By \ref{hw:step}, there should be
\begin{mathpar}
\inferrule{
  \mu' = \stdUpdateMu{\mu_n, \lowproj{\buf_n}} \\
  d = \stdNextMu{\mu'} \\
  \cf{\mem_n, \reg_n, \buf_n, \mu'} \sstep{d}{} \cf{\mem_{n+1}, \reg_{n+1}, \buf_{n+1}, \mu_{n+1}}
}{
  \cf{\mem_n, \reg_n, \buf_n, \mu_n} \sstep{}{} \cf{\mem_{n+1}, \reg_{n+1}, \buf_{n+1}, \mu_{n+1}}
}
\end{mathpar}
Consider the following cases for $d$.

\pgheading{Case $d = \dfetch$ or $d = \dexecute{i}$}
In this case,
\begin{equation*}
  k' = f(\qty{C_0, \dots, C_{n}}) = f(\qty{C_0, \dots, C_{n-1}}) = k,
\end{equation*}
so $S_{k'} = S_{k} = \cf{\mem_n, \reg_n}$.
Furthermore, by rules for fetch and execute steps, $\mem_{n+1}=\mem_{n}$ and $\reg_{n+1} = \reg_{n}$.
Thus, \eqref{eq:func-corr-goal} holds.

\pgheading{Case $d = \dexecute{i}$}
In this case,
\begin{equation*}
  k' = f(\qty{C_0, \dots, C_{n}}) = f(\qty{C_0, \dots, C_{n-1}}) = k.
\end{equation*}

\pgheading{Case $d = \dcommit$}
In this case,
\begin{equation*}
  k' = f(\qty{C_0, \dots, C_{n}}) = f(\qty{C_0, \dots, C_{n-1}}) = k.
\end{equation*}
So we need to prove that there exists $S_{k+1}$ such that
\begin{itemize}
  \item $S_k \step{o_{k+1}} S_{k+1}$ and
  \item $S_{k+1} = \cf{\reg_{n+1}, \mem_{n+1}}$.
\end{itemize}

By rules for commit steps, $\dom{\buf_n} \neq \emptyset$.
Let $i = \min\dom{\buf_n}$. Then, there should be $\reg_n = \apl(\pref i {\buf_n}, \reg_n)$.
Since $C_n$ is well-formed, by \Cref{def:well-formed-buf}, there must exist $\instr \in \Instrs$ such that
\begin{equation*}
  \instr = \stdReadMemInst{\mem_n, \sem{\pc}_{\reg_n}}.
\end{equation*}
Furthermore, by rules for commit steps, the following statements hold:
\begin{itemize}
  \item For $\instr \not \in \qty{\beqz{x}{\addr''}, \jmp{e}}$, they must be executed, so $\buf_n(i) = \instr'\specat{\tg}$ and $\instr' \neq \instr$, i.e., the first case in \Cref{def:well-formed-buf} does not hold when the HW configuration performs a commit step.
  \item For $\instr\in \qty{\beqz{x}{\addr''}, \jmp{e}}$, there must be $\buf_n(i) = \pc \gets (\addr', t, c)\specat{\varepsilon}$, i.e., the branch is resolved.
\end{itemize}
Thus, we only need to consider the following cases:

\pgheading{Subcase $\instr = x \gets e$}
Since $C_n$ is well-formed, by \Cref{def:well-formed-buf}, there must be $x\neq \pc$ and $\buf_n(i) = x \gets \sem{e}_{\reg_n}\specat{\varepsilon}$.
Then, by \ref{hw:commit-assign}, we have
\begin{align*}
  & \mem_{n+1} = \mem_n \\
  & \reg_{n+1} = \reg_n[\pc \mapsto \sem{\stdNextPc{\pc}}_{\reg_n}][x \mapsto \sem{e}_{\reg_n}].
\end{align*}
Let
\begin{equation*}
  S_{k+1} = \cf{\mem_n, \reg_n[\pc \mapsto \sem{\stdNextPc{\pc}}_{\reg_n}][x \mapsto \sem{e}_{\reg_n}]}.
\end{equation*}
By \ref{isa:isa-assign}, there should be $S_k \step{()} S_{k+1}$.
Since $x\neq \pc$, there should also be $S_{k+1} = \cf{\mem_{n+1}, \reg_{n+1}}$.

\pgheading{Subcase $\instr = \load{x}{e}{\sz}$}
Since $C_n$ is well-formed, by \Cref{def:well-formed-buf}, there must be
\begin{align*}
  & \buf_n(i) = x \gets (v, t, c) \specat{(\stdCap{\addr}, \sz, j)}\\
  & (\stdCap{\addr}, \_, \tcap) = \sem{e}_{\reg_n} \\
  & \stdCheckCap{\stdCap{\addr}, \sz, \readonly}.
\end{align*}
\Cref{def:well-formed-buf} also implies that $j < i$.
Recall that by \ref{hw:commit-assign}, $i = \min\dom{\buf_n}$, so there should be
\begin{align*}
  & (v, t, c) = \stdReadRes{\memData[\addr, \addr + \sz], t_c, \memTag[\addr/ 2], \sz} \\
  & \mem = (\memData, \memTag)\\
  & \stdCapMeta{\addr} = (\_, \_, \_, t_c).
\end{align*}
By \ref{hw:commit-assign}, there should be
\begin{align*}
  & \mem_{n+1} = \mem_n \\
  & \reg_{n+1} = \reg_n[\pc \mapsto \sem{\stdNextPc{\pc}}_{\reg_n}][x \mapsto (v, t, c)].
\end{align*}
By \ref{memory:read-mem}, the above statements also imply that $(v, t, c) = \stdReadMem{\mem, (\stdCap{\addr}, \_, \tcap),\sz}$.
Let
\begin{align*}
  & S_{k + 1} = \cf{\mem_n, \reg_n[\pc \mapsto \sem{\stdNextPc{\pc}}_{\reg_n}][x \mapsto (v, t, c)]}.
\end{align*}
Since $x\neq \pc$, there should be $S_k = \cf{\mem_{n+1}, \reg_{n+1}}$.
Furthermore, by \ref{isa:isa-load}, $S_k \step{o_{k+1}} S_{k+1}$.

\pgheading{Subcase $\instr = \store{x}{e}{\sz}$}
Since $C_n$ is well-formed, by \Cref{def:well-formed-buf}, there must be
\begin{align*}
  & \buf_n(i) = \store{(v', t, c')}{\stdCap{\addr}}{\sz}\specat{\varepsilon} \\
  & (\stdCap{\addr}, \_, \tcap) = \sem{e}_{\reg_n} \\
  & \stdCheckCap{\stdCap{\addr}, \sz, \readwrite} \\
  & (v, t, c) = \sem{x}_{\reg_n} \\
  & (v', c') = \stdWriteRes{(v, t, c), t_c, \sz} \\
  & \stdCapMeta{\addr} = (\_, \_, \_, t_c).
\end{align*}
Denote $\mem_n = (\memData, \memTag)$.
By \ref{hw:commit-store}, there should be
\begin{align*}
  & \mem_{n+1} = (\memData[[\addr, \addr + \sz) \mapsto v'], \memTag[[\addr/\szCap, \lceil(\addr + \sz)/\szCap\rceil ) \mapsto c']) \\
  & \reg_{n+1} = \reg_n[\pc \mapsto \sem{\stdNextPc{\pc}}_{\reg_n}]
\end{align*}
Let
\begin{align*}
  & S_{k+1} = (\stdWriteMem{\mem_n, \sem{e}_{\reg_n}, \sz, \sem{x}_{\reg_n}}, \reg_n[\pc \mapsto \sem{\stdNextPc{\pc}}_{\reg_n}])
\end{align*}
By \ref{memory:write-mem}, there should be
$S_{k+1} = \cf{\mem_{n+1}, \reg_{n+1}}$, and by \ref{isa:isa-store}, there should be $S_k \step{o_{k+1}} S_{k+1}$.

\pgheading{Subcase $\instr = \jmp{e}$}
According to our previous discussion, there should be $\buf_n(i) = \pc \gets (\addr', t, c) \specat{\varepsilon}$.
Since $C_n$ is well-formed, by \Cref{def:well-formed-buf}, there must be $t = \untainted$ and $\sem{e}_{\reg_n} = (\addr', \_, c)$.

By \ref{hw:commit-assign}, there should be
\begin{align*}
  & \mem_{n+1} = \mem_n \\
  & \reg_{n+1} = \reg_n[\pc \mapsto \sem{\stdNextPc{\pc}}_{\reg_n}][\pc \mapsto (\addr', \untainted, c)] = \reg_n[\pc \mapsto (\addr', \untainted, c)].
\end{align*}
Note that the last step holds because the second update on $\pc$ overwrites the first one.
Let
\begin{equation*}
  S_{k + 1} = \cf{\mem_n, \reg_n[\pc \mapsto (\addr', \untainted, c)]}.
\end{equation*}
Note that $S_{k+1} = \cf{\mem_{n+1}, \reg_{n+1}}$.
By \ref{isa:isa-jmp}, there should be $S_k \step{o_{k+1}} S_{k+1}$.

\pgheading{Subcase $\instr = \beqz{x}{\addr'}$}
According to our previous discussion, there should be $\buf_n(i) = \pc \gets (\addr', t, c) \specat{\varepsilon}$.
Since $C_n$ is well-formed, there must be
\begin{align*}
  & (v, \_, \_) = \sem{x}_{\reg_n}
  & & z = (v = 0) \;?\; \addr'' : 1
  & & (\addr', t, c) = \sem{\pc + z}_{\reg_n}.
\end{align*}

By \ref{hw:commit-assign}, there should be
\begin{align*}
  & \mem_{n+1} = \mem_n \\
  & \reg_{n+1} = \reg_n[\pc \mapsto \sem{\stdNextPc{\pc}}_{\reg_n}][\pc \mapsto (\addr', t, c)] = \reg_n[\pc \mapsto (\addr', t, c)].
\end{align*}
Note that the last step holds because the second update on $\pc$ overwrites the first one.
Let
\begin{equation*}
  S_{k + 1} = \cf{\mem_n, \reg_n[\pc \mapsto (\addr', t, c)]}.
\end{equation*}
Note that $S_{k+1} = \cf{\mem_{n+1}, \reg_{n+1}}$.
By \ref{isa:isa-beqz}, there should be $S_k \step{o_{k+1}} S_{k+1}$.
\end{proof}

\subsection{Relative Speculative Constant Time}
\label{sec:sec-proof}

\begin{definition}[Untainted Projection of Reorder Buffer]
  Operands of in-flight instructions in reorder buffer $\buf$ have three formats, and their untainted projection is defined as
\begin{equation*}
\begin{array}{rcl}
  \lowproj{e} & \defsym & e \\
  \lowproj{(v, t, c)} & \defsym &
  \begin{cases}
    (v, \untainted, c) & t  = \untainted \\
    \bot & t =  \tainted \\
  \end{cases} \\
  \lowproj{\bot} & \defsym & \bot \\
\end{array}
\end{equation*}
Untainted projection of instructions in $\buf$ is defined as
\begin{equation*}
\begin{array}{rcl}
  \lowproj{(\pc \gets E)} & \defsym & \pc \gets E\\
  \lowproj{(x \gets E)} & \defsym & x \gets \lowproj{E} \quad x \neq \pc \\
  \lowproj{(\load{X}{E}{\sz})} & \defsym & \load{X}{E}{\sz}\\
  \lowproj{(\store{X}{E}{\sz})} & \defsym & \store{\lowproj{X}}{E}{\sz}.
\end{array}
\end{equation*}
Untainted projection of $\buf$ satisfies the following properties:
\begin{itemize}
  \item $\dom{\lowproj{\buf}} = \dom{\buf}$
  \item $\forall i \in \dom{\buf},\;
  \buf(i) = \instr\specat{\tg} \Rightarrow \lowproj{\buf}(i) = \lowproj{\instr}\specat{\tg}$.
\end{itemize}
\label{def:low-proj-buf}
\end{definition}

\begin{definition}[Public Equivalence for HW Configurations]
HW configurations $C=(\mem, \reg, \buf, \mu)$, $C'=(\mem', \reg', \buf', \mu')$ are \emph{publicly equivalent}, i.e., $C\pubeq C'$ if
\begin{itemize}
  \item $\cf{\mem, \reg} \pubeq \cf{\mem', \reg'}$
  \item $\lowproj{\buf} = \lowproj{\buf'}$
  \item $\mu = \mu'$
\end{itemize}
\label{def:pub-eq}
\end{definition}

\begin{lemma}[$\apl$ Capability Monotonicity]
Let $C_0$ be the initial HW configuration and $C_0 \sstep{}{} C_1 \sstep{}{} \dots \sstep{}{} C_n$ be the execution trace.
Denote $C_n = \cf{\mem, \reg, \buf, \mu}$.
For all $i$, denote $\reg_i = \apl(\pref i \buf, \reg)$.
Then,
\begin{align*}
  \forall x,\; \reg_i(x) = ((\addr, (\perm, \rbase, \rend, \taint)), \_, \tcap) \Rightarrow &
  (\addr, (\perm, \rbase, \rend, \taint)) \in \rc{\reg, \mem} \land \\
  & [\rbase, \rend) \subseteq \taintAddr{\reg, \mem, \taint} \land \\
  & [\rbase, \rend) \subseteq \permAddr{\reg, \mem, \perm}.
\end{align*}
\label{lemma:reg-cap-mono}
\end{lemma}
\begin{proof}
  The proof goes by induction on $n$ and applying \Cref{lemma:apl-invariant}.
\end{proof}

\begin{lemma}[Expr Capability Monotonicity]
Let $C_0$ be the initial HW configuration and $C_0 \sstep{}{} C_1 \sstep{}{} \dots \sstep{}{} C_n$ be the execution trace.
Denote $C_n = \cf{\mem, \reg, \buf, \mu}$.
For all $i$, denote $\reg_i = \apl(\pref i \buf, \reg)$.
Then,
\begin{align*}
  \forall e,\; \sem{e}_{\reg_i} = ((\addr, (\perm, \rbase, \rend, \taint)), \_, \tcap) \Rightarrow &
  (\addr, (\perm, \rbase, \rend, \taint)) \in \rc{\reg, \mem} \land \\
  & [\rbase, \rend) \subseteq \taintAddr{\reg, \mem, \taint} \land \\
  & [\rbase, \rend) \subseteq \permAddr{\reg, \mem, \perm}.
\end{align*}
\label{lemma:exp-cap-mono}
\end{lemma}
\begin{proof}
  The proof goes by applying \Cref{lemma:reg-cap-mono} and \Cref{lemma:exprmono}.
\end{proof}

\begin{lemma}[$\apl$ Public Equivalence]
Let $C = \cf{\mem, \reg, \buf, \mu}$ and $C'=\cf{\mem', \reg', \buf', \mu'}$ be two HW configurations such that $C \pubeq C'$.
For all $i$, denote $\reg_i = \apl(\pref i \buf, \reg)$, $\reg_i' = \apl(\pref i {\buf'}, \reg')$.
Then,
\begin{equation*}
  \forall x,\; \exists v\neq \bot,\; \reg_i(x) = v \Rightarrow \exists v'\neq\bot,\; \reg_i'(x) = v' \land \lowproj{v} = \lowproj{v'}.
\end{equation*}
\label{lemma:reg-pub-eq}
\end{lemma}
\begin{proof}
  The proof goes by induction on $i$ and \Cref{def:pub-eq}.
\end{proof}

\begin{lemma}[Expr Public Equivalence]
Let $C = \cf{\mem, \reg, \buf, \mu}$ and $C'=\cf{\mem', \reg', \buf', \mu'}$ be two HW configurations such that $C \pubeq C'$.
For all $i$, denote $\reg_i = \apl(\pref i \buf, \reg)$, $\reg_i' = \apl(\pref i {\buf'}, \reg')$.
Then,
\begin{equation*}
  \forall e,\; \exists v\neq \bot,\; \sem{e}_{\reg_i} = v \Rightarrow \exists v'\neq\bot,\; \sem{e}_{\reg_i'} = v' \land \lowproj{v} = \lowproj{v'}.
\end{equation*}
\label{lemma:exp-pub-eq}
\end{lemma}
\begin{proof}
  The proof goes by applying \Cref{lemma:reg-pub-eq}.
\end{proof}

\begin{lemma}[Auxiliary lemma for $\aplsan$ equivalence]
  \label{lemma:aplsaneqaux}
  Let $\buf$ and $\buf'$ be buffers such that
  $\lowproj \buf = \lowproj{\buf'}$, and $\reg, \reg'$ be two partial register files such that $\lowproj \reg = \lowproj{\reg'}$. We have
  \[
    \lowproj{\apl(\buf, \reg)} = \lowproj {\apl(\buf', \reg')}.
  \]
\end{lemma}
\begin{proof}
  The proof is by induction on the size of the domain of $\buf$. If it
  is empty, by $\lowproj \buf = \lowproj{\buf'}$, the domain of
  $\buf'$ must also be empty.
  In the inductive case, by the inductive hypothesis, it suffices to observe that:
    \begin{multline*}
      \lowproj{\aplinst(\buf(\min \dom{\buf}), r)} =\\
      \lowproj{\aplinst(\buf'(\min \dom{\buf'}), r')}.
  \end{multline*}

  The conclusion follows by case analysis on $\min \dom{\buf}$.
\end{proof}
\begin{lemma}[Expr $\aplsan$ Equivalence]
Let $C = \cf{\mem, \reg, \buf, \mu}$ and $C'=\cf{\mem', \reg', \buf', \mu'}$ be two HW configurations such that $C \pubeq C'$.
For all $i$, denote $\reg_i = \aplsan(\pref i \buf, \reg)$, $\reg_i' = \aplsan(\pref i {\buf'}, \reg')$.
Then,
\begin{equation*}
  \forall e,\; \sem{e}_{\reg_i} = (\_, \untainted, \_) \Rightarrow \sem{e}_{\reg_i} = \sem{e}_{\reg_i'}.
\end{equation*}
\label{lemma:expr-aplsan-eq}
\end{lemma}
\begin{proof}
  Observe that by $C \pubeq C'$, we deduce that
  $\dom {\pref i {\buf'}} = \dom {\pref i {\buf}}$.
  The proof goes by cases on the size of those domains. If
  such size is 0, both buffers are empty, $\reg_i=\reg$, and
  $\reg_i' = \reg'$.
  Under these assumptions, we deduce that
  $\lowproj{\sem{e}_{\reg_i}} = \lowproj{\sem{e}_{\reg_i'}}$ from
  \Cref{lemma:exp-pub-eq}. As we have $\sem{e}_{\reg_i}\neq \bot$ and
  $\lowproj{\sem{e}_{\reg_i'}} = \sem{e}_{\reg_i'}$.

  If the domain of the buffer is not empty, we are required to establish
  \[
    \sem{e}_{\aplsan(\pref i {\buf}, \reg)} = \sem{e}_{\aplsan(\pref i {\buf'}, \reg')}.
  \]
  by definition of $\aplsan$, this reduces to proving
  \[
    \sem{e}_{\lowproj{\apl(\pref i {\buf}, \reg)}} = \sem{e}_{\lowproj{\apl(\pref i {\buf'}, \reg')}}.
  \]
  By assumption, we know that
  $\sem{e}_{\lowproj{\apl(\pref i {\buf}, \reg)}}$ is an untainted
  value, and therefore it cannot be $\bot$. Since by
  \Cref{lemma:aplsaneqaux}, we have
  \[
    \lowproj{\apl(\pref i {\buf}, \reg)} = \lowproj{\apl(\pref i
      {\buf}, \reg)},
  \]
  the conclusion is trivial.
\end{proof}

\again{thm:sct}
\begin{proof}
Let
\begin{align*}
  & S_0 \step{o_1} S_1 \step{o_2} S_2 \dots
  && S_0' \step{o_1'} S_1' \step{o_2'} S_2' \dots
\end{align*}
be the corresponding ISA execution traces.
By assumption, there should be $o_1o_2\dots = o_1'o_2'\dots$.
It suffices to prove that for all $n\in \mathbb{N}$, if one HW configuration successfully executes for $n$ steps, i.e., there exist $C_1, \dots C_n$ such that
\begin{equation*}
  C_0 \sstep{}{} C_1 \sstep{}{} C_2 \dots \sstep{}{}C_n
\end{equation*}
then the other one can also execute at least $n$ steps, i.e., there exist $C_1', \dots, C_n'$ such that
\begin{equation*}
  C_0' \sstep{}{} C_1' \sstep{}{} C_2' \dots \sstep{}{}C_n'
\end{equation*}
and
\begin{equation*}
  \forall k\leq n, C_k\pubeq C_k'.
\end{equation*}

We prove by induction on $n$.
For the base case where $n=0$, since $S_0\pubeq S_0'$ and $\dom{\bufinit}=\emptyset$, then by \Cref{def:pub-eq} and the definition of $C_0$ and $C_0'$, there should be $C_0 \pubeq C_0'$.

Assume the statement holds for $n$. We prove it also holds for $n+1$.
Suppose there exists $C_{n+1}$ such that $C_n\sstep{}{}C_{n+1}$.
We just need to prove that there exists $C_{n+1}'$ such that $C_n'\sstep{}{}C_{n+1}'$ and $C_{n+1}\pubeq C_{n+1}'$.

By \Cref{def:well-formed-buf}, since $C_0$ and $C_0'$ have empty reorder buffers, i.e., $\dom{\bufinit} = \emptyset$, then both $C_0$ and $C_0'$ are well-formed.
Then, by \Cref{lemma:hw-wf-invariant}, for all $k \leq n$, $C_k$ and $C_k'$ are well-formed.

Denote $C_n = \cf{\mem_n, \reg_n, \buf_n, \mu_n}$.
By \ref{hw:step}, there should be
\begin{mathpar}
\inferrule{
  \mu = \stdUpdateMu{\mu_n, \lowproj{\buf_n}}\\
  d = \stdNextMu{\mu}\\
  \cf{\mem_n, \reg_n, \buf_n, \mu} \sstep{d}{}
  C_{n+1}
}{
  \cf{\mem_n, \reg_n, \buf_n, \mu_n} \sstep{}{}
  C_{n+1}
}
\end{mathpar}

Denote $C_n' = \cf{\mem_n', \reg_n', \buf_n', \mu_n'}$. Since $C_n \pubeq C_n'$, then $\lowproj{\buf_n} = \lowproj{\buf_n'}$ and $\mu_n = \mu_n'$.
Denote $\mu' = \stdUpdateMu{\mu_n', \lowproj{\buf_n'}}$.
Thus, there should be $\mu = \mu'$ and $d=\stdNextMu{\mu} = \stdNextMu{\mu'}$.

It suffices to prove that there exists $C_{n+1}'$
such that
\begin{equation}
  \cf{\mem_n', \reg_n', \buf_n', \stdUpdateMu{\mu_n', \lowproj{\buf_n'}}} \sstep{d}{}
  C_{n+1}'
  \label{eq:next-state-goal}
\end{equation}
and
\begin{equation}
  C_{n+1} \pubeq C_{n+1}'
  \label{eq:pub-eq-goal}
\end{equation}

Consider the following cases for $d$.

\pgheading{Case $d = \dfetch$}
By \ref{hw:fetch-other} and \ref{hw:fetch-branch-predict-pc}, there exists $\instr$ such that
\begin{equation*}
  \instr = \stdReadMemInst{\mem_n, \sem{\pc}_{\aplsan(\buf_n, \reg_n)}}.
\end{equation*}
By \ref{memory:read-mem-inst}, $\sem{\pc}_{\aplsan(\buf_n, \reg_n)} = ((\addr, (\perm, \rbase, \rend, \untainted)), \untainted, \tcap)$.
By \Cref{lemma:expr-aplsan-eq}, we have $\sem{\pc}_{\aplsan(\buf_n, \reg_n)} = \sem{\pc}_{\aplsan(\buf_n', \reg_n')}$.
Furthermore, by \Cref{lemma:exp-cap-mono}, there should be $\sem{\pc}_{\aplsan(\buf_n, \reg_n)} \in \rc{\reg_n, \mem_n}$ and $\sem{\pc}_{\aplsan(\buf_n', \reg_n')} \in \rc{\reg_n', \mem_n'}$.
Then, by \Cref{lemma:read-mem-inst}, $\instr = \stdReadMemInst{\mem_n', \sem{\pc}_{\aplsan(\buf_n', \reg_n')}}$.
This indicates that $S_n'$ also satisfies the condition of $\sstep{d}{}$.

Let $i = \sup \dom{\buf_n} + 1 = \sup\dom{\buf_n'} + 1$ (implied by $C_n \pubeq C_n'$).
Consider the following cases:
\begin{itemize}
\item $\instr \not \in \qty{\beqz x \addr'', \jmp e}$:
Let
\begin{align*}
  C_{n+1} & = (\mem_n, \reg_n, \buf_n[i \mapsto \instr \specat{\varepsilon}], \stdUpdateMu{\mu, \sem{\pc}_{\aplsan(\buf_n, \reg_n)}}) \\
  C_{n+1}' & = (\mem_n', \reg_n', \buf_n'[i \mapsto \instr \specat{\varepsilon}], \stdUpdateMu{\mu', \sem{\pc}_{\aplsan(\buf_n', \reg_n')}}).
\end{align*}
By \ref{hw:fetch-other}, there should be $C_n' \sstep{}{} C_{n+1}'$ and $C_{n+1}\pubeq C_{n+1}'$.
\item $\instr \in \qty{\beqz x \addr'', \jmp e}$:
Denote
\begin{equation*}
  (\stdCap{\addr}, t, c) = \sem{\pc}_{\aplsan(\buf_n, \reg_n)} = \sem{\pc}_{\aplsan(\buf_n', \reg_n')}.
\end{equation*}
Let
\begin{align*}
  \buf_{n+1} & = \buf_n[i \mapsto \pc \gets ((\stdPredPc{\mu}, \stdCapMeta{\addr}), t, c)\specat{\stdCap{\addr}}] \\
  C_{n+1} & = (\mem_n, \reg_n, \buf_{n+1}, \stdUpdateMu{\mu, \sem{\pc}_{\aplsan(\buf_n, \reg_n)}}) \\
  \buf_{n+1}' & = \buf_n'[i \mapsto \pc \gets ((\stdPredPc{\mu'}, \stdCapMeta{\addr}), t, c)\specat{\stdCap{\addr}}] \\
  C_{n+1}' & = (\mem_n', \reg_n', \buf_{n+1}', \stdUpdateMu{\mu', \sem{\pc}_{\aplsan(\buf_n', \reg_n')}}).
\end{align*}
Note that $\mu = \mu'$.
By \ref{hw:fetch-branch-predict-pc}, there should be $C_n' \sstep{}{} C_{n+1}'$ and $C_{n+1}\pubeq C_{n+1}'$.
\end{itemize}

\pgheading{Case $d = \dexecute{i}$}
Consider the following cases:

\pgheading{Subcase $\buf_n(i) = x \gets e \specat{\varepsilon}$, $x\neq \pc$}
Since $\lowproj{\buf_n} = \lowproj{\buf_n'}$, there should be
$\buf_n'(i) = x \gets e \specat{\varepsilon} = \buf_n(i)$.
By \ref{hw:execute-assign}, there should be $\sem{e}_{\apl (\pref i {\buf_n}, \reg_n)} = (v, t, c)\neq \bot$.
By \Cref{lemma:exp-pub-eq}, 
this implies that there must exist $(v', t', c')$ such that $\sem{e}_{\apl (\pref i {\buf_n'}, \reg_n')} = (v', t', c') \neq \bot$, and $\lowproj{(v, t, c)} = \lowproj{(v', t', c')}$.
Let
\begin{align*}
  C_{n+1} & = \cf{\mem_n, \reg_n, \buf_n[i \mapsto x \gets (v, t, c)\specat{\varepsilon}], \mu} \\
  C_{n+1}' & = \cf{\mem_n', \reg_n', \buf_n'[i \mapsto x \gets (v', t', c')\specat{\varepsilon}], \mu'}.
\end{align*}
By \ref{hw:execute-assign}, there should be $C_n' \sstep{}{} C_{n+1}'$.
By \Cref{def:low-proj-buf}, $C_{n+1} \pubeq C_{n+1}'$.

\pgheading{Subcase $\buf_n(i) = \pc \gets \addr' \specat{\stdCap{\addr}}$}
By introspection of the rules for assignments, one among rules \ref{hw:execute-jmp-ok}, \ref{hw:execute-jmp-hazard}, \ref{hw:execute-beqz-ok} and \ref{hw:execute-beqz-hazard} was applied,
There should be $\stdReadMemInst{\mem_n, \sem{\pc}_{\aplsan(\pref i {\buf_n}, \reg_n)}} = \jmp{e}$ or $\beqz{x}{\addr}$. The two cases can be proved similarly, so we only prove the case for $\jmp{e}$ here and omit the proof for the other.

By \Cref{def:low-proj-buf,def:pub-eq}, $C_n \pubeq C_n'$ implies that $\buf_n(i) = \buf_n'(i) = \pc \gets \addr' \specat{\stdCap{\addr}}$.

By \ref{memory:read-mem-inst}, $\sem{\pc}_{\aplsan(\pref i {\buf_n}, \reg_n)} = ((\addr, (\perm, \rbase, \rend, \untainted)), \untainted, \tcap)$.
By \Cref{lemma:expr-aplsan-eq}, we have $\sem{\pc}_{\aplsan(\pref i {\buf_n}, \reg_n)} = \sem{\pc}_{\aplsan(\pref i {\buf_n'}, \reg_n')}$.
Furthermore, by \Cref{lemma:exp-cap-mono}, there should be $\sem{\pc}_{\aplsan(\pref i {\buf_n}, \reg_n)} \in \rc{\reg_n, \mem_n}$ and $\sem{\pc}_{\aplsan(\pref i {\buf_n'}, \reg_n')} \in \rc{\reg_n', \mem_n'}$.
Then, by \Cref{lemma:read-mem-inst},
\begin{align*}
  \jmp{e} & = \stdReadMemInst{\mem_n, \sem{\pc}_{\aplsan(\pref i {\buf_n}, \reg_n)}} \\
  & = \stdReadMemInst{\mem_n', \sem{\pc}_{\aplsan(\pref i {\buf_n'}, \reg_n')}}
\end{align*}

Next, denote
\begin{align*}
  & (\addr_0, t, c) = \sem{e}_{\aplsan(\pref i {\buf_n}, \reg_n)} &
  & (\addr_0', t', c') = \sem{e}_{\aplsan(\pref i {\buf_n'}, \reg_n')}.
\end{align*}
Both \ref{hw:execute-jmp-ok} and \ref{hw:execute-jmp-hazard} require the resolved jump target to be
a capability, so $c = \tcap$.
We aim to prove that $\addr_0 = \addr_0'$ and $c' = \tcap$ by discussing the taint status of the result.
\begin{itemize}
  \item $t = \untainted$: By \Cref{lemma:expr-aplsan-eq}, there should be $\sem{e}_{\aplsan(\pref i {\buf_n}, \reg_n)} = \sem{e}_{\aplsan(\pref i {\buf_n'}, \reg_n')}$, hence $\addr_0 = \addr_0'$ and $c' = c = \tcap$.
  \item $t = \tainted$: By the definition of $\aplsan$, this is only possible if $\dom{\pref i {\buf_n}} = \emptyset$. Then, there should be $\aplsan(\pref i {\buf_n}, \reg_n) = \reg_n$.
  Since $C_n \pubeq C_n'$, then there is also $\dom{\pref i {\buf_n'}} = \emptyset$, thereby $\aplsan(\pref i {\buf_n'}, \reg_n') = \reg_n'$.
  By \ref{isa:isa-jmp}, there should exist $\cf{\mem_{n+1}, \reg_{n+1}}$ such that
  \begin{equation*}
    \cf{\mem_n, \reg_n} \step{\stdBr{\addr_0}} \cf{\mem_{n+1}, \reg_{n+1}}.
  \end{equation*}
  Since $S_0 \equiv_{\mathit{ISA}}S_0'$, then there should also exist $\cf{\mem_{n+1}', \reg_{n+1}'}$ such that
  \begin{align*}
    & \cf{\mem_n', \reg_n'} \step{\stdBr{\addr_0'}} \cf{\mem_{n+1}', \reg_{n+1}'} \\
    & \stdBr{\addr_0} = \stdBr{\addr_0'}.
  \end{align*}
  Therefore, $\addr_0 = \addr_0'$.
  Moreover, \ref{isa:isa-jmp} also requires the jump target to be a capability, so the ISA step of
  $\cf{\mem_n', \reg_n'}$ gives $c' = \tcap$.
\end{itemize}

Since $\buf_n(i) = \buf_n'(i)$, the two executions agree on $\addr'$, and we have just
established that $\addr_0 = \addr_0'$ and $c = c' = \tcap$.
Hence the premises of \ref{hw:execute-jmp-ok} hold for $C_n$ if and only if they hold for $C_n'$,
and likewise for \ref{hw:execute-jmp-hazard}, so the same rule applies to $C_n$ and $C_n'$.
Consider the following cases:
\begin{itemize}
\item \ref{hw:execute-jmp-ok} applies, i.e., $\addr' = (\addr_0, \_, \tcap)$:
Let
\begin{align*}
  C_{n+1} & = \cf{\mem_n, \reg_n, \buf_n[i \mapsto \pc \gets \addr' \specat{\varepsilon}], \stdUpdateMu{\mu, \addr_0}} \\
  C_{n+1}' & = \cf{\mem_n', \reg_n', \buf_n'[i \mapsto \pc \gets \addr' \specat{\varepsilon}], \stdUpdateMu{\mu', \addr_0'}}.
\end{align*}
By \ref{hw:execute-jmp-ok}, $C_n' \sstep{}{} C_{n+1}'$.
By \Cref{def:pub-eq}, $C_{n+1} \pubeq C_{n+1}'$.
\item \ref{hw:execute-jmp-hazard} applies, i.e., $\addr' = (\addr_1, \_, \_)$ with $\addr_1 \neq \addr_0$:
Let
\begin{align*}
  C_{n+1} & = (\mem_n, \reg_n, \pref i {\buf_n}[i \mapsto \pc \gets (\addr_0, \untainted, \tcap) \specat{\varepsilon}], \stdUpdateMu{\mu, \addr_0}) \\
  C_{n+1}' & = (\mem_n', \reg_n', \pref i {\buf_n'}[i \mapsto \pc \gets (\addr_0, \untainted, \tcap) \specat{\varepsilon}], \stdUpdateMu{\mu', \addr_0'}).
\end{align*}
By \ref{hw:execute-jmp-hazard}, $C_n' \sstep{}{} C_{n+1}'$.
By \Cref{def:pub-eq}, $C_{n+1} \pubeq C_{n+1}'$.
\end{itemize}

\pgheading{Subcase $\buf_n(i) = \load{x}{e}{\sz}$, $x\neq \pc$}
Since $\lowproj{\buf_n} = \lowproj{\buf_n'}$, then $\buf_n'(i) = \load{x}{e}{\sz} = \buf_n(i)$.

By \ref{hw:execute-load-mem} and \ref{hw:execute-load-fwd}, there should be
\begin{align*}
  & ((\addr, (\perm, \rbase, \rend, t_c)), t_0, \tcap) = \sem{e}_{\aplsan(\pref i {\buf_n}, \reg_n)}\\
  & \stdCheckCap{(\addr, (\perm, \rbase, \rend, t_c)), \sz, \readonly} \qquad x \neq \pc\\
\end{align*}

Denote $\sem{e}_{\aplsan(\pref i {\buf_n'}, \reg_n')} = (\stdCap{\addr'}, \_, c')$.
We aim to prove that $\stdCap{\addr'} = (\addr, (\perm, \rbase, \rend, t_c))$ and $c' = \tcap$ by discussing the taint status of the result.
\begin{itemize}
  \item $t_0 = \untainted$: By \Cref{lemma:expr-aplsan-eq}, there should be $\sem{e}_{\aplsan(\pref i {\buf_n}, \reg_n)} = \sem{e}_{\aplsan(\pref i {\buf_n'}, \reg_n')}$, so the statement holds.
  \item $t_0 = \tainted$: By the definition of $\aplsan$, this is only possible if $\dom{\pref i {\buf_n}} = \emptyset$. Then, there should be $\aplsan(\pref i {\buf_n}, \reg_n) = \reg_n$.
  Since $C_n \pubeq C_n'$, then there is also $\dom{\pref i {\buf_n'}} = \emptyset$, thereby $\aplsan(\pref i {\buf_n'}, \reg_n') = \reg_n'$.

  Note that the premises of \ref{isa:isa-load} are satisfied:
  since $C_n$ is well-formed, there should be
  $\load{x}{e}{\sz} = \stdReadMemInst{\mem_n, \sem{\pc}_{\reg_n}}$; furthermore,
  $\aplsan(\pref i {\buf_n}, \reg_n) = \reg_n$ implies that
  $\sem{e}_{\reg_n} = ((\addr, (\perm, \rbase, \rend, t_c)), t_0, \tcap)$,
  so $\stdCheckCap{\sem{e}_{\reg_n}, \sz, \readonly}$ holds and, by \ref{memory:read-mem},
  $\stdReadMem{\mem_n, \sem{e}_{\reg_n}, \sz}$ is defined.
  Hence, by \ref{isa:isa-load}, there should exist $(\mem_{n+1}, \reg_{n+1})$ such that
  \begin{equation*}
    (\mem_n, \reg_n) \step{\stdLoad{(\addr, (\perm, \rbase, \rend, t_c))}} (\mem_{n+1}, \reg_{n+1}).
  \end{equation*}
  Since $S_0 \equiv_{\mathit{ISA}} S_0'$, then there should also exist $(\mem_{n+1}', \reg_{n+1}')$ such that
  \begin{align*}
    & (\mem_n', \reg_n') \step{\stdLoad{\stdCap{\addr'}}} (\mem_{n+1}', \reg_{n+1}') \\
    & \stdLoad{(\addr, (\perm, \rbase, \rend, t_c))} = \stdLoad{\stdCap{\addr'}}
  \end{align*}
  Thus, $\stdCap{\addr'} = (\addr, (\perm, \rbase, \rend, t_c))$.
  $c' = \tcap$ can also be implied from $(\mem_n', \reg_n') \step{\stdLoad{\stdCap{\addr'}}} (\mem_{n+1}', \reg_{n+1}')$ by \ref{isa:isa-load} and \ref{memory:read-mem}.
\end{itemize}

Next, we analyze the remaining part of rules to determine whether the load data comes from memory or from a previous store.
The key insight is to leverage $\lowproj{\buf_n} = \lowproj{\buf_n'}$.
Specifically, for all $j$, if $\buf_n(j) = \store{v_0}{\stdCap{\addr'}}{\sz}\specat{\tg}$, then $\lowproj{\buf_n} = \lowproj{\buf_n'}$ implies that $\buf_n'(j) = \store{v_0'}{\stdCap{\addr'}}{\sz}\specat{\tg}$ where $\lowproj{v_0} = \lowproj{v_0'}$.
Consider the following two cases:
\begin{itemize}
\item \ref{hw:execute-load-mem} is applied, i.e.,
\begin{multline*}
  \forall j < i,\; \buf_n(j) = \store{x'}{\stdCap{\addr'}}{\sz'}\specat{\tg}\\
    \Rightarrow [\addr, \addr + \sz) \cap [\addr', \addr' + \sz') = \emptyset.
\end{multline*}
According to the above analysis, the statement still holds if we replace $\buf_n$ with $\buf_n'$.

Let
\begin{align*}
  \mem_n & = (\memData, \memTag) &
  v_0 & = \memData[\addr, \addr + \sz] &
  c_0 & = \memTag[\addr / \szCap]\\
  \mem_n' & = (\memData', \memTag') &
  v_0' & = \memData'[\addr, \addr + \sz] &
  c_0' & = \memTag'[\addr / \szCap].
\end{align*}
Recall that $\stdCheckCap{((\addr, (\perm, \rbase, \rend, \taint_c)), \sz, \readonly)}$,
so $\sz = \szCap \Rightarrow \addr \% \szCap = 0$.
By \Cref{lemma:exp-cap-mono}, there should be 
\begin{align*}
  & [\rbase, \rend) \subseteq \taintAddr{\reg_n, \mem_n, t_c} \\
  & [\rbase, \rend) \subseteq \taintAddr{\reg_n', \mem_n', t_c}.
\end{align*}
By \ref{memory:check-cap}, there should be $[\addr, \addr + \sz) \subseteq [\rbase, \rend)$.
Thus,
\begin{equation*}
  [\addr, \addr+ \sz) \subseteq \taintAddr{\reg_n, \mem_n, t_c} \cap \taintAddr{\reg_n', \mem_n', t_c}.
\end{equation*}
Note that $C_n \pubeq C_n'$ implies that $\cf{\mem_n, \reg_n} \pubeq \cf{\mem_n', \reg_n'}$.
Thus, by \Cref{def:pub-eq-isa}, the above conditions imply that
\begin{align*}
\lowproj{\stdReadRes{v_0, t_c, c_0, \sz}} = \lowproj{\stdReadRes{v_0', t_c, c_0', \sz}}.
\end{align*}
Let
\begin{align*}
  \buf_{n+1} & = \buf_n[i \mapsto (x \gets \lowproj{\stdReadRes{v_0, t_c, c_0, \sz}}) \\
  & \qquad \qquad \qquad \specat{((\addr, (\perm, \rbase, \rend, t_c)), \sz, \bot)}]\\
  C_{n+1} & = (\mem_n, \reg_n, \buf_{n+1}, \stdUpdateMu{\mu, (\addr, (\perm, \rbase, \rend, t_c))})\\
  \buf_{n+1}' & = \buf_n'[i \mapsto (x \gets \lowproj{\stdReadRes{v_0', t_c, c_0', \sz}}) \\
  & \qquad \qquad \qquad \specat{((\addr, (\perm, \rbase, \rend, t_c)), \sz, \bot)}]\\
  C_{n+1}' & = (\mem_n', \reg_n', \buf_{n+1}', \stdUpdateMu{\mu', (\addr, (\perm, \rbase, \rend, t_c))}).
\end{align*}
By \ref{hw:execute-load-mem}, $C_n \sstep{}{} C_{n+1}$.
By \Cref{def:pub-eq}, $C_{n+1} \pubeq C_{n+1}'$.

\item \ref{hw:execute-load-fwd} is applied, i.e.,
\begin{align*}
  & j = \max \lbrace j < i: \buf_n(j) = \store{\_}{\stdCap{\addr''}}{\sz''}\specat{\tg} \\
  & \qquad \qquad \qquad \qquad \qquad \land [\addr, \addr + \sz) \cap [\addr'', \addr'' + \sz'') \neq \emptyset \rbrace \\
  & \buf_n(j) = \store{(v_0, t_0, c_0)}{\stdCap{\addr'}}{\sz'} \specat{\tg} \\
  & [\addr, \addr + \sz) = [\addr', \addr' + \sz')
  \qquad t_0 \sqsubseteq t_c.
\end{align*}

By the above analysis based on $\lowproj{\buf_n} = \lowproj{\buf_n'}$, similar statements hold when replacing $\buf_n$ with $\buf_n'$, i.e.,
\begin{align*}
  & j = \max \lbrace j < i: \buf_n'(j) = \store{\_}{\stdCap{\addr''}}{\sz''}\specat{\tg} \\
  & \qquad \qquad \qquad \qquad \qquad \land [\addr, \addr + \sz) \cap [\addr'', \addr'' + \sz'') \neq \emptyset \rbrace \\
  & \buf_n'(j) = \store{(v_0', t_0', c_0')}{\stdCap{\addr'}}{\sz} \specat{\tg} \\
  & t_0' \sqsubseteq t_c.
\end{align*}
Furthermore, since $C_n$ and $C_n'$ are well-formed and $C_n \pubeq C_n'$, there should be
\begin{align*}
  & \stdCheckCap{\stdCap{\addr'}, \sz', \readwrite} \\
  & \stdCapMeta{\addr'} = (\_, \_, \_, t_\text{store}) \\
  & \store{x}{e}{\sz} = \stdReadMemInst{\mem_n, \sem{\pc}_{\apl(\pref j {\buf_n}, \reg_n)}} \\
  & \qquad \qquad \qquad = \stdReadMemInst{\mem_n', \sem{\pc}_{\apl(\pref j {\buf_n'}, \reg_n')}} \\
  & (v_1, t_0, c_1) = \sem{x}_{\apl(\pref j {\buf_n}, \reg_n)} \\
  & (v_0, c_0) = \stdWriteRes{(v_1, t_0, c_1), t_\text{store}, \sz} \\
  & (v_1', t_0', c_1') = \sem{x}_{\apl(\pref j {\buf_n'}, \reg_n')} \\
  & (v_0', c_0') = \stdWriteRes{(v_1', t_0', c_1'), t_\text{store}, \sz} .
\end{align*}
Note that the load and the store capabilities are aliasing and both pass $\stdCheckCap{\cdot}$.
Following the proof idea for \Cref{lemma:addr-same-taint}, we can also derive that $\taint_c = \taint_\text{store}$.

Furthermore, by \Cref{lemma:exp-pub-eq}, $\lowproj{(v_1, t_0, c_1)} = \lowproj{(v_1', t_0', c_1')}$.
Then, $t_0 \sqsubseteq t_c$ implies that $\lowproj{(v_1, t_c, c_1)} = \lowproj{(v_1', t_c, c_1')}$.
By \Cref{lemma:read-write-lowproj}, there should be
\begin{equation*}
  \lowproj{\stdReadRes{v_0, t_c, c_0, \sz}} = \lowproj{\stdReadRes{v_0', t_c, c_0', \sz}}.
\end{equation*}

Let
\begin{align*}
  \buf_{n+1} & = \buf_n[i \mapsto (x \gets \lowproj{\stdReadRes{v_0, t_c, c_0, \sz}}) \\
  & \qquad \qquad \qquad \specat{((\addr, (\perm, \rbase, \rend, t_c)), \sz, \bot)}]\\
  C_{n+1} & = (\mem_n, \reg_n, \buf_{n+1}, \stdUpdateMu{\mu, (\addr, (\perm, \rbase, \rend, t_c))})\\
  \buf_{n+1}' & = \buf_n'[i \mapsto (x \gets \lowproj{\stdReadRes{v_0', t_c, c_0', \sz}}) \\
  & \qquad \qquad \qquad \specat{((\addr, (\perm, \rbase, \rend, t_c)), \sz, \bot)}]\\
  C_{n+1}' & = (\mem_n', \reg_n', \buf_{n+1}', \stdUpdateMu{\mu', (\addr, (\perm, \rbase, \rend, t_c))}).
\end{align*}
By \ref{hw:execute-load-mem}, $C_n \sstep{}{} C_{n+1}$.
By \Cref{def:pub-eq}, $C_{n+1} \pubeq C_{n+1}'$.

\end{itemize}

\pgheading{Subcase $\buf_n(i) = \store{x}{e}{\sz}$}
Since $\lowproj{\buf_n} = \lowproj{\buf_n'}$, then $\buf_n'(i) = \store{x}{e}{\sz} = \buf_n(i)$.

By \ref{hw:execute-store-ok} and \ref{hw:execute-store-hazard}, there should be
\begin{align*}
  & ((\addr, (\perm, \rbase, \rend, t_c)), t_0, \tcap) = \sem{e}_{\aplsan(\pref i {\buf_n}, \reg_n)}\\
  & \stdCheckCap{(\addr, (\perm, \rbase, \rend, t_c)), \sz, \readwrite} \\
  & (v, t, c) = \sem{x}_{\apl (\pref i {\buf_n}, \reg_n)} \\
  & (v_0, c_0) = \stdWriteRes{(v, t, c), t_c, \sz} \\
\end{align*}

By \Cref{lemma:exp-pub-eq}, there should exist $(v', t', c')$ such that $(v', t', c') = \sem{x}_{\apl (\pref i {\buf_n'}, \reg_n')}$ and $\lowproj{(v, t, c)} = \lowproj{(v', t', c')}$.

Denote $\sem{e}_{\aplsan(\pref i {\buf_n'}, \reg_n')} = (\stdCap{\addr''}, \_, c'')$.
We aim to prove that $\stdCap{\addr''} = (\addr, (\perm, \rbase, \rend, t_c))$ and $c'' = \tcap$ by discussing the taint status of the result.
\begin{itemize}
  \item $t_0 = \untainted$: By \Cref{lemma:expr-aplsan-eq}, there should be $\sem{e}_{\aplsan(\pref i {\buf_n}, \reg_n)} = \sem{e}_{\aplsan(\pref i {\buf_n'}, \reg_n')}$, so the statement holds.
  \item $t_0 = \tainted$: By the definition of $\aplsan$, this is only possible if $\dom{\pref i {\buf_n}} = \emptyset$. Then, there should be $\aplsan(\pref i {\buf_n}, \reg_n) = \reg_n$.
  Since $C_n \pubeq C_n'$, then there is also $\dom{\pref i {\buf_n'}} = \emptyset$, thereby $\aplsan(\pref i {\buf_n'}, \reg_n') = \reg_n'$.

  Note that the premises of \ref{isa:isa-store} are satisfied:
  since $C_n$ is well-formed, there should be
  $\store{x}{e}{\sz} = \stdReadMemInst{\mem_n, \sem{\pc}_{\reg_n}}$; furthermore,
  $\aplsan(\pref i {\buf_n}, \reg_n) = \reg_n$ implies that
  $\sem{e}_{\reg_n} = ((\addr, (\perm, \rbase, \rend, t_c)), t_0, \tcap)$,
  so $\stdCheckCap{\sem{e}_{\reg_n}, \sz, \readwrite}$ holds and, by \ref{memory:write-mem},
  $\stdWriteMem{\mem_n, \sem{e}_{\reg_n}, \sz, \sem{x}_{\reg_n}}$ is defined.
  Hence, by \ref{isa:isa-store}, there should exist $v_o$ and $(\mem_{n+1}, \reg_{n+1})$ such that
  \begin{equation*}
    (\mem_n, \reg_n) \step{\stdStore{v_o, (\addr, (\perm, \rbase, \rend, t_c))}} (\mem_{n+1}, \reg_{n+1}).
  \end{equation*}
  Since $S_0 \equiv_{\mathit{ISA}} S_0'$, then there should also exist $v_o'$ and $(\mem_{n+1}', \reg_{n+1}')$ such that
  \begin{align*}
    & (\mem_n', \reg_n') \step{\stdStore{v_o', \stdCap{\addr''}}} (\mem_{n+1}', \reg_{n+1}') \\
    & \stdStore{v_o, (\addr, (\perm, \rbase, \rend, t_c))} = \stdStore{v_o', \stdCap{\addr''}}.
  \end{align*}
  Thus, $\stdCap{\addr''} = (\addr, (\perm, \rbase, \rend, t_c))$.
  $c'' = \tcap$ can also be implied from $(\mem_n', \reg_n') \step{\stdStore{v_o', \stdCap{\addr''}}} (\mem_{n+1}', \reg_{n+1}')$ by \ref{isa:isa-store} and \ref{memory:write-mem}.
\end{itemize}

Let $(v_0', c_0') = \stdWriteRes{(v', t', c'), t_c, \sz}$.
By the definition of $\stdWriteRes{}$, there should also be $\lowproj{(v_0, t, c_0)} = \lowproj{(v_0', t', c_0')}$.

Next, we analyze the load/store aliasing checking in \ref{hw:execute-store-ok} and \ref{hw:execute-store-hazard}.
The key insight is to leverage $\lowproj{\buf_n} = \lowproj{\buf_n'}$.
Specifically, for all $j$, if 
\begin{equation*}
  \buf_n(j) = x \gets v \specat{(\stdCap{\addr'}, \sz', k)},
\end{equation*}
then $\lowproj{\buf_n} = \lowproj{\buf_n'}$ implies that 
\begin{align*}
  & \buf_n'(j) = x \gets v' \specat{(\stdCap{\addr'}, \sz', k)}\\
  & \lowproj{v} = \lowproj{v'}.
\end{align*}
Consider the following two cases:
\begin{itemize}
\item \ref{hw:execute-store-ok} is applied, i.e.,
\begin{multline*}
  \forall j > i, \buf_n(j) = x \gets v \specat{(\stdCap{\addr'}, \sz', k)} \land k \neq i \\
  \Rightarrow
  (k > i \lor [\addr, \addr + \sz) \cap [\addr', \addr + \sz') = \emptyset).
\end{multline*}
According to the above analysis, the statement still hold if we replace $\buf_n$ with $\buf_n'$.
Let
\begin{align*}
  \buf_{n+1} & = \buf_n[i \mapsto \store{(v_0, t, c_0)}{(\addr, (\perm, \rbase, \rend, t_c))}{\sz}]\\
  C_{n+1} & = \cf{\mem_n, \reg_n, \buf_{n+1}, \stdUpdateMu{\mu, (\addr, (\perm, \rbase, \rend, t_c))}}\\
  \buf_{n+1}' & = \buf_n'[i \mapsto \store{(v_0', t', c_0')}{(\addr, (\perm, \rbase, \rend, t_c))}{\sz}]\\
  C_{n+1}' & = \cf{\mem_n', \reg_n', \buf_{n+1}', \stdUpdateMu{\mu', (\addr, (\perm, \rbase, \rend, t_c))}}.
\end{align*}
By \ref{hw:execute-store-ok}, $C_n' \sstep{}{} C_{n+1}'$.
By \Cref{def:pub-eq}, $C_{n+1} \pubeq C_{n+1}'$.

\item \ref{hw:execute-store-hazard} is applied, i.e.,
\begin{multline*}
  j = \min \lbrace j > i: \buf_n(j) = x\gets v \specat{(\stdCap{\addr'}, \sz', k)} \\
  \land (
    k < i \land [\addr, \addr+\sz) \cap [\addr', \addr'+\sz') \neq \emptyset
  ) \rbrace
\end{multline*}
According to the above analysis, the right-hand side does not change if we replace $\buf_n$ with $\buf_n'$.
Let
\begin{align*}
  \buf_{n+1} & = \pref i {\buf_n}[i \mapsto \store{(v_0, t, c_0)}{(\addr, (\perm, \rbase, \rend, t_c))}{\sz}]\\
  C_{n+1} & = \cf{\mem_n, \reg_n, \buf_{n+1}, \stdUpdateMu{\mu, (\addr, (\perm, \rbase, \rend, t_c))}}\\
  \buf_{n+1}' & = \pref i {\buf_n'}[i \mapsto \store{(v_0', t', c_0')}{(\addr, (\perm, \rbase, \rend, t_c))}{\sz}]\\
  C_{n+1}' & = \cf{\mem_n', \reg_n', \buf_{n+1}', \stdUpdateMu{\mu', (\addr, (\perm, \rbase, \rend, t_c))}}.
\end{align*}
By \ref{hw:execute-store-hazard}, $C_n' \sstep{}{} C_{n+1}'$.
By \Cref{def:pub-eq}, $C_{n+1} \pubeq C_{n+1}'$.
\end{itemize}

\pgheading{Case $d = \dcommit$}
Let $i = \min \dom{\buf_n} = \min \dom{\buf_n'}$ (the equivalence is implied from $C_n \pubeq C_n'$).
Consider the following two cases:
\begin{itemize}
\item $\buf_n(i) = x \gets v \specat{\tg}$:
Since $C_n \pubeq C_n'$, then $\lowproj{\buf_n} = \lowproj{\buf_n'}$.
This implies that there must exist $v'$, $\tg'$ such that $\buf_n'(i) = x \gets v' \specat{\tg'}$ and $\lowproj{v} = \lowproj{v'}$.
Thus, there should be
\begin{align*}
  \mem_{n+1} & = \mem_n & \reg_{n+1} & = \reg_n[x \mapsto v] &
  C_{n+1} & = \cf{\mem_{n+1}, \reg_{n+1}, \buf_n \backslash i, \mu} \\
  \mem_{n+1}' & = \mem_n' & \reg_{n+1}' & = \reg_n'[x \mapsto v'] &
  C_{n+1}' & = \cf{\mem_{n+1}', \reg_{n+1}', \buf_n' \backslash i, \mu'},
\end{align*}
and $C_n \sstep{}{} C_{n+1}$, $C_n' \sstep{}{} C_{n+1}'$.
\item $\buf_n(i) = \store{(v, t, c)}{\stdCap{\addr}}{\sz}\specat{\varepsilon}$:
Since $C_n \pubeq C_n'$, then $\lowproj{\buf_n} = \lowproj{\buf_n'}$.
By the definition of untainted projection for reorder buffer and store instruction,
there must be $\buf_n'(i) = \store{(v', t', c')}{\stdCap{\addr}}{\sz}\specat{\varepsilon}$ for some $(v', t', c')$ where $\lowproj{(v, t, c)} = \lowproj{(v', t', c')}$.
Let $(\memData, \memTag) = \mem_n$, and $(\memData', \memTag') = \mem_n'$.
Thus, there should be
\begin{align*}
  \mem_{n+1} & = (\memData[[\addr, \addr + \sz) \mapsto v], \memTag[\addr / 2 \mapsto c]) \\
  \reg_{n+1} & = \reg_n \qquad
  C_{n+1} = (\mem_{n+1}, \reg_{n+1}, \buf_n \backslash i, \stdUpdateMu{\mu, \stdCap{\addr}}) \\
  \mem_{n+1}' & = (\memData'[[\addr, \addr + \sz) \mapsto v'], \memTag[\addr / 2 \mapsto c']) \\
  \reg_{n+1}' & = \reg_n' \qquad
  C_{n+1}' = \cf{\mem_{n+1}', \reg_{n+1}', \buf_n' \backslash i, \stdUpdateMu{\mu', \stdCap{\addr}}},
\end{align*}
and $C_n \sstep{}{} C_{n+1}$, $C_n' \sstep{}{} C_{n+1}'$.
\end{itemize}

By \Cref{lemma:functional-correctness}, for both cases, there should be
\begin{align*}
  S_{f(\qty{C_0, C_1, \dots, C_{n}})} & = \cf{\mem_{n+1}, \reg_{n+1}} \\
  S_{f(\qty{C_0', C_1', \dots, C_{n}'})}' & = \cf{\mem_{n+1}', \reg_{n+1}'} \\
\end{align*}
where $f$ is defined in \Cref{lemma:functional-correctness}.
By the definition of $f$, since for all $k\leq n$, $C_k \pubeq C_k'$, 
then $f(\qty{C_0, C_1, \dots, C_{n}}) = f(\qty{C_0', C_1', \dots, C_{n}'})$.
Furthermore, by \Cref{lemma:isa-pubeq}, $S_{f(\qty{C_0, C_1, \dots, C_{n}})} \pubeq S_{f(\qty{C_0, C_1, \dots, C_{n}})}'$.
Thus, $\cf{\mem_{n+1}, \reg_{n+1}} \pubeq \cf{\mem_{n+1}', \reg_{n+1}'}$.

Furthermore, note that $\lowproj{\buf_n} = \lowproj{\buf_n'}$ and $\mu = \mu'$. Therefore, for both cases, there should be $C_{n+1} \pubeq C_{n+1}'$.
\end{proof}

\fi

\end{document}